\documentclass[11pt]{amsart}
\usepackage{amsmath,amssymb,graphicx}
\usepackage[english]{babel}
\usepackage[T1]{fontenc}
\usepackage{amsmath, amscd,amsthm, amsfonts, amssymb}
\usepackage{epsfig}
\usepackage{hhline, latexsym}
\usepackage{float}

\newlength{\dinwidth}
\newlength{\dinmargin}
\newcommand{\M}{\mathcal{M}}
\newcommand{\R}{\mathbb{R}}

\newcommand{\C}{\mathbb{C}}
\newcommand{\Z}{\mathbb{Z}}
\newcommand{\B}{\mathbb{B}}
\newcommand{\z}{\mathbf{z}}

\newcommand{\ud}{\,\mathrm{d}}
\newcommand{\Rs}{\mathcal{R}}

\newtheorem{proposition}{Proposition}[section]
\newtheorem{corollary}{Corollary}[section]
\newtheorem{remark}{Remark}[section]
\newtheorem{lemma}{Lemma}[section]
\newtheorem{example}{Example}[section]

\begin{document}

\def\theequation {\thesection.\arabic{equation}}
\makeatletter\@addtoreset {equation}{section}\makeatother

\title[]{On the numerical evaluation of\\ algebro-geometric solutions \\ to integrable  equations} 

\author{C.~Kalla}
\address{Institut de Math\'ematiques de Bourgogne,
		Universit\'e de Bourgogne, 9 avenue Alain Savary, 21078 Dijon
		Cedex, France}
    \email{Caroline.Kalla@u-bourgogne.fr}

\author{C.~Klein}
\address{Institut de Math\'ematiques de Bourgogne,
		Universit\'e de Bourgogne, 9 avenue Alain Savary, 21078 Dijon
		Cedex, France}
    \email{Christian.Klein@u-bourgogne.fr}

\date{\today}    

\begin{abstract}

Physically meaningful periodic solutions to certain integrable partial differential 
equations are given in terms of multi-dimensional theta functions 
associated to real Riemann surfaces. Typical analytical problems in the numerical 
evaluation of these solutions are studied. In the case of 
hyperelliptic surfaces efficient algorithms exist even for almost 
degenerate surfaces. This allows the numerical study of solitonic 
limits. For general real Riemann surfaces, the choice of a homology 
basis adapted to the anti-holomorphic involution is important for a 
convenient formulation of the solutions and smoothness conditions. Since 
existing algorithms for algebraic curves produce a homology basis not 
related to automorphisms of the curve, we study symplectic 
transformations to an adapted basis and give explicit formulae for M-curves. As 
examples we discuss solutions of the Davey-Stewartson and the multi-component 
nonlinear Schrödinger equations.
\end{abstract}


\thanks{We thank D.~Korotkin and V.~Shramchenko for useful discussions and hints. 
This work has been supported in part by the project FroM-PDE funded by the European
Research Council through the Advanced Investigator Grant Scheme, the Conseil R\'egional de Bourgogne
via a FABER grant and the ANR via the program ANR-09-BLAN-0117-01. }

\maketitle

\section{Introduction}

The importance of Riemann surfaces for the construction of almost 
periodic solutions to various integrable partial differential 
equations (PDEs) was realized at the beginning of the 1970s by Novikov, 
Dubrovin and Its, Matveev. The latter found the Its-Matveev 
formula for the Korteweg-de Vries (KdV) equation in terms of 
multi-dimensional theta functions on hyperelliptic Riemann surfaces. 
Similar formulae were later obtained for other integrable PDEs as nonlinear 
Schr\"odinger  (NLS) and sine-Gordon equations. For the history of the topic the 
reader is referred to the reviews \cite{BBEIM} 
and \cite{Dintro}.
Krichever \cite{K} showed that theta-functional solutions to the 
Kadomtsev-Petviashvili equation can be obtained on arbitrary 
Riemann  surfaces. The problems of real-valuedness and smoothness of 
these solutions were solved by Dubrovin and Natanzon in  \cite{DN}.

Novikov criticized the practical relevance of theta functions 
since no numerical algorithms existed at the time to actually compute 
the found solutions. He suggested an effective treatment of theta 
functions (see, for instance,  \cite{Dintro}) by a suitable 
parametrization of the characteristic quantities of a Riemann 
surface, i.e., the periods of holomorphic and certain meromorphic 
differentials on the given surface. This program is limited to genera 
smaller than 4 since so-called Schottky relations exist for higher 
genus between the components of the period matrix of a Riemann 
surface. The task to find such relations is known as the Schottky 
problem. This led to the 
famous Novikov conjecture for the Schottky problem that a Riemann 
matrix (a symmetric matrix with negative definite real part) is the 
matrix of $\mathcal{B}$-periods of the normalized holomorphic differentials of 
a Riemann surface if and only if  Krichever's formula with this 
matrix yields a solution to the KP equation. The conjecture was 
finally proven by Shiota \cite{Shiota}.

First plots of KP solutions appeared in \cite{Mum} and via 
Schottky uniformizations in \cite{bobo}. Since all compact Riemann 
surfaces can be defined via non-singular plane algebraic curves of the form
\begin{equation}
    F(x,y) := \sum_{n=1}^{N}\sum_{m=1}^{M}a_{mn}x^{m}y^{n}=0, \quad     x,y\in \mathbb{C},
    \label{algcur}
\end{equation}
with constant complex coefficients $a_{nm}$, Deconinck and van Hoeij 
developed an approach to the symbolic-numerical treatment of 
algebraic curves. This approach is distributed as the 
\textit{algcurves} package with Maple, see \cite{deco1,deco2,deho}. A purely numerical approach to real 
hyperelliptic Riemann surfaces was given in \cite{FK1,FK2}, and for general 
Riemann surfaces in \cite{FK}. For a review on computational 
approaches to Riemann surfaces the reader is referred to \cite{Bob}.

In this paper we want to address typical analytical problems appearing in the 
numerical study of theta-functional solutions to integrable PDEs, and 
present the state of the art of the field by considering 
concrete examples. The 
case of hyperelliptic Riemann surfaces ($N=2$ in (\ref{algcur})) 
is the most accessible, since equation (\ref{algcur}) can be solved 
explicitly for $y$, and  
since a basis for differentials and homology can be given 
a priori.  Families of hyperelliptic curves can be conveniently parametrized by 
their branch points. The codes \cite{FK1,FK2} are 
able to treat 
effectively numerically collisions of branch points, a limit in 
which certain periods of the  corresponding hyperelliptic surface diverge. If the 
limiting  Riemann surface has genus $0$, the theta series breaks down to a 
finite sum which gives for an appropriate choice of the 
characteristic well 
known solitonic solutions to the studied equation.

For solutions 
defined on general real
algebraic curves, i.e., curves (\ref{algcur}) with all $a_{nm}$ real, 
an important point in applications are reality and smoothness 
conditions. 
These are conveniently formulated for a homology basis 
for which the $\mathcal{A}$-cycles are invariant under the action of 
the anti-holomorphic involution. However, the existing 
algorithms for the computational treatment of algebraic curves 
produce a basis of the homology that is in general not 
related to possible automorphisms of the curve. 
To implement the reality and smoothness requirements, a 
transformation to the basis for which the conditions are formulated 
has to be constructed. We study the necessary symplectic 
transformations and give explicit relations for so-called M-curves, 
curves with the maximal number of real ovals.

To illustrate these concepts, we study for the first time numerically 
theta-functional solutions to integrable equations from 
the family of NLS equations, namely,
the multi-component nonlinear Schrödinger equation
\begin{equation}		
\mathrm{i}\,\frac{\partial \psi_{j}}{\partial t}+\frac{\partial^{2} \psi_{j}}{\partial x^{2}}+2\left(\sum_{k=1}^{n}s_{k}|\psi_{k}|^{2}\right)\psi_{j} =0,  \qquad  j=1,\ldots,n,   \label{n-NLS}
\end{equation}
denoted by $n$-NLS$^{s}$,  where 
$s=(s_{1},\ldots,s_{n})$,  $s_{k}=\pm 1$, 
and the $(2+1)$-dimensional 
 Davey-Stewartson (DS) equations, 
\begin{align}		
\mathrm{i}\,\psi_{t}+\psi_{xx}-\alpha^{2} \,\psi_{yy}+2\,(\Phi+\rho\,|\psi|^{2})\,\psi &=0,     \nonumber  \\	
\Phi_{xx}+\alpha^{2} \,\Phi_{yy}+2\rho\,|\psi|^{2}_{xx} & =0,  \label{DSintro} 
\end{align}
where $\alpha=\mathrm{i}$ or $\alpha=1$ and where $\rho=\pm1$.
Both equations (\ref{n-NLS}) and (\ref{DSintro}) reduce to the NLS equation under certain conditions: 
the former obviously in the case $n=1$, the latter  if $\psi$ is independent 
of the variable $y$ and satisfies certain boundary conditions,  for instance that 
$\Phi+\rho\,|\psi|^{2}$ tends to zero when $x$ tends to 
infinity.

Integrability of the NLS 
equation  was shown by Zakharov and Shabat 
\cite{ZSh} and algebro-geometric solutions were 
given by Its \cite{Its}. 
The multi-component nonlinear 
Schrödinger equation (\ref{n-NLS}) in the case $n=2, \,s=(1,1),$ is called the vector NLS or 
Manakov system. Manakov \cite{Man} first examined this
equation as an asymptotic model for the propagation of the electric
field in a waveguide. Its integrability was shown for $n=2$ by 
Zakharov and Schulman in \cite{ZS} and  for the general case in 
\cite{RSL}. Algebro-geometric solutions to the 2-NLS$^{s}$ equation with $s=(1,1)$ were 
presented in \cite{EEI}, and for the general case in \cite{Kalla}.
The DS equation (\ref{DSintro}) was introduced in \cite{DS} to describe 
the evolution of a three-dimensional wave packet on water of finite 
depth. Its  integrability was shown in \cite{AF}, and solutions in terms of 
multi-dimensional theta functions on general Riemann surfaces were 
given in \cite{Mal, Kalla}.

To ensure the correct numerical implementation of the formulae of \cite{Kalla}, we 
check for each point in the spacetime whether 
certain identities for theta functions are satisfied. Since these 
identities are not used in the code, they provide a strong test for 
the computed quantities.  Numerically the identities are never exactly satisfied, but 
to high precision. The code reports a warning if the residual of the 
test relations is larger than $10^{-6}$ which is well below plotting 
accuracy. Typically the conditions are satisfied to machine 
precision\footnote{We work with double precision, i.e., a precision 
of $10^{-16}$; due to rounding errors this is typically reduced to 
$10^{-12}$ to $10^{-14}$.}.  In addition we compute the solutions on 
a numerical grid and numerically differentiate them. We check in this way for low genus that the solutions to $n$-NLS$^{s}$ and 
DS in terms of multi-dimensional theta functions satisfy the 
respective equations to better than $10^{-6}$. These two completely independent 
tests ensure that the presented plots are showing the correct solutions to better than 
plotting accuracy.

The paper is organized as follows: in Section 2 we recall some facts 
from the theory of multi-dimensional theta functions and the theory of real 
Riemann surfaces, necessary to give theta-functional solutions to the 
$n$-NLS$^{s}$ and DS equations. In Section 3 we consider the hyperelliptic 
case and study concrete examples of low genus,  also in almost 
degenerate situations. In Section 4 we consider examples of 
non-hyperelliptic real Riemann surfaces and discuss symplectic 
transformations needed to obtain smooth solutions.  We add some 
concluding remarks in Section 5.

\section{Theta functions and real Riemann surfaces}

In this section we recall basic facts on Riemann surfaces, in 
particular real surfaces, and multi-dimensional theta functions 
defined on them. Solutions to the $n$-NLS$^{s}$ and the DS
equations in terms of theta functions will be given following \cite{Kalla}.

\subsection{Theta functions}

Let $\Rs _{g}$ be a 
compact Riemann surface of genus $g>0$.  Denote by $(\mathcal{A},\mathcal{B}):=(\mathcal{A}_{1},\ldots,\mathcal{A}_{g},\mathcal{B}_{1},\ldots,\mathcal{B}_{g})$ a canonical homology basis, 
 and by $(\omega_{1},\ldots,\omega_{g})$ the  basis of holomorphic differentials normalized via
\begin{equation}
\int_{\mathcal{A}_{k}}\omega_{j}=2\mathrm{i}\pi\delta_{kj}, \qquad k,j=1,\ldots,g. \label{norm hol diff}
\end{equation}
The matrix $\B$ with entries $\mathbb{B}_{kj}=\int_{\mathcal{B}_{k}}\omega_{j}$ of $\mathcal{B}$-periods of the normalized holomorphic differentials  $\omega_{j}, \,j=1,\ldots,g$,
is symmetric and has a negative definite real part. The theta function with (half integer) characteristic $\delta=[\delta_1, \delta_2]$ is defined by
\begin{equation}
\Theta_{\B}[\delta](\z)=\sum_{\mathbf{m}\in\Z^{g}}\exp\left\{\tfrac{1}{2}\langle \B(\mathbf{m}+\delta_1),\mathbf{m}+\delta_1\rangle+\langle \mathbf{m}+\delta_1,\z+2\mathrm{i}\pi\delta_2\rangle\right\},\label{theta}
\end{equation}
for any $\z\in\C^{g}$; here $\delta_1,\delta_2\in \left\{0,\frac{1}{2}\right\}^{g}$ 
are the vectors of the characteristic $\delta$; $\langle.,.\rangle$ denotes the 
scalar product $\left\langle \mathbf{u},\mathbf{v} \right\rangle=\sum_{i}u_{i}\,v_{i}$ for any $\mathbf{u},\mathbf{v}\in\C^{g}$. The theta function $\Theta[\delta](\z)$ is even if 
the characteristic $\delta$ is even, i.e., if $4\left\langle 
\delta_1,\delta_2 \right\rangle$ is even, and odd if the characteristic 
$\delta$ is odd, i.e., if $4\left\langle \delta_1,\delta_2 
\right\rangle$ is odd. An even characteristic is called non-singular if 
$\Theta[\delta](0)\neq 0$, and an odd characteristic is called 
non-singular if the gradient $\nabla\Theta[\delta](0)$ is non-zero. The theta function with characteristic is related to 
the theta function with zero characteristic (the Riemann theta 
function denoted by $\Theta$) as 
follows
\begin{equation}
\Theta[\delta](\z)=\Theta(\z+2\mathrm{i}\pi\delta_2+\B\delta_1)\,\exp\left\{\tfrac{1}{2}\langle \B\delta_1,\delta_1\rangle+\langle\z+2\mathrm{i}\pi\delta_2,\delta_1\rangle\right\}.\label{2.3}
\end{equation}

Denote by $\Lambda$ the lattice $\Lambda=\{2\mathrm{i}\pi \mathbf{N}+\B 
\mathbf{M}, \,\,\mathbf{N},\mathbf{M}\in\Z^{g}\}$ generated by the 
$\mathcal{A}$ and $\mathcal{B}$-periods of the normalized holomorphic differentials $\omega_{j}, \,j=1,\ldots,g$. The complex torus $J(\Rs_{g})=\C^{g} / \Lambda$ is called the Jacobian of the Riemann surface $\Rs_{g}$. The theta function (\ref{theta}) has the following quasi-periodicity property  with respect to the lattice $\Lambda$:
\begin{equation}
\Theta[\delta](\mathbf{z}+2\mathrm{i}\pi \mathbf{N}+\B \mathbf{M}) \nonumber
\end{equation}
\begin{equation}
=\Theta[\delta](\z)  \exp\left\{-\tfrac{1}{2}\langle \B\mathbf{M},\mathbf{M}\rangle-\langle \z,\mathbf{M}\rangle+2\mathrm{i}\pi(\langle\delta_1,\mathbf{N}\rangle-\langle\delta_2,\mathbf{M}\rangle)\right\}.\label{2.4}
\end{equation}

For the formulation of solutions to  physically relevant integrable 
equations in terms of multi-dimensional theta functions, there is 
typically a preferred homology basis in which the solution takes a 
simple form. Let 
$(\mathcal{A},\mathcal{B})$ and 
$(\tilde{\mathcal{A}},\tilde{\mathcal{B}})$ be arbitrary canonical homology basis defined on $\Rs_{g}$, represented here
 by $2g$-dimensional vectors.  Under the change of homology basis 
\begin{equation}
\left(\begin{matrix}
A&B\\
C&D
\end{matrix}\right)
\left(\begin{matrix}
\mathbf{\tilde{\mathcal{A}}}\\
\mathbf{\tilde{\mathcal{B}}}
\end{matrix}\right) 
=
\left(\begin{matrix}
\mathbf{\mathcal{A}}\\
\mathbf{\mathcal{B}}
\end{matrix}\right),  \label{transf Vinn2}
\end{equation}
where $\left(\begin{matrix}
A&B\\
C&D
\end{matrix}\right)\in Sp(2g,\Z)$ is a symplectic matrix,
 the theta function (\ref{theta}) transforms as 
\begin{equation}
\Theta_{\mathbb{B}}[\delta](\z)=\,\kappa\,\sqrt{\det \tilde{\mathbb{K}}}\,\exp\left\{\tfrac{1}{2}\,\tilde{\z}^{t}\,(\tilde{\mathbb{K}}^{t})^{-1}B\,\tilde{\z}\right\}\,\Theta_{\tilde{\mathbb{B}}}[\tilde{\delta}](\tilde{\z}), \label{transf theta}
\end{equation}
where $\tilde{\mathbb{K}}=2\mathrm{i}\pi A+B\,\tilde{\mathbb{B}}$ and
\begin{align}
\mathbb{B}&=2\mathrm{i}\pi\,(2\mathrm{i}\pi\,C+D\,\tilde{\mathbb{B}})\,\tilde{\mathbb{K}}^{-1},\\
\tilde{\z}&=(2\mathrm{i}\pi)^{-1}\,\tilde{\mathbb{K}}^{t}\,\z,\\
\left(\begin{matrix}
\delta_{1}\\
\delta_{2}
\end{matrix}\right)
&=
\left(\begin{array}{rr}
A&-B\\
-C&D
\end{array}\right)\left(\begin{matrix}
\tilde{\delta}_{1}\\
\tilde{\delta}_{2}
\end{matrix}\right)+\frac{1}{2}\,\text{Diag}\left(\begin{matrix}
B\,A^{t}\\
D\,C^{t}
\end{matrix}\right), \label{transf caract}
\end{align}
for any $\z\in\C^{g}$, where $\mbox{Diag}(.)$ denotes the column vector of the 
diagonal entries of the matrix.
Here $\kappa$ is a constant independent of $\z$ and 
$\tilde{\mathbb{B}}$ (the exact value of $\kappa$ is not needed for our purposes).

The Abel map $\Rs_{g}\longrightarrow J(\Rs_{g})$ is defined by 
\begin{equation}
    \int_{p_{0}}^{p}:=\int_{p_{0}}^{p}\omega,  \label{abel}
\end{equation} 
for any $p\in\Rs_{g}$, where $p_{0}\in\Rs_{g}$ is the base point of 
the application, and where $\omega=(\omega_{1},\ldots,\omega_{g})^{t}$ is the vector of the normalized holomorphic differentials.

Now let $k_{a}$ denote a local parameter near $a\in\Rs_{g}$ and consider 
the following expansion of the normalized holomorphic differentials  $\omega_{j}, \, j=1,\ldots,g$,
\begin{equation} 
\omega_{j}(p)= \left(V_{a,j}+W_{a,j}\,k_{a}(p)+\ldots\right)\,\ud k_{a}(p), \label{exp hol diff}
\end{equation}
for any point $p\in\Rs_{g}$ lying in a neighbourhood of $a$, where $V_{a,j},\,W_{a,j}\in\C$.
Let us denote by $D_{a}$ (resp. $D'_{a}$) the operator of the directional derivative along the vector $\mathbf{V}_{a}=(V_{a,1},\ldots,V_{a,g})^{t}$ (resp.  $\mathbf{W}_{a}$).
According to \cite{Mum} and \cite{Kalla}, the theta function satisfies the following identities derived from Fay's identity \cite{Fay}:
\begin{equation}
D_{a}D_{b}\ln\Theta(\z)\,=\,q_{1}+q_{2}\,\frac{\Theta(\z+\int^{b}_{a})\,\Theta(\z-\int^{b}_{a})}{\Theta(\z)^{2}}\,,
\label{cor Fay2}
\end{equation}
\begin{equation}
D'_{a}\ln\frac{\Theta(\z+\int^{b}_{a})}{\Theta(\z)}+D_{a}^{2}\ln\frac{\Theta(\z+\int^{b}_{a})}{\Theta(\z)}+\Big(D_{a}\ln\frac{\Theta(\z+\int^{b}_{a})}{\Theta(\z)}-K_{1}\Big)^{2}+2D^{2}_{a}\ln\Theta(\z)+K_{2}=0,\label{my corol}
\end{equation}
for any $\z\in\C^{g}$ and any distinct points $a,b\in\Rs_{g}$; here the scalars $q_{i},K_{i}$, $i=1,2$ depend on the points $a,b$ and are given by
\begin{equation}
q_{1}(a,b)=D_{a}D_{b}\ln\Theta[\delta](\textstyle\int^{b}_{a}), \label{q1}
\end{equation}
\begin{equation}
q_{2}(a,b)=\frac{D_{a}\,\Theta[\delta](0)\,D_{b}\,\Theta[\delta](0)}{\Theta[\delta](\int^{b}_{a})^{2}},\label{q2}
\end{equation}
\begin{equation}
K_{1}(a,b)=\frac{1}{2}\,\frac{D_{a}'\,\Theta[\delta](0)}{D_{a}\,\Theta[\delta](0)}+D_{a}\ln\Theta[\delta](\textstyle\int^{b}_{a})\,, \label{K1}
\end{equation}
\begin{equation}
K_{2}(a,b)=-\,D'_{a}\ln\Theta(\textstyle\int^{b}_{a})-D_{a}^{2}\ln\left(\Theta(\textstyle\int^{b}_{a})\,\Theta(0)\right)-\Big(D_{a}\ln\Theta(\textstyle\int^{b}_{a})-K_{1}(a,b)\Big)^{2}, \label{K2}
\end{equation}
where $\delta$ is a non-singular odd characteristic.

\subsection{Real Riemann surfaces}

A Riemann surface $\mathcal{R}_{g}$ is called real  if it admits an 
anti-holomorphic involution 
$\tau:\mathcal{R}_{g}\rightarrow\mathcal{R}_{g}, \,\tau^{2}=\mbox{id}$. The 
connected components of the set of
fixed points of the anti-involution $\tau$ are called real ovals of $\tau$. We denote 
by $\mathcal{R}_{g}(\R)$ the set of fixed points. 
According to Harnack's inequality \cite{Harnack}, the number $k$ of real ovals of a real Riemann surface of genus $g$ cannot exceed 
$g+1$:   $0\leq k \leq g+1$. 
Curves with the maximal number $k=g+1$ of real ovals are called M-curves.

The complement $\mathcal{R}_{g}\setminus 
\mathcal{R}_{g}(\R)$ has either one or two connected components. The curve $\mathcal{R}_{g}$ is called a \textit{dividing} curve 
if $\mathcal{R}_{g}\setminus\mathcal{R}_{g}(\R)$ has two components, and $\mathcal{R}_{g}$ is called \textit{non-dividing} if $\mathcal{R}_{g}\setminus \mathcal{R}_{g}(\R)$ is connected (notice that an M-curve is always a dividing curve).

\begin{example}
Consider the hyperelliptic curve of genus $g$ defined by the equation
\begin{equation}
\mu^{2}= \prod_{i=1}^{2g+2} (\lambda-\lambda_{i}), \label{hyp real1}
\end{equation}
where the branch points $\lambda_{i}\in\R$ are ordered such that $\lambda_{1}<\ldots<\lambda_{2g+2}$. On such a curve, we can define two anti-holomorphic involutions $\tau_{1}$ and $\tau_{2}$, given respectively by $\tau_{1}(\lambda,\mu)=(\overline{\lambda},\overline{\mu})$ and $\tau_{2}(\lambda,\mu)=(\overline{\lambda},-\overline{\mu})$. Projections of real ovals of $\tau_{1}$  on the $\lambda$-plane coincide with  the intervals $[\lambda_{2g+2},\lambda_{1}],\ldots,[\lambda_{2g},\lambda_{2g+1}]$, whereas projections of real ovals of $\tau_{2}$ on the $\lambda$-plane coincide with  the intervals $[\lambda_{1},\lambda_{2}],\ldots,[\lambda_{2g+1},\lambda_{2g+2}]$. Hence the curve (\ref{hyp real1}) is an M-curve with respect to both anti-involutions $\tau_{1}$ and $\tau_{2}$.  \label{example rea hyp}
\end{example}

Let  $(\mathcal{A},\mathcal{B})$ be a basis of the homology group $H_{1}(\Rs_{g})$. 
According to Proposition 2.2 in  Vinnikov's paper 
\cite{Vin}, there exists a canonical homology basis (that we call for 
simplicity
`Vinnikov basis' in the following) such that
\begin{equation}
\left(\begin{matrix}
\tau \mathbf{\mathcal{A}}\\
\tau \mathbf{\mathcal{B}}
\end{matrix}\right)
=
\left(\begin{matrix}
\mathbb{I}_{g}&0\\
\mathbb{H}\,\,&-\mathbb{I}_{g}\,\,
\end{matrix}\right)
\left(\begin{matrix}
\mathbf{\mathcal{A}}\\
\mathbf{\mathcal{B}}
\end{matrix}\right), \label{hom basis}
\end{equation}
where $\mathbb{I}_{g}$ is the $g\times g$ unit matrix, and $\mathbb{H}$ is a  block diagonal $g\times g$ matrix, defined as follows:
\\\\
1) if $\mathcal{R}_{g}(\R)\neq \emptyset$, 
\[\mathbb{H}={\left(\begin{matrix}
0&1&&&&&&\\
1&0&&&&&&\\
&&\ddots&&&&&\\
&&&0&1&&&\\
&&&1&0&&&\\
&&&&&0&&\\
&&&&&&\ddots&\\
&&&&&&&0
\end{matrix}\right)} \quad \text{if $\mathcal{R}_{g}$ is dividing},\]
\[\mathbb{H}={\left(\begin{matrix}
1&&&&&\\
&\ddots&&&&\\
&&1&&&\\
&&&0&&\\
&&&&\ddots&\\
&&&&&0
\end{matrix}\right)} \quad \text{if $\mathcal{R}_{g}$ is non-dividing};\]
rank$(\mathbb{H})=g+1-k$ in both cases.
\\\\
2) if $\mathcal{R}_{g}(\R)= \emptyset$, (i.e. the curve does not have real oval), then
\[\mathbb{H}={\left(\begin{matrix}
0&1&&&\\
1&0&&&\\
&&\ddots&&\\
&&&0&1\\
&&&1&0
\end{matrix}\right)}
\quad \text{or} \quad 
\mathbb{H}={\left(\begin{matrix}
0&1&&&&\\
1&0&&&&\\
&&\ddots&&&\\
&&&0&1&\\
&&&1&0&\\
&&&&&0
\end{matrix}\right)};\]
rank$(\mathbb{H})=g$ if $g$ is even, rank$(\mathbb{H})=g-1$ if $g$ is odd.

Now let us choose the canonical homology basis in $H_{1}(\Rs_{g})$ 
satisfying (\ref{hom basis}), take $a,b\in\Rs_{g}$ and assume that 
$\tau a=a$ and $\tau b=b$. Denote by  $\ell$ a contour connecting the 
points $a$ and $b$ which does not intersect the canonical homology basis. Then
 the action of $\tau$ on the generators $\left(\mathbf{\mathcal{A}},\mathbf{\mathcal{B}},\ell\right)$ of the relative homology group $H_{1}(\Rs_{g},\{a,b\})$ is given by
\begin{equation}
\left(\begin{matrix}
\tau \mathbf{\mathcal{A}}\\
\tau \mathbf{\mathcal{B}}\\
\tau \ell
\end{matrix}\right)
=
\left(\begin{array}{ccc}
\mathbb{I}_{g}&0&0\\
\mathbb{H}\,\,&-\mathbb{I}_{g}\,\,&0\\
\mathbf{N}^{t}&\,\,\mathbf{M}^{t}&1
\end{array}\right)
\left(\begin{matrix}
\mathbf{\mathcal{A}}\\
\mathbf{\mathcal{B}} \\
\ell
\end{matrix}\right), \label{hom basis 3}
\end{equation}
where the vectors $\mathbf{N},\mathbf{M}\in\Z^{g}$ are related 
by (see \cite{Kalla})
\begin{equation}
2\,\mathbf{N}+\mathbb{H}\mathbf{M}=0. \label{NM stable}
\end{equation}

\subsection{Theta-functional solutions of the $n$-NLS$^{s}$ equation}

Algebro-geometric data associated to smooth theta-functional 
solutions of the $n$-NLS$^{s}$ equation (\ref{n-NLS}) consist of
$\{\Rs_{g},\tau,f,z_{a}\}$, where
$\mathcal{R}_{g}$ is a compact Riemann surface of genus $g>0$ 
dividing with respect to an anti-holomorphic involution $\tau$, and  
admitting a real meromorphic function  $f$ of degree $n+1$; here $z_{a}\in\R$ is a non critical value of $f$ such that the fiber $f^{-1}(z_{a})=\left\{a_{1},\ldots,a_{n+1}\right\}$ over $z_{a}$ belongs to the set $\Rs_{g}(\R)$. 
Let us choose natural local parameters $k_{a_{j}}$ near  $a_{j}$ given by the projection map $f$, namely, $k_{a_{j}}(p)=f(p)-z_{a}$ 
for any point $p\in\Rs_{g}$ lying in a neighbourhood of $a_{j}$.

Denote by  $\left(\mathbf{\mathcal{A}},\mathbf{\mathcal{B}},\ell_{j}\right)$ the generators of the relative homology group $H_{1}(\Rs_{g},\{a_{n+1},a_{j}\})$.
Let $\mathbf{d}\in\R^{g}$ and $\theta\in\R$. Then the following 
functions   $\psi_{j}, \,j=1,\ldots,n$,  give smooth solutions
of the $n$-NLS$^{s}$ equation  (\ref{n-NLS}), see \cite{Kalla},
\begin{equation}    
\psi_{j}(x,t)=|A_{j}|\,e^{\mathrm{i}\theta}\,\frac{\Theta(\mathbf{Z}-\mathbf{d}+\mathbf{r}_{j})}{\Theta(\mathbf{Z}-\mathbf{d})}\,\exp \left\{-\mathrm{i}\,(E_{j}\,x-F_{j}\,t)\right\},\label{sol n-NLS}
\end{equation}
where $|A_{j}|=|q_{2}(a_{n+1},a_{j})|^{1/2}\,\exp\left\{\tfrac{1}{2}\left\langle\mathbf{d},\mathbf{M}_{j}\right\rangle\right\}.$
The vector $\mathbf{M}_{j}\in\Z^{g}$ is defined by the action of $\tau$ on the relative homology group $H_{1}\left(\Rs_{g},\{a_{n+1},a_{j}\}\right)$ (see (\ref{hom basis 3})). 
Moreover, $\mathbf{r}_{j}=\int_{\ell_{j}}\omega$, and the vector $\mathbf{Z}$ reads
\[\mathbf{Z}=\mathrm{i}\,\mathbf{V}_{a_{n+1}}\,x+\mathrm{i}\,\mathbf{W}_{a_{n+1}}\,t,\] 
where vectors $\mathbf{V}_{a_{n+1}}$ and $\mathbf{W}_{a_{n+1}}$ are defined in (\ref{exp hol diff}). The scalars $E_{j},F_{j}$ are given by
\begin{equation}
E_{j}=K_{1}(a_{n+1},a_{j}),\qquad F_{j}=K_{2}(a_{n+1},a_{j})-2\,\sum_{k=1}^{n}q_{1}(a_{n+1},a_{k}), \label{E,N n-NLS}
\end{equation}
and scalars $q_{i},K_{i}$, $i=1,2$ are defined in (\ref{q1})-(\ref{K2}). 
According to \cite{Kalla}, necessary conditions for the functions $\psi_{j}$ in (\ref{sol n-NLS}) to solve the $n$-NLS$^{s}$ equation are the identities (\ref{cor Fay2}) and (\ref{my corol}) with $(a,b):=(a_{n+1},a_{j})$.

The signs $s_{1},\ldots,s_{n}$ in (\ref{n-NLS}) are given by 
\begin{equation}
s_{j}= \exp\left\{\mathrm{i}\pi(1+\alpha_{j})\right\}, \label{sj}
\end{equation} 
where  $\alpha_{j}\in\Z$ denote certain intersection indices on $\Rs_{g}$ defined as follows:
let $\tilde{a}_{n+1},\tilde{a}_{j}\in\Rs_{g}(\R)$ lie in a neighbourhood 
of $a_{n+1}$ and $a_{j}$ respectively such that $f(\tilde{a}_{n+1})=f(\tilde{a}_{j})$. Denote by $\tilde{\ell}_{j}$ an oriented contour connecting $\tilde{a}_{n+1}$ and $\tilde{a}_{j}$.
Then 
\begin{equation}
\alpha_{j}=(\tau \tilde{\ell}_{j}-\tilde{\ell}_{j})\circ\ell_{j}  \label{alpha}
\end{equation}
is the intersection index of the closed contour $\tau 
\tilde{\ell}_{j}-\tilde{\ell}_{j}$ and the contour $\ell_{j}$; this 
index is computed in the relative homology group $H_{1}(\Rs_g,\{a_{n+1},a_{j}\})$.

In particular, it was shown in \cite{Kalla} that solutions of the 
focusing $n$-NLS$^{s}$ equation, i.e., for $s=(1,\ldots,1)$, are obtained when the branch points of the meromorphic function $f$ are pairwise conjugate.

\subsection{Theta-functional solutions of the DS equations}

Now let us introduce smooth solutions of the DS equations.
In characteristic coordinates 
\[\xi=\frac{1}{2}(x-\mathrm{i}\alpha \,y),\quad \eta=\frac{1}{2}(x+\mathrm{i}\alpha \,
y),\qquad  \alpha=\mathrm{i} \, \text{or}\, 1,\]
the DS equations  (\ref{DSintro}) take the form
\begin{align}	
\mathrm{i}\psi_{t}+\frac{1}{2}(\partial_{\xi\xi}+\partial_{\eta\eta})\psi+2\,\phi\,\psi &=0,     \nonumber  \\
\partial_{\xi}\partial_{\eta}\phi+\rho\, 
\frac{1}{2}(\partial_{\xi\xi}+\partial_{\eta\eta})|\psi|^{2} & =0,  \label{DS} 
\end{align}
where $\phi:=\Phi+\rho\,|\psi|^{2}$, $\rho=\pm1$.
Recall that DS1$^{\rho}$ denotes the Davey-Stewartson 
equation when $\alpha=\mathrm{i}$ (in this case $\xi$ and $\eta$ are both 
real), and DS2$^{\rho}$ when $\alpha=1$ (in this case $\xi$ and $\eta$ are pairwise conjugate).

In both cases, for the DS1$^{\rho}$ and 
DS2$^{\rho}$ equations, the solutions have the form \cite{Mal, Kalla}:
\begin{equation}
\psi(\xi,\eta,t)=|A|\,e^{\mathrm{i}\theta}\,\frac{\Theta(\mathbf{Z}-\mathbf{d}+\mathbf{r})}{\Theta(\mathbf{Z}-\mathbf{d})}\,\exp\left\{ -\mathrm{i}\left(G_{1}\,\xi+G_{2}\,\eta-G_{3}\,\tfrac{t}{2}\right)\right\}, \label{psi DS}
\end{equation}
\begin{equation}
\phi(\xi,\eta,t)=\frac{1}{2}\,(\ln\,\Theta(\mathbf{Z}-\mathbf{d}))_{\xi\xi}+\frac{1}{2}\,(\ln\,\Theta(\mathbf{Z}-\mathbf{d}))_{\eta\eta}+\frac{h}{4}.\label{phi DS}
\end{equation}
Here $\mathbf{r}=\int_{a}^{b}\omega$ for some distinct points $a,b\in\Rs_{g}$, and the vector $\mathbf{Z}$ is 
defined as
\begin{equation}
\mathbf{Z}=\mathrm{i}\left(\kappa_{1}\,\mathbf{V}_{a}\,\xi-\kappa_{2}\,\mathbf{V}_{b}\,\eta+(\kappa^{2}_{1}\,\mathbf{W}_{a}-\kappa^{2}_{2}\,\mathbf{W}_{b})\,\tfrac{t}{2}\right). \label{Z DS}
\end{equation}
Moreover, the scalars $G_{1},\,G_{2}$ and $G_{3}$ read
\begin{equation}
G_{1}=\kappa_{1}\,K_{1}(a,b),\qquad G_{2}=\kappa_{2}\,K_{1}(b,a), \label{N12 DS}
\end{equation}
\begin{equation}
G_{3}=\kappa^{2}_{1}\,K_{2}(a,b)+\kappa^{2}_{2}\,K_{2}(b,a)+h, \label{N3 DS}
\end{equation}
where the scalars $K_{1},K_{2}$ are defined in (\ref{K1}) and (\ref{K2}). 
As shown in \cite{Kalla}, necessary conditions for the functions $\psi$ (\ref{psi DS}) and $\phi$ (\ref{phi DS}) to solve the DS equations are  the identities (\ref{cor Fay2}) and (\ref{my corol}).

Algebro-geometric data associated to smooth solutions (\ref{psi DS}), 
(\ref{phi DS}) 
of the DS1$^{\rho}$ equation consist of 
$\{\Rs_{g},\tau,a,b,k_{a},k_{b}\}$, where $\Rs_{g}$ is a compact Riemann surface of genus $g>0$, dividing with respect to an 
anti-holomorphic involution $\tau$, $a,b$ are two distinct points  in 
$\Rs_{g}(\R)$, and $k_{a},k_{b}$ denote local parameters near $a$ and 
$b$ respectively which satisfy $\overline{k_{a}(\tau p)}=k_{a}(p)$ 
for any $p$ lying  in a neighbourhood of $a$, and $\overline{k_{b}(\tau 
p)}=k_{b}(p)$ for any $p$ lying  in a neighbourhood of $b$. The remaining 
quantities satisfy the conditions: $\mathbf{d}\in\R^{g}$, 
$\theta,h\in\R$, $\kappa_{2}\in\R\setminus\left\{0\right\}$, and
\begin{equation}
\kappa_{1}=-\rho\,\tilde{\kappa}_{1}^{2}\,\kappa_{2}\,q_{2}(a,b)\,\exp\left\{\tfrac{1}{2}\left\langle \B \mathbf{M},\mathbf{M}\right\rangle+\left\langle \mathbf{r}+\mathbf{d},\mathbf{M}\right\rangle\right\}, \label{kappa DS1}
\end{equation}
for some $\tilde{\kappa}_{1}\in\R$, where $\mathbf{M}\in\Z^{g}$ is 
defined in (\ref{hom basis 3}). The scalar $|A|$ is given by
\[|A|=\left|\tilde{\kappa}_{1}\,\kappa_{2}\,q_{2}(a,b)\right|\,\exp\left\{\left\langle \mathbf{d},\mathbf{M}\right\rangle\right\},\]
where the quantity $q_{2}$ is defined in (\ref{q2}).

Algebro-geometric data associated to smooth solutions (\ref{psi DS}), 
(\ref{phi DS}) of the DS2$^{\rho}$ equation consist of 
$\{\Rs_{g},\tau,a,b,k_{a},k_{b}\}$, where $\Rs_{g}$ is a compact Riemann surface of genus $g>0$ with an 
anti-holomorphic involution $\tau$,  $a,b$ are two distinct points such that $\tau a=b$, and $k_{a},k_{b}$ denote local parameters near $a$ and $b$ respectively which satisfy $\overline{k_{b}(\tau p)}=k_{a}(p)$ for any point $p$ lying in a neighbourhood of $a$. Moreover, $\mathbf{d}\in\mathrm{i}\R^{g}$, $\theta,h\in\R$, $\kappa_{1},\kappa_{2}\in\C\setminus\left\{0\right\}$ satisfy  $\overline{\kappa_{1}}=\kappa_{2}$, and the scalar $|A|$ is given by
\[|A|=|\kappa_{1}|\,|q_{2}(a,b)|^{1/2}.\]
Smooth solutions of the DS2$^{+}$ equation are obtained when the curve $\Rs_{g}$ is an M-curve with respect to $\tau$, whereas solutions to DS2$^{-}$ are smooth if  the associated Riemann surface does not have real oval with respect to $\tau$, 
and if there is no pseudo-real function of degree $g-1$ on it (i.e., 
function which satisfies $\overline{f(\tau p)}=-f(p)^{-1}$), see 
\cite{Mal}.

\begin{remark}
    The symmetric structure of the DS equations (\ref{DS}) with respect to $\xi$ 
    and $\eta$  implies that a solution $\psi=\Psi(\xi,\eta,t)$ to 
    DS1$^{+}$ leads to a solution $\Psi(-\xi,\eta,t)$ of DS1$^{-}$. 
\end{remark}

\section{Hyperelliptic case}
Here we consider concrete examples for the solutions, in 
terms of multi-dimensional theta functions, to DS and $n$-NLS$^{s}$ on 
hyperelliptic  Riemann surfaces. We first review the numerical methods to visualize 
the solutions and discuss how the accuracy is tested.

\subsection{Computation on real hyperelliptic curves}

The simplest example of algebraic curves are hyperelliptic 
curves, 
\begin{equation}\mu^{2}=
\begin{cases}
     &   \prod_{i=1}^{2g+2}(\lambda-\lambda_{i}), \mbox{ without 
     branching at infinity}\\
     & \prod_{i=1}^{2g+1}(\lambda-\lambda_{i}), \mbox{ with 
     branching at infinity}
\end{cases},  \label{hyperel}
\end{equation}
where $g$ 
is the genus of the Riemann surface, and where we have for the branch 
points $\lambda_{i}\in \mathbb{C}$ the relations 
$\lambda_{i}\neq\lambda_{j}$  for $i\neq j$. If the number of finite branch points 
is odd, the curve is 
branched at infinity. Recall that all Riemann surfaces of 
genus $2$ are hyperelliptic, and that the involution $\sigma$ which interchanges 
the sheets, $\sigma(\lambda,\mu)=(\lambda,-\mu)$, is an automorphism on  any
hyperelliptic curve in the  form (\ref{hyperel}). A vector of 
holomorphic differentials for these  surfaces is given by 
$(1,\lambda,\ldots,\lambda^{g-1})^{t}\ud\lambda/\mu$.  
For a real hyperelliptic  curve, the 
branch points are either real or pairwise conjugate. As we saw in Example 2.1,  if all branch points $\lambda_{i}$ are real and ordered such that  
$\lambda_{1}<\ldots<\lambda_{2g+2}$, the hyperelliptic curve is an M-curve with respect to both
anti-holomorphic involutions  $\tau_{1}$ and $\tau_{2}$ defined in 
the example. The other case of 
interest in the context of smooth solutions to $n$-NLS$^{s}$ and DS are real
curves without real branch point. For the involution $\tau_{1}$, a curve given by 
$\mu^{2}=\prod_{i=1}^{g+1}(\lambda-\lambda_{i})(\lambda-\overline{\lambda}_{i})$, with $\lambda_{i}\in \mathbb{C}\setminus\mathbb{R}$, $i=1,\ldots,g+1$, in this case is 
dividing (two points whose projections onto $\mathbb{C}$ have 
respectively a positive and a negative imaginary part cannot be connected by a contour which does not cross a real oval), whereas a curve given by 
$\mu^{2}=-\prod_{i=1}^{g+1}(\lambda-\lambda_{i})(\lambda-\overline{\lambda}_{i})$ has no real oval, 
and vice versa for the involution $\tau_{2}$.

In the following, we will only consider real  
hyperelliptic curves without branching 
at infinity and write the defining equation in the form 
$\mu^{2}= (\lambda-\xi)(\lambda-\eta)\prod_{i=1}^{g}(\lambda-E_{i})(\lambda-F_{i})$. 
It is  possible to introduce a convenient homology basis 
on the related surfaces, see 
Fig.~\ref{cutsystem} for the case $\eta=\overline{\xi}$.
\begin{figure}[htb!]
\begin{center}
  \includegraphics[width=0.7\textwidth]{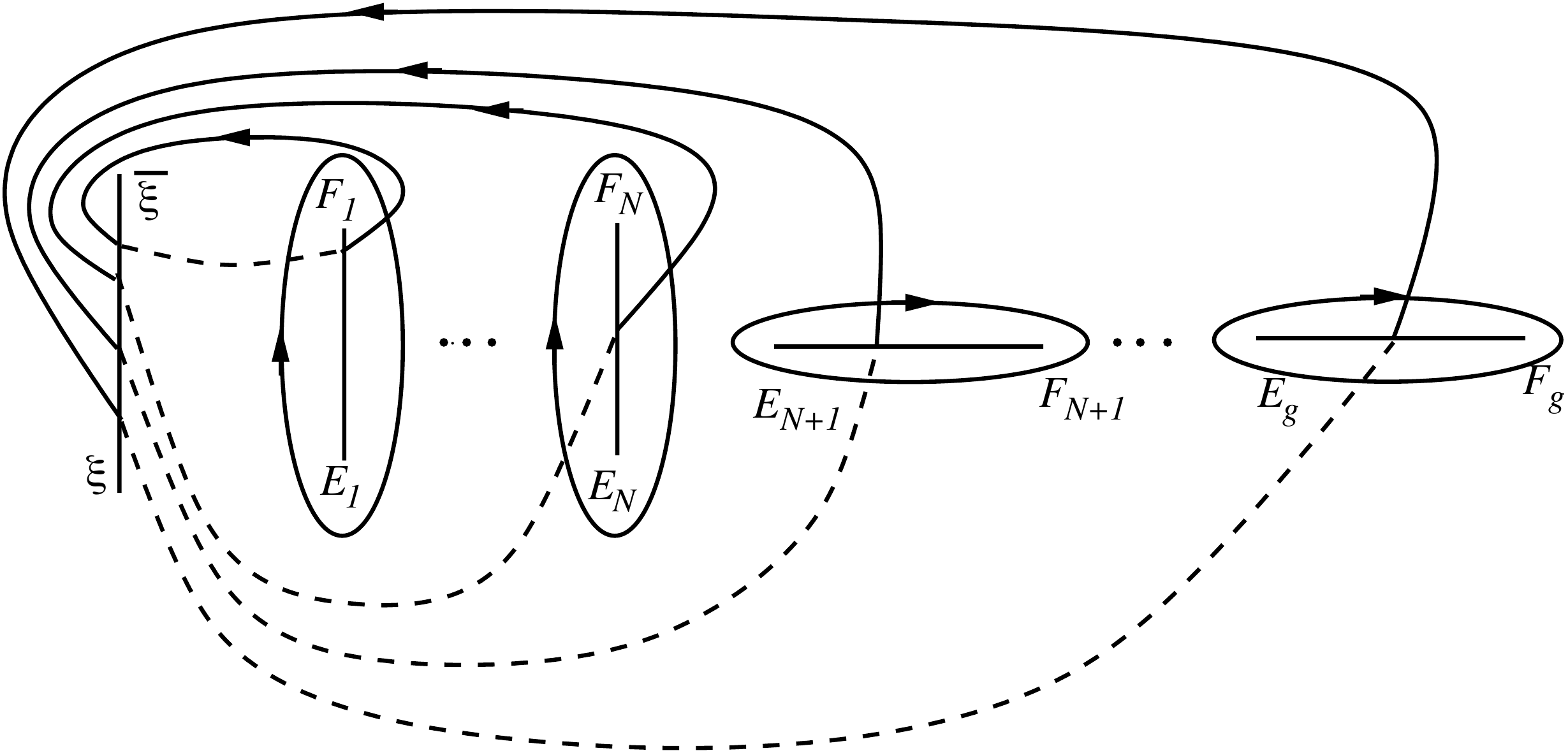}
\end{center}
 \caption{\textit{Homology basis on real hyperelliptic 
 curves, contours on sheet 1 are solid, contours on sheet 2 are dashed.
 $\mathcal{A}$-cycles are the closed contours entirely on sheet 1.}}
   \label{cutsystem}
\end{figure}

The simple form of 
the algebraic relation between $\mu$ and $\lambda$ for hyperelliptic 
curves makes the generation of very efficient numerical codes possible, see, 
for instance, \cite{FK1,FK2} for details.  These codes allow the treatment of almost degenerate 
Riemann surfaces, i.e., the case where the branch points almost 
collide pairwise,  where the 
distance of the branch points is of the order of machine 
precision: $|E_{i}- F_{i}|\sim 10^{-14}$. The homology basis Fig.~\ref{cutsystem} is adapted to this kind of degeneration.

The Abel map $\int_{a}^{b}\omega$ between two points $a$ and $b$ is computed in the following way: 
the sheet identified at the point $a=(\lambda(a),\mu(a))$ (where we take 
for $\mu$ the root computed by Matlab) is labeled sheet $1$, and at the point $(\lambda(a),-\mu(a))$, sheet $2$. Then the 
ramification point whose projection to the $\lambda$-sphere has the minimal distance to $\lambda(a)$ is determined. For 
simplicity we assume always that this is the point $\xi$ in 
Fig.~\ref{cutsystem} (for another branch point, this leads to the 
addition of half-periods, see e.g.~\cite{BBEIM}). This means we 
compute $\int_{a}^{b}\omega$ as 
$\int_{a}^{b}\omega=\int_{\xi}^{b}\omega-\int_{\xi}^{a}\omega$. The choice of a 
branch point as the base point of the Abel map has the advantage that 
a change of sheet of a point $a$ just implies a change of sign of the 
integral: 
$\int_{\xi}^{(\lambda(a),\mu(a))}\omega=-\int_{\xi}^{(\lambda(a),-\mu(a))}\omega$. 
To compute the integral $\int_{\xi}^{a}\omega$, one has to 
analytically continue $\mu$ on the connecting line between 
$\lambda(a)$ and $\xi$ onto the $\lambda$-sphere. Whereas the root $\mu$ is not 
supposed to have any branching on the considered path,  the square root in Matlab is 
branched on the negative real axis.  To analytically continue $\mu$ 
on the path $[\lambda(a),\xi]$, we compute  the Matlab root  at some  
$\lambda_{j}\in[\lambda(a),\xi] $, $j=0,\ldots,N_{c}$ and 
analytically continue it starting from $\mu(a)$ by demanding that 
$|\mu(\lambda_{j+1})-\mu(\lambda_{j})|<|\mu(\lambda_{j+1})+\mu(\lambda_{j})|$.
The 
so defined sheets will be denoted here and in the following by 
numbers, i.e., a point on sheet 1 with projection $\lambda(a)$ into the base 
is denoted by $(\lambda(a))^{(1)}$.

Thus the computation of the Abel map is reduced to the computation of 
line integrals on the 
connecting line between $\lambda(a)$ and $\xi$ in the complex 
$\lambda$-plane. For the numerical computation of such integrals we use 
Clenshaw-Curtis integration (see, for instance, \cite{tref}): to compute 
an integral $\int_{-1}^{1}h(x)\ud x$, this algorithm
samples the integrand on the $N_{c}+1$ Chebyshev 
collocation points $x_{j}=\cos(j\pi/N_{c})$, 
$j=0,\ldots,N_{c}$. The integral is approximated as the sum: 
$\int_{-1}^{1}h(x)\ud x\sim \sum_{j=0}^{N_{c}}w_{j}\,h(x_{j})$ (see 
\cite{tref} on how to obtain the weights $w_{j}$). It can be shown 
that the convergence of the integral is exponential for analytic 
functions $h$ as the ones considered here. To compute the Abel map, 
one uses the transformation $\lambda \to 
\lambda(a)(1+x)/2+\xi(1-x)/2$,  to the  Clenshaw-Curtis integration 
variable. The same 
procedure is then carried out for the integral from $\xi$ to $b$.

The theta functions are approximated as in \cite{FK1} as a sum,
\begin{equation}
    \Theta_{\B}[\delta](\mathbf{z}) \sim 
    \sum_{m_{1}=-N_{\theta}}^{N_{\theta}}\ldots 
    \sum_{m_{g}=-N_{\theta}}^{N_{\theta}}\exp\left\{
    \tfrac{1}{2}\langle \B(\mathbf{m}+\delta_1),\mathbf{m}+\delta_1\rangle+\langle \mathbf{m}+\delta_1,\z+2\mathrm{i}\pi\delta_2\rangle\right\}
    \label{thetanum}.
\end{equation}
The periodicity properties of the theta function (\ref{2.4}) 
make it possible to write $\mathbf{z}=\mathbf{z}_{0}+2\mathrm{i}\pi 
\mathbf{N}+\mathbb{B}\mathbf{M}$  for some $\mathbf{N},\mathbf{M}\in\Z^{g}$, where $\mathbf{z}_{0}=2\mathrm{i}\pi 
\mathbf{\alpha}+\mathbb{B}\mathbf{\beta}$ with 
$\alpha_{i}, \beta_{i}\in\,\, ]-\frac{1}{2},\frac{1}{2}]$ for $i=1,\ldots,g$. The 
value of $ N_{\theta}$ is determined by the condition that all terms 
in (\ref{theta}) with $|m_{i}|>N_{\theta}$ are smaller than machine precision, which 
is controlled by the largest eigenvalue of the real part of the 
Riemann matrix  (the one with minimal 
absolute value since the real part is negative definite), see \cite{FK1,FK}. 

To control the accuracy of the numerical solutions, we use 
essentially two approaches. First we check the theta identity (\ref{my 
corol}), which is the underlying reason for the studied functions being 
solutions to $n$-NLS$^{s}$ and DS, at each point in the spacetime. This test
requires the computation of theta  derivatives not needed in the 
solution (which slightly reduces the efficiency of the code since 
additional quantities are computed), 
but provides an immediate check whether the solution satisfies (\ref{my 
corol}) with the required accuracy. Since this identity is not implemented in the code, it 
provides a strong test. This ensures that all quantities entering the solution are 
computed with the necessary precision.
In addition, the solutions are computed on Chebyshev 
collocation points (see, for instance, \cite{tref}) for each of the 
physical variables. This can be used for an expansion of the 
computed solution in terms of Chebyshev polynomials, a so-called 
spectral method having in 
practice exponential convergence for analytic functions as the ones 
considered here. Since the 
derivatives of the Chebyshev polynomials can be expressed linearly in terms of Chebyshev 
polynomials, a derivative acts on the space of polynomials via a so 
called differentiation matrix. With these standard Chebyshev differentiation 
matrices (see \cite{tref}), the solution can be numerically 
differentiated. The computed derivatives allow to check with which 
numerical precision the PDE is satisfied by a numerical solution. With these two independent tests, we ensure that the 
shown solutions are correct to much better than plotting accuracy 
(the code reports a warning if the above tests are not satisfied to 
better than $10^{-6}$).

\subsection{Solutions to the DS equations}

The elliptic solutions are the well known travelling wave solutions 
and will not be discussed here. 
The simplest examples we will consider for the DS solutions are 
given on hyperelliptic  curves of genus $2$. As we saw in Section 2.4, for DS1$^{\rho}$ reality and 
smoothness conditions  imply that the branch points of the 
 curve are either all real (M-curve) or all pairwise conjugate (dividing curve). The points $a$ 
and $b$ must project to real points on the $\lambda$-sphere and must be stable under the 
anti-holomorphic involution $\tau$ (we use here $\tau=\tau_{1}$, as defined 
in Example 2.1, except for DS2$^{-}$). For DS2$^{\rho}$, we have $\tau a=b$ 
where the projection of $a$ onto the $\lambda$-sphere is the conjugate of the 
projection of $b$. For DS2$^{+}$ the curve must have only real branch 
points (M-curve), whereas for DS2$^{-}$ it must have 
no real oval. 

For DS we will mainly give plots for fixed time since for low genus, 
the solution is essentially travelling in one direction. For higher 
genus, we show a more interesting time dependence in Fig.~\ref{figDS2mg4}.

We first consider the defocusing variants, DS1$^{+}$ and DS2$^{+}$ on 
M-curves. In genus $2$ we study the family of curves with the branch 
points $-2,-1,0,\epsilon,2,2+\epsilon$ for  $\epsilon=1$ 
and $\epsilon=10^{-10}$. In the former case the solutions will be 
periodic in the $(x,y)$-plane, in the latter almost solitonic since the Riemann surface is 
almost degenerate (in the limit $\epsilon\to 0$ the surface degenerates 
to a surface of genus  $0$; the resulting solutions are discussed 
in more detail in \cite{Kalla2}). To obtain non-trivial solutions in the solitonic limit, 
we use  $\mathbf{d}=\frac{1}{2}
\left[\begin{smallmatrix}
    1 & 1  \\
    0 & 0
\end{smallmatrix}\right]^{t}
$ in all examples.
In Fig.~\ref{figDS1pg2} it can be seen that these are 
in fact \emph{dark solitons}, i.e., the solutions tend asymptotically 
to a non-zero constant and the solitons thus represent `shadows' on a background of light. The 
well known features from soliton collisions for (1+1)-dimensional integrable equations, 
namely, the propagation without change of shape, and the 
unchanged shape and phase shift after the collision, can be seen here 
in the $(x,y)$-plane. 

The corresponding solutions to DS2$^{+}$ can be seen in 
Fig.~\ref{figDS2pg2}. We only show the square modulus of the solution here for simplicity. 
For the real and the imaginary part of such a solution for the DS1$^{-}$-case, see 
Fig.~\ref{figDS1mg2}. Because of remark 2.1
all solutions shown for DS1$^{+}$ on M-curves are after the
    change of coordinate $\xi\to-\xi$ solutions to DS1$^{-}$.  For this reason 
    DS1$^{-}$  solutions on M-curves will not be presented here.

\begin{figure}[htb!]
\begin{center}
\includegraphics[width=0.45\textwidth]{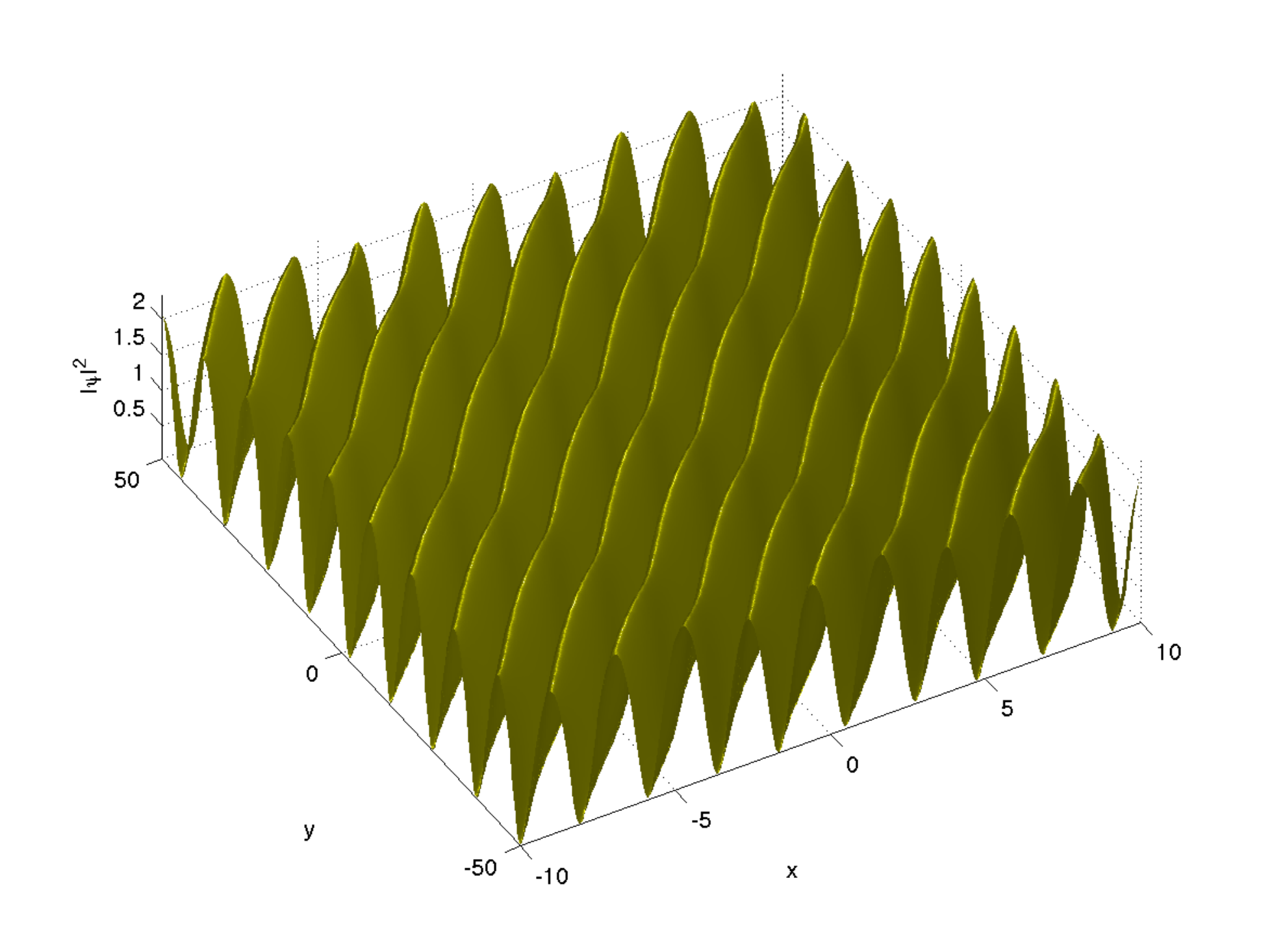}
\includegraphics[width=0.45\textwidth]{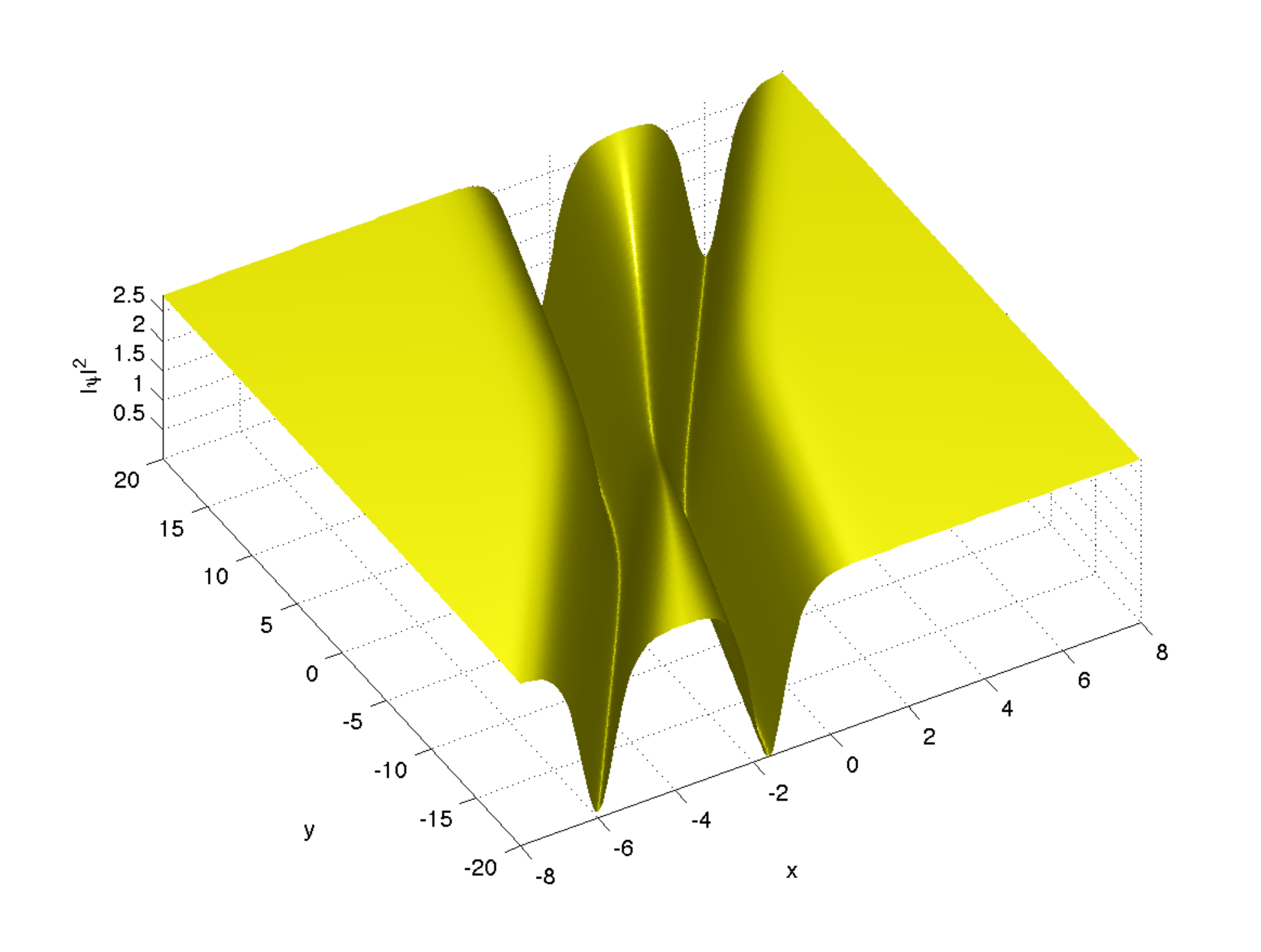}
\end{center}
 \caption{\textit{Solution (\ref{psi DS}) to the DS1$^{+}$ equation at $t=0$ 
 on a hyperelliptic curve of 
 genus 2 with branch points $-2,-1,0,\epsilon,2,2+\epsilon$ and 
 $a=(-1.9)^{(1)}$, $b=(-1.1)^{(2)}$ for $\epsilon=1$ on the left and 
 $\epsilon=10^{-10}$, the almost solitonic limit, on the right.}}
   \label{figDS1pg2}
\end{figure}

\begin{figure}[htb!]
\begin{center}
\includegraphics[width=0.45\textwidth]{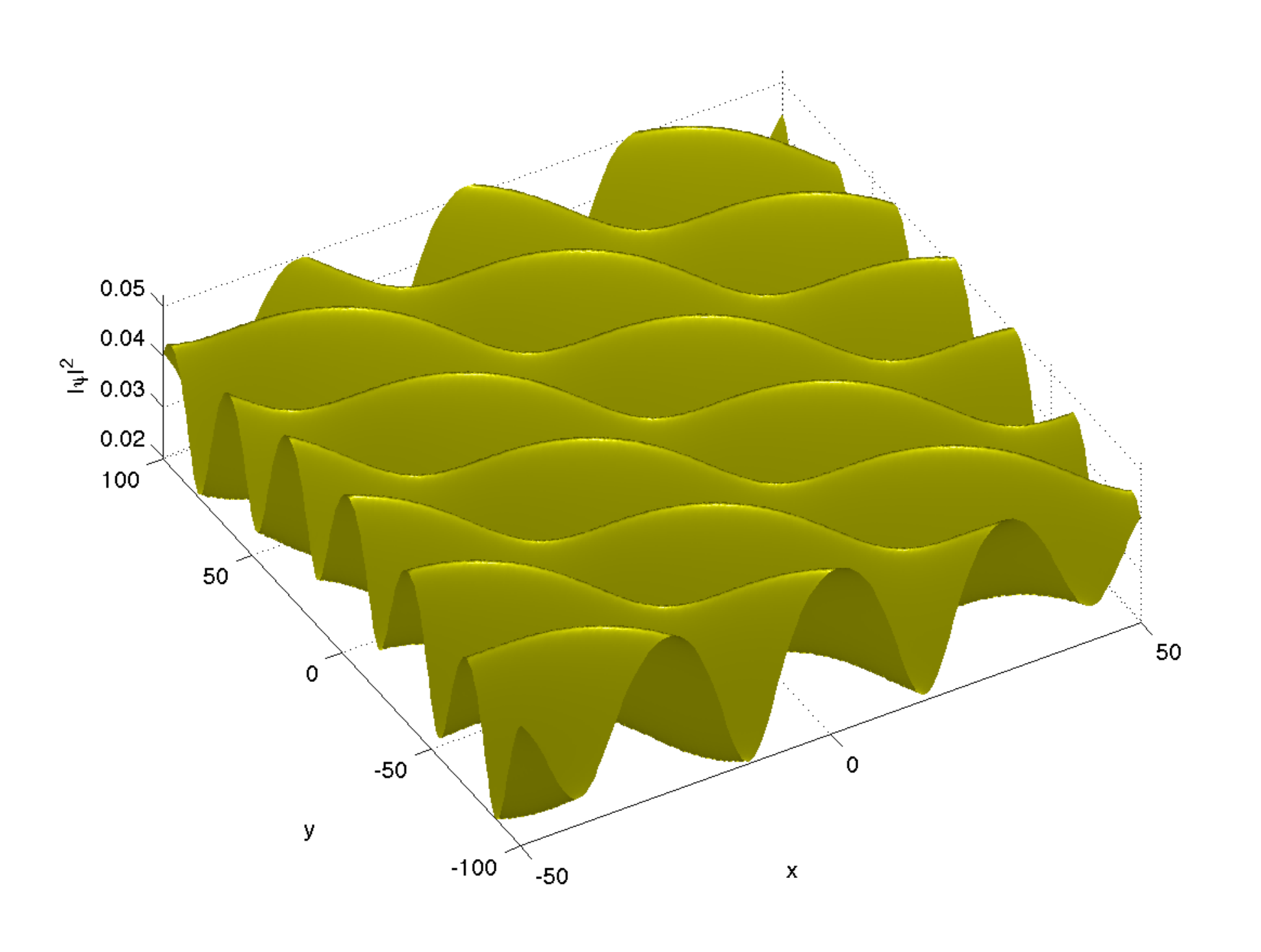}
\includegraphics[width=0.45\textwidth]{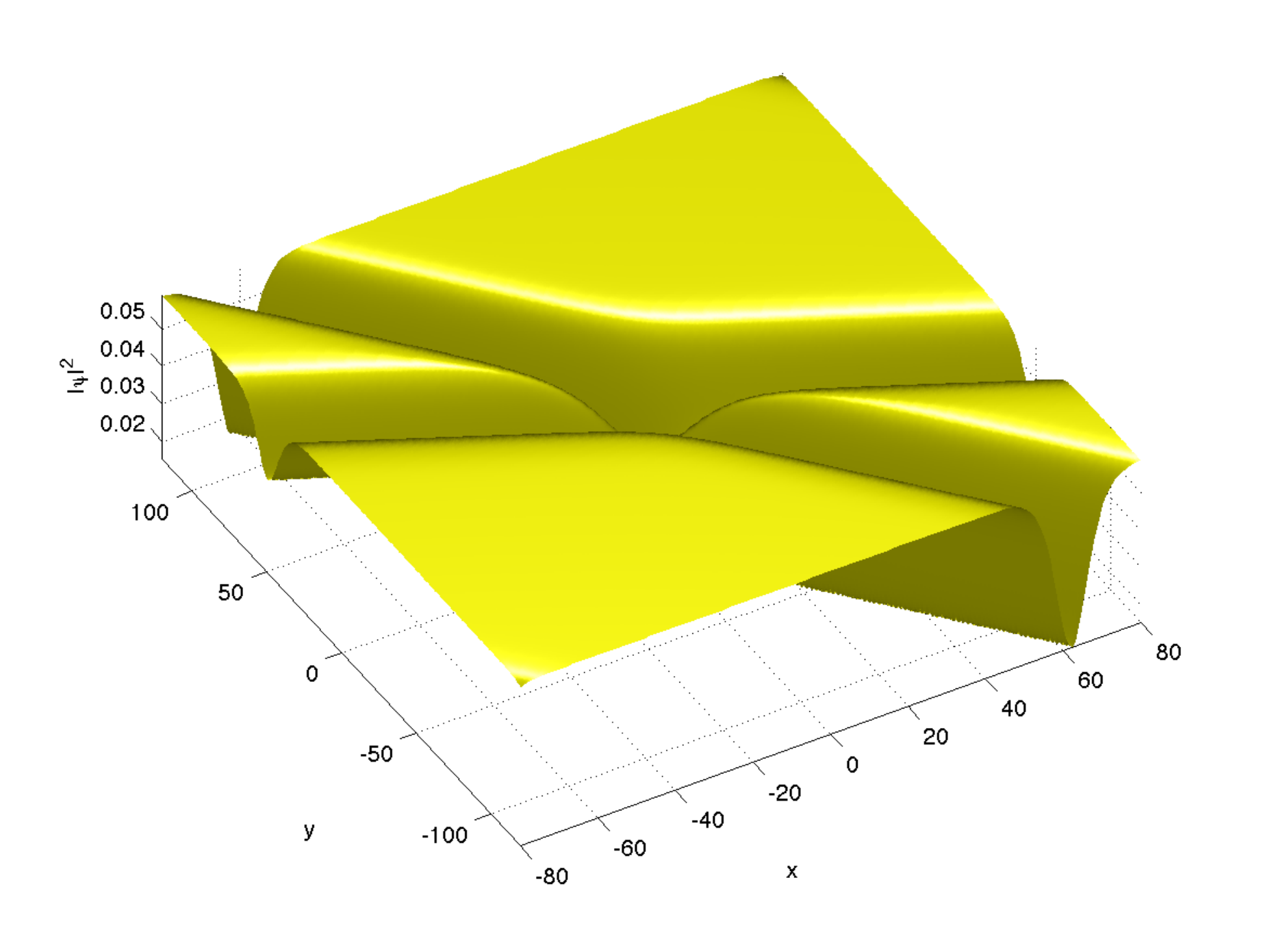}
\end{center}
 \caption{\textit{Solution (\ref{psi DS}) to the DS$2^{+}$ equation  at $t=0$ on a hyperelliptic curve of 
 genus 2 with branch points $-2,-1,0,\epsilon,2,2+\epsilon$ and 
 $a=(-1.5+2\mathrm{i})^{(1)}$, $b=(-1.5-2\mathrm{i})^{(2)}$ for $\epsilon=1$ on the left and 
 $\epsilon=10^{-10}$, the almost solitonic limit, on the right.}}
   \label{figDS2pg2}
\end{figure}

In the same way one can study, on a genus $4$ hyperelliptic curve, the formation of 
the dark 4-soliton for these two equations. We consider the curve
with branch points 
$-4,-3,-2,-2+\epsilon,0,\epsilon,2,2+\epsilon,4,4+\epsilon$ for 
$\epsilon=1$ and $\epsilon=10^{-10}$, and use  
$\mathbf{d}=\frac{1}{2}
\left[\begin{smallmatrix}
    1 & 1 & 1 & 1 \\
    0 & 0 & 0 & 0
\end{smallmatrix}\right]^{t}
$. The DS1$^{+}$ solutions for 
this  curve can be seen in Fig.~\ref{figDS1pg4}. The corresponding solutions to DS2$^{+}$  is shown in 
Fig.~\ref{figDS2pg4}.

\begin{figure}[htb!]
\begin{center}
\includegraphics[width=0.45\textwidth]{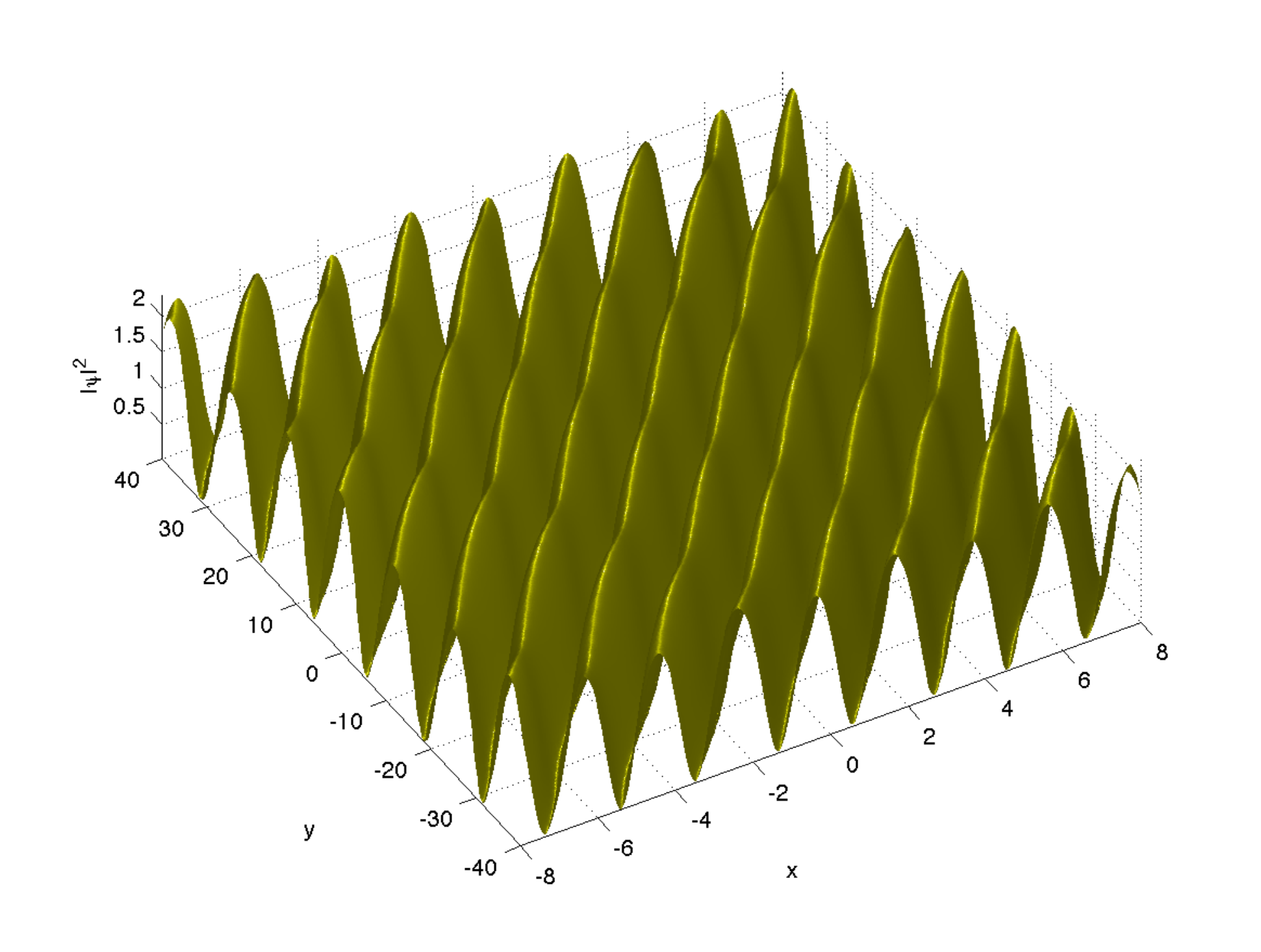}
\includegraphics[width=0.45\textwidth]{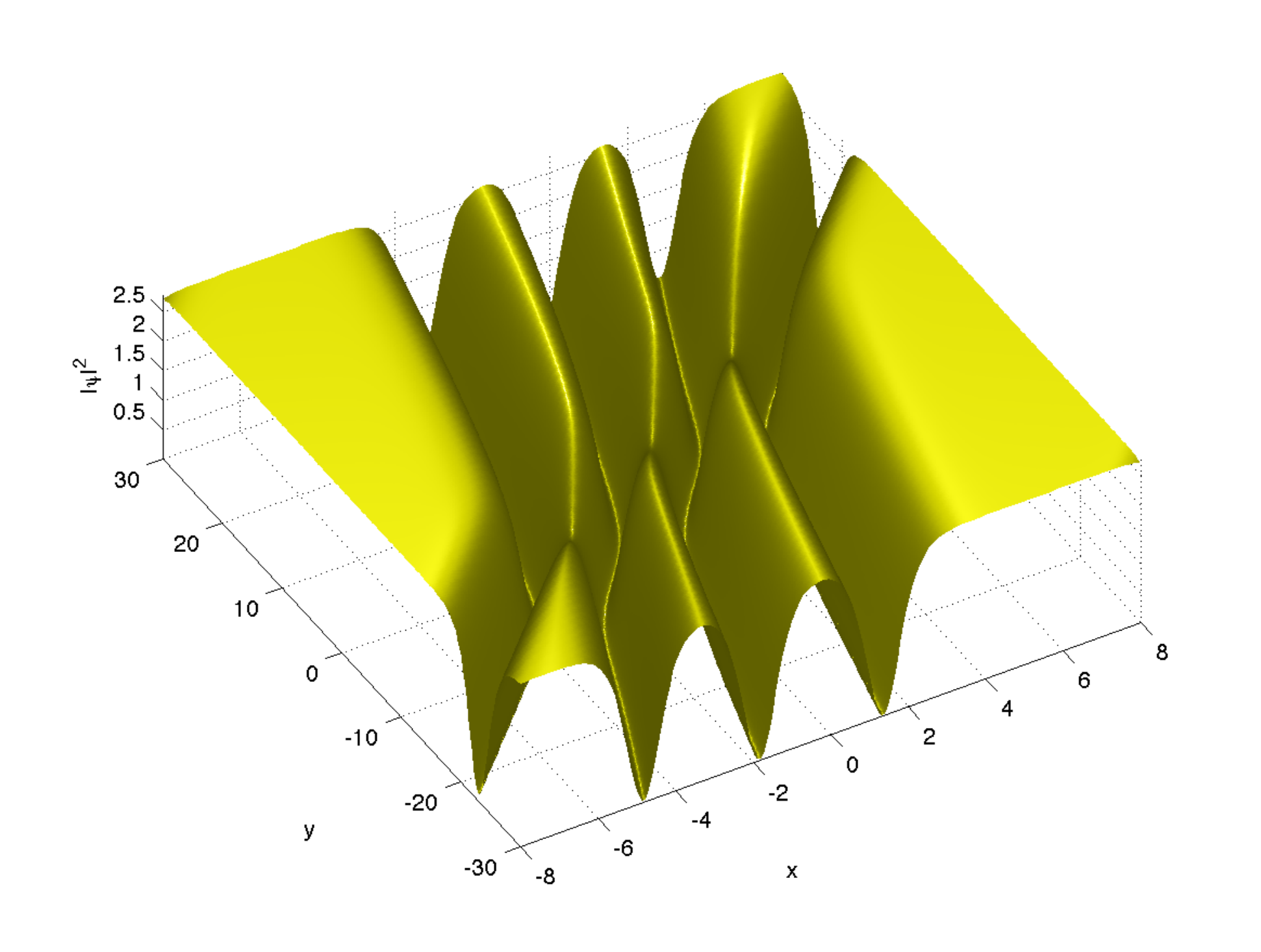}
\end{center}
 \caption{\textit{Solution (\ref{psi DS}) to the DS1$^{+}$ equation  at $t=0$ 
 on a  hyperelliptic curve of 
 genus 4 with branch points $-4,-3,-2,-2+\epsilon,0,\epsilon,2,2+\epsilon,4,4+\epsilon$ and 
 $a=(-3.9)^{(1)}$, $b=(-3.1)^{(2)}$ for $\epsilon=1$ on the left and 
 $\epsilon=10^{-10}$, the almost solitonic limit, on the right.}}
   \label{figDS1pg4}
\end{figure}

\begin{figure}[htb!]
\begin{center}
\includegraphics[width=0.45\textwidth]{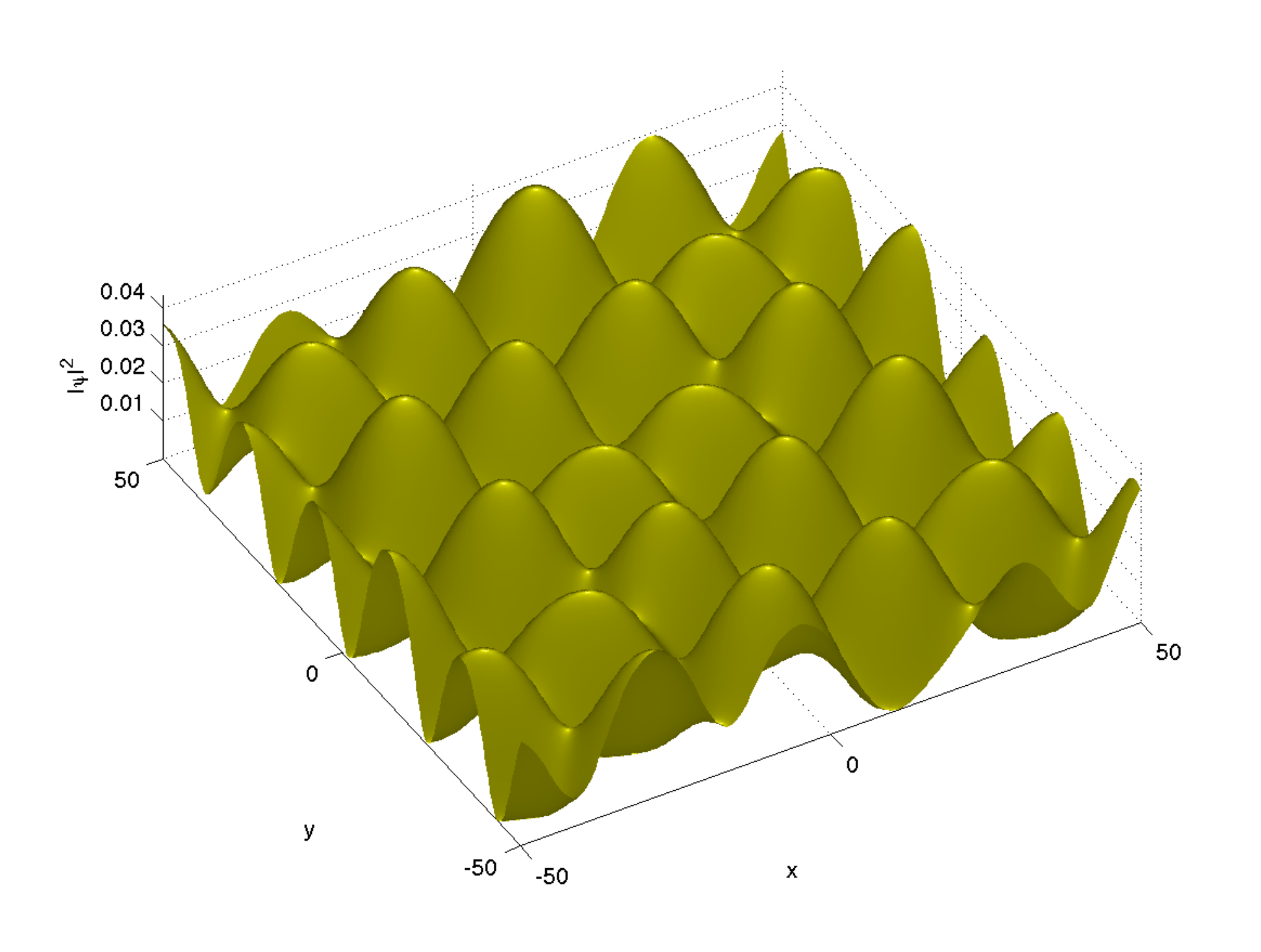}
\includegraphics[width=0.45\textwidth]{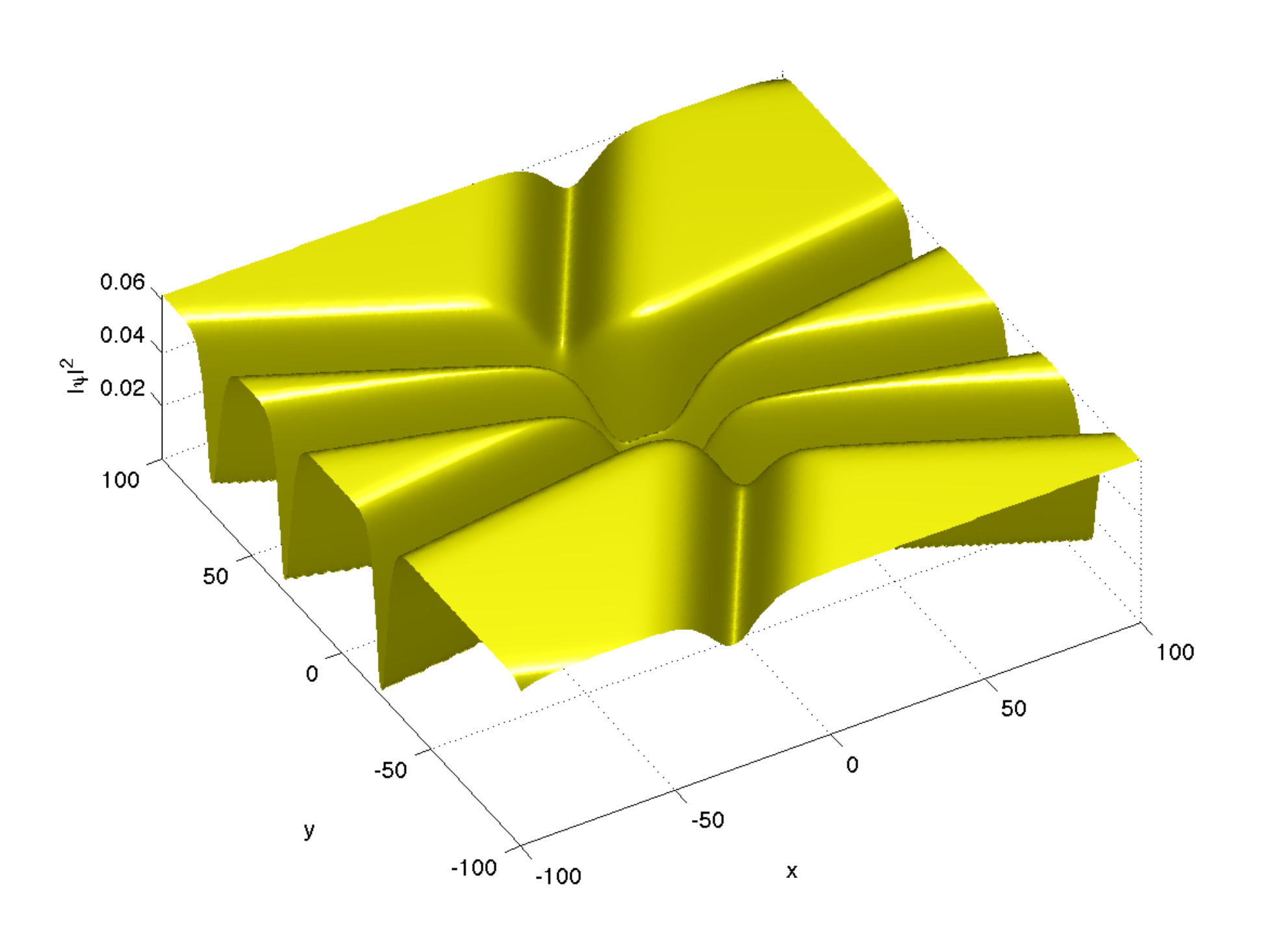}
\end{center}
 \caption{\textit{Solution (\ref{psi DS}) to the DS2$^{+}$ equation at $t=0$ on a hyperelliptic curve of 
 genus 4 with branch points $-4,-3,-2,-2+\epsilon,0,\epsilon,2,2+\epsilon,4,4+\epsilon$ 
 and 
 $a=(-1.5+2\mathrm{i})^{(1)}$, $b=(-1.5-2\mathrm{i})^{(1)}$ for $\epsilon=1$ on the left and 
 $\epsilon=10^{-10}$, the almost solitonic limit, on the right.}}
   \label{figDS2pg4}
\end{figure}

Solutions to the focusing variants of these equations can be obtained 
on  hyperelliptic curves with pairwise conjugate branch points. For such solutions
the solitonic limit cannot be obtained as above since the quotient of 
theta functions in (\ref{psi DS})
tends to a constant in this case. To obtain the well-known \emph{bright 
solitons} (solutions tend to zero at spatial infinity) in this way, the  hyperelliptic curve has to be completely degenerated 
(all branch points must collide pairwise to double points) which leads to limits of 
the form '$0/0$' in the expression for the solution (\ref{psi DS})
which are not convenient for a numerical treatment; 
see \cite{Kalla2} for an analytic discussion. Therefore we only 
consider non-degenerate  hyperelliptic curves here. To obtain smooth 
solutions, we use $\mathbf{d}=0$. 
A solution in genus 2 of the DS1$^{-}$ equation is studied on 
the curve with the branch points $-2\pm \mathrm{i}, -1\pm \mathrm{i}, 1\pm \mathrm{i}$ 
in Fig.~\ref{figDS1mg2}.

\begin{figure}[htb!]
\begin{center}
  \includegraphics[width=0.45\textwidth]{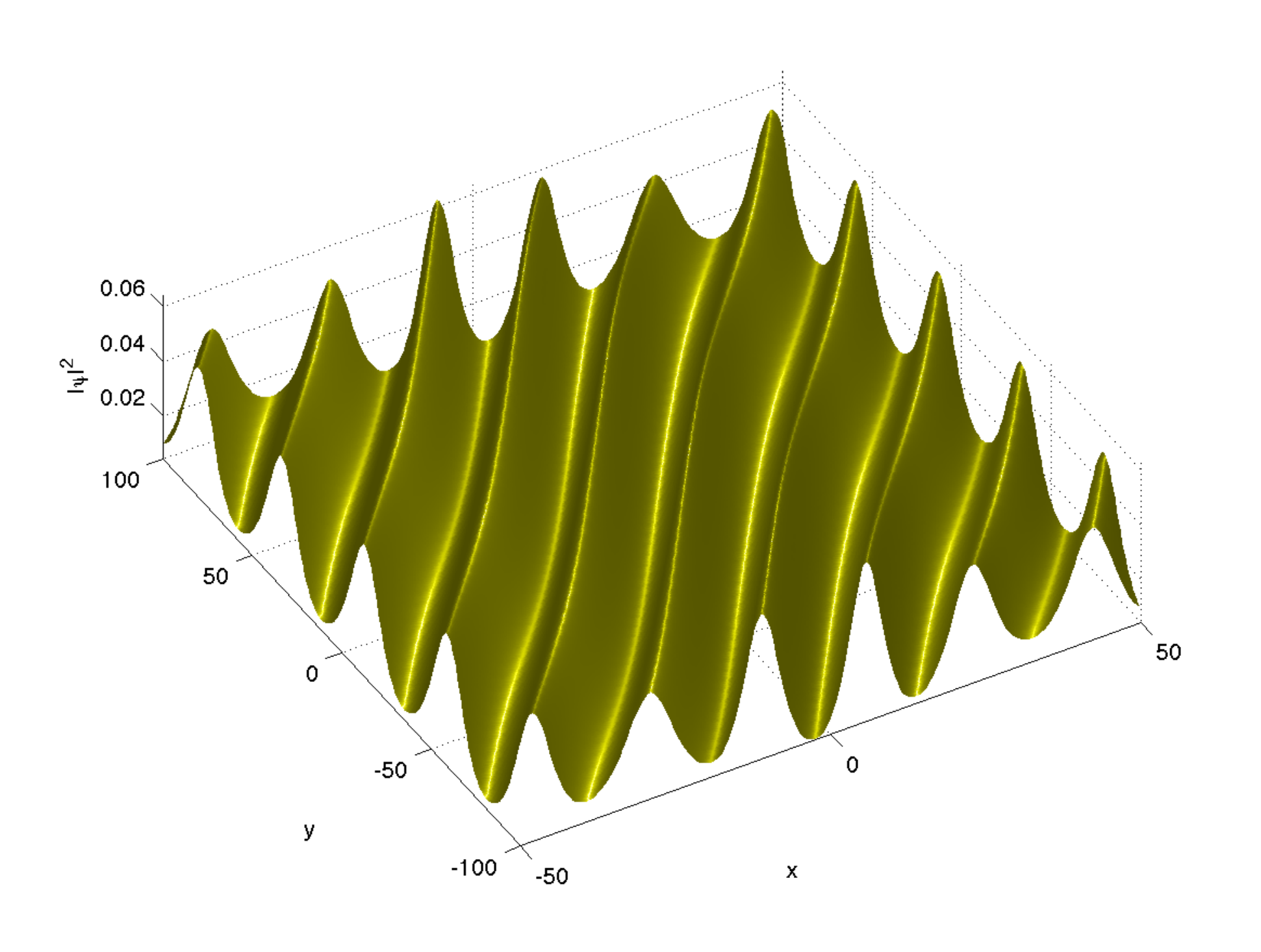}
  \includegraphics[width=0.45\textwidth]{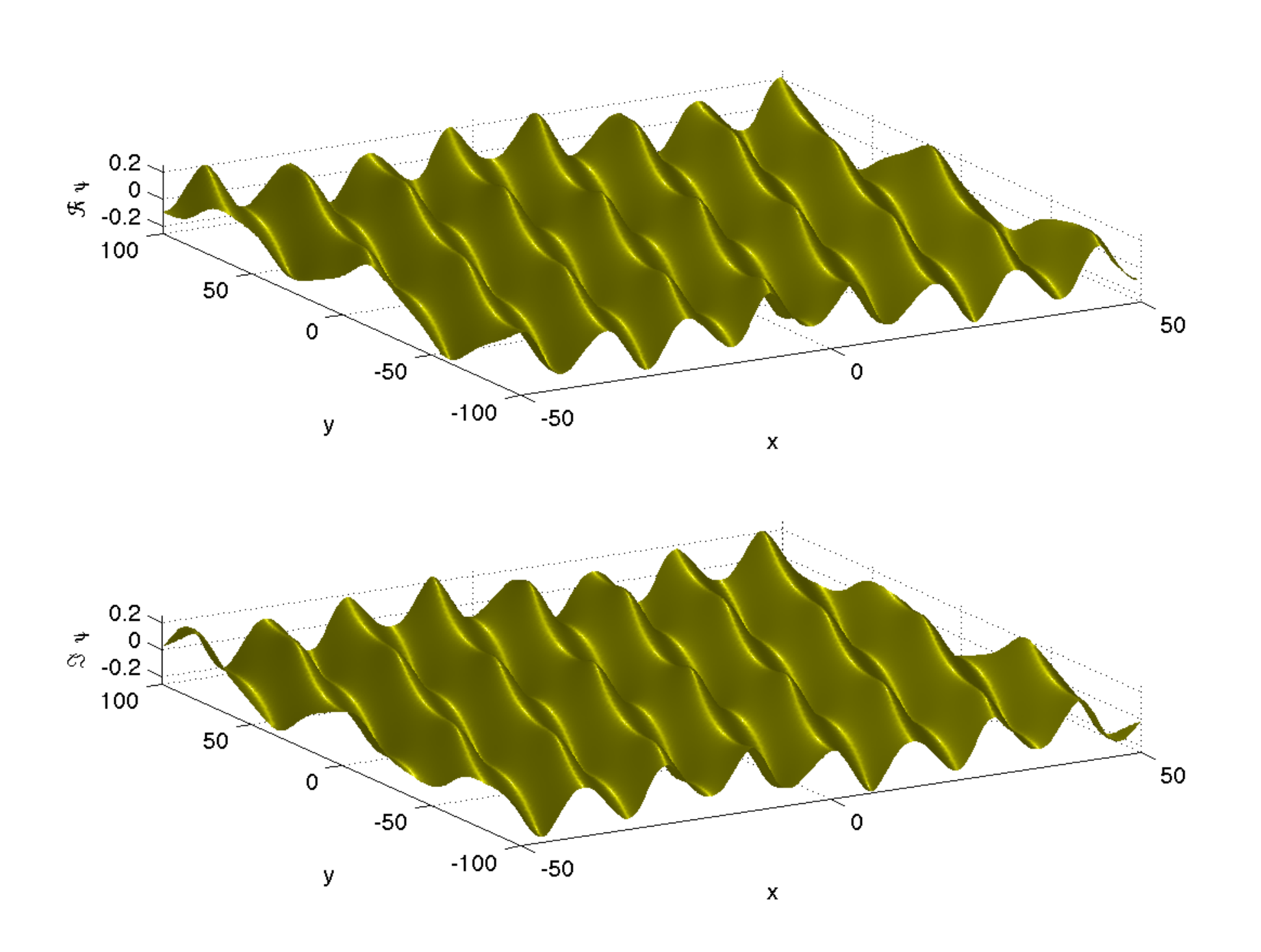}
\end{center}
 \caption{\textit{Solution (\ref{psi DS}) to the DS1$^{-}$ equation at $t=0$ on a hyperelliptic curve of 
 genus 2 with branch points $-2\pm \mathrm{i}$, $-1\pm \mathrm{i}$, $1\pm \mathrm{i}$ and 
 $a=(-4)^{(1)}$, $b=(-3)^{(2)}$. The square modulus of the 
 solution is shown on the left, real and imaginary parts on the right.}}
   \label{figDS1mg2}
\end{figure}

A typical example of 
a DS1$^{-}$ solution  on a 
 hyperelliptic curve of  genus 4 with branch points $-2\pm \mathrm{i}, -1\pm \mathrm{i}, \pm \mathrm{i}, 1\pm \mathrm{i}, 2\pm \mathrm{i}$ 
is shown in Fig.~\ref{figDS1mg4}. 
\begin{figure}[htb!]
\begin{center}
  \includegraphics[width=0.6\textwidth]{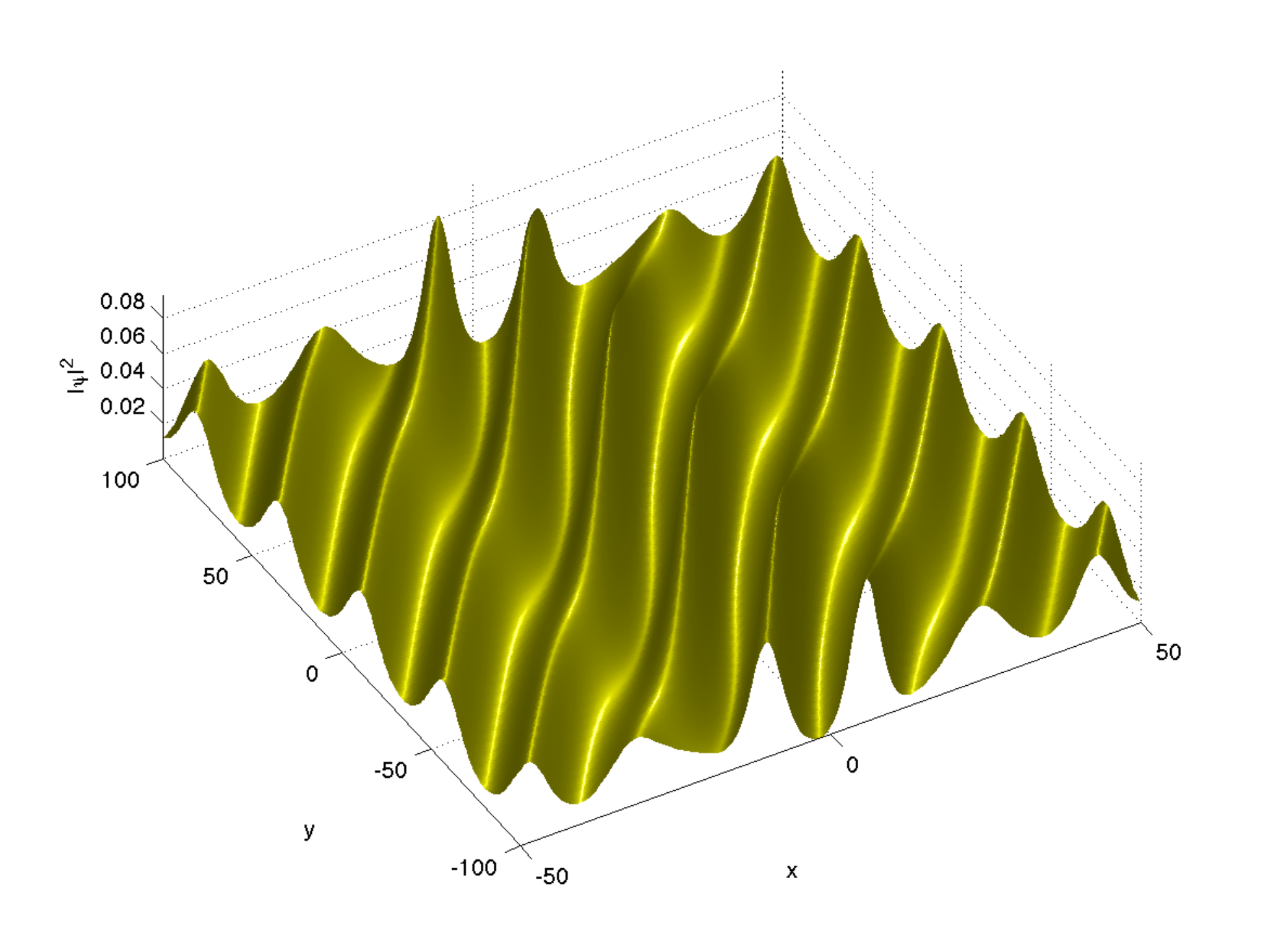}
\end{center}
 \caption{\textit{Solution (\ref{psi DS}) to the DS1$^{-}$ equation at $t=0$ on a hyperelliptic curve of 
 genus 4 with branch points $-2\pm \mathrm{i}$, $-1\pm \mathrm{i}$, $\pm \mathrm{i}$, $1\pm \mathrm{i}$, 
 $2\pm \mathrm{i}$ and 
 $a=(-4)^{(1)}$, $b=(-3)^{(2)}$.}}
   \label{figDS1mg4}
\end{figure}


Smooth solutions to DS2$^{-}$ can be obtained on  Riemann surfaces without 
real oval for points $a$ and $b$ satisfying $\tau a =b$. 
As mentioned above, hyperelliptic curves of the form  
$\mu^{2}=-\prod_{i=1}^{2g+2}(\lambda-\lambda_{i})$ with pairwise 
conjugate branch points have no real oval
for the standard involution $\tau_{1}$ as defined in Example 2.1. On the other hand, surfaces defined 
by the algebraic equation 
$\mu^{2}=\prod_{i=1}^{2g+2}(\lambda-\lambda_{i})$ have no real oval 
for the involution 
$\tau_{2}$ (see Example 2.1). 
We will consider here $\tau_{2}$ for the same curves as for DS1$^{-}$. An example for genus 
2 can be seen in Fig.~\ref{figDS2mg2}. An example for a DS2$^{-}$ solution of genus 4 can be seen in 
Fig.~\ref{figDS2mg4}.

\begin{figure}[htb!]
\begin{center}
  \includegraphics[width=0.6\textwidth]{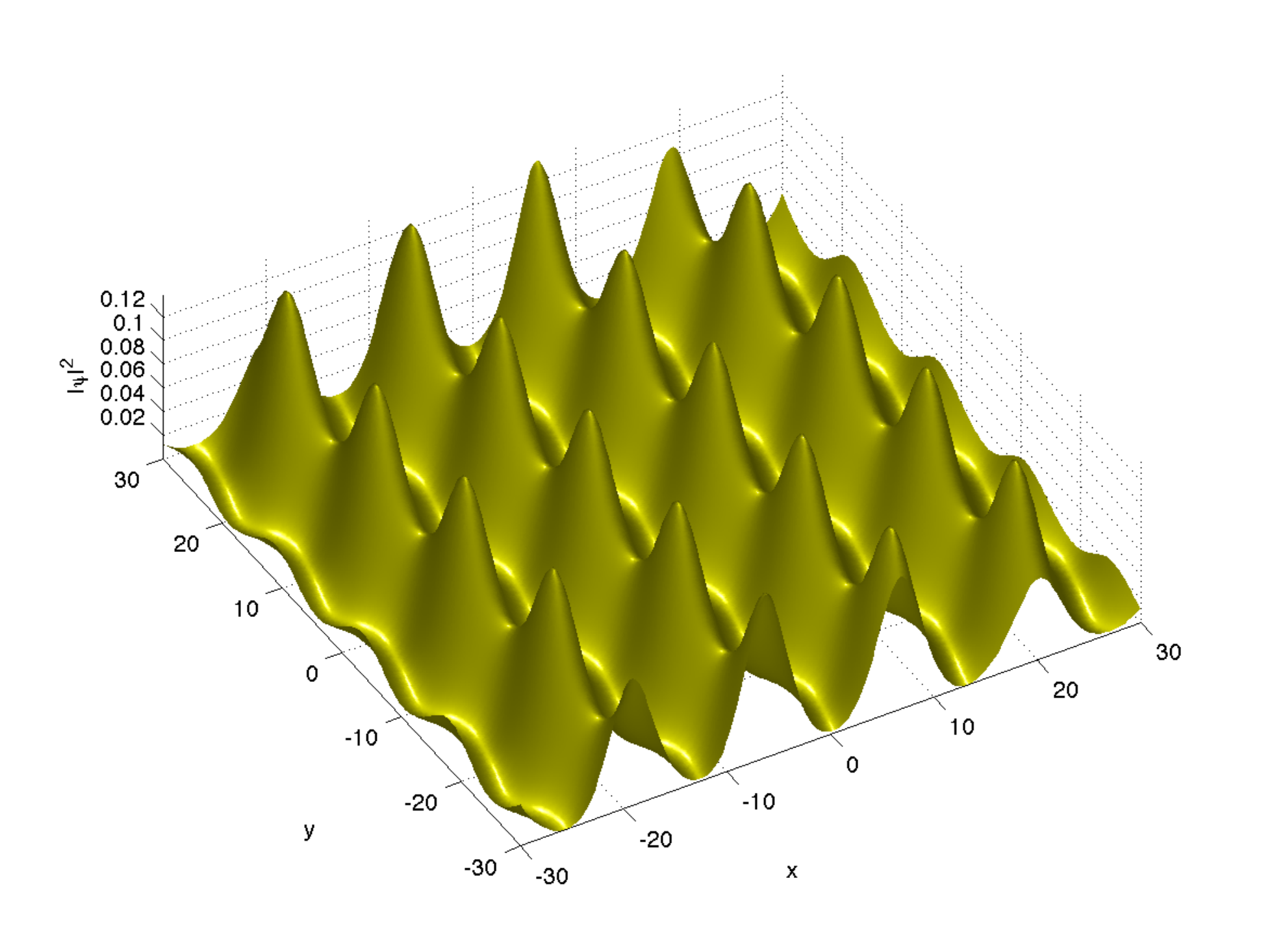}
\end{center}
 \caption{\textit{Solution to the DS$2^{-}$ equation at $t=0$ on a hyperelliptic curve of 
 genus 2 with branch points $-2\pm \mathrm{i}$, $-1\pm \mathrm{i}$, $1\pm \mathrm{i}$ and 
 $a=(-1.5+2\mathrm{i})^{(1)}$, $b=(-1.5-2\mathrm{i})^{(2)}$.} }
   \label{figDS2mg2}
\end{figure}

\begin{figure}[htb!]
\begin{center}
\includegraphics[width=0.9\textwidth]{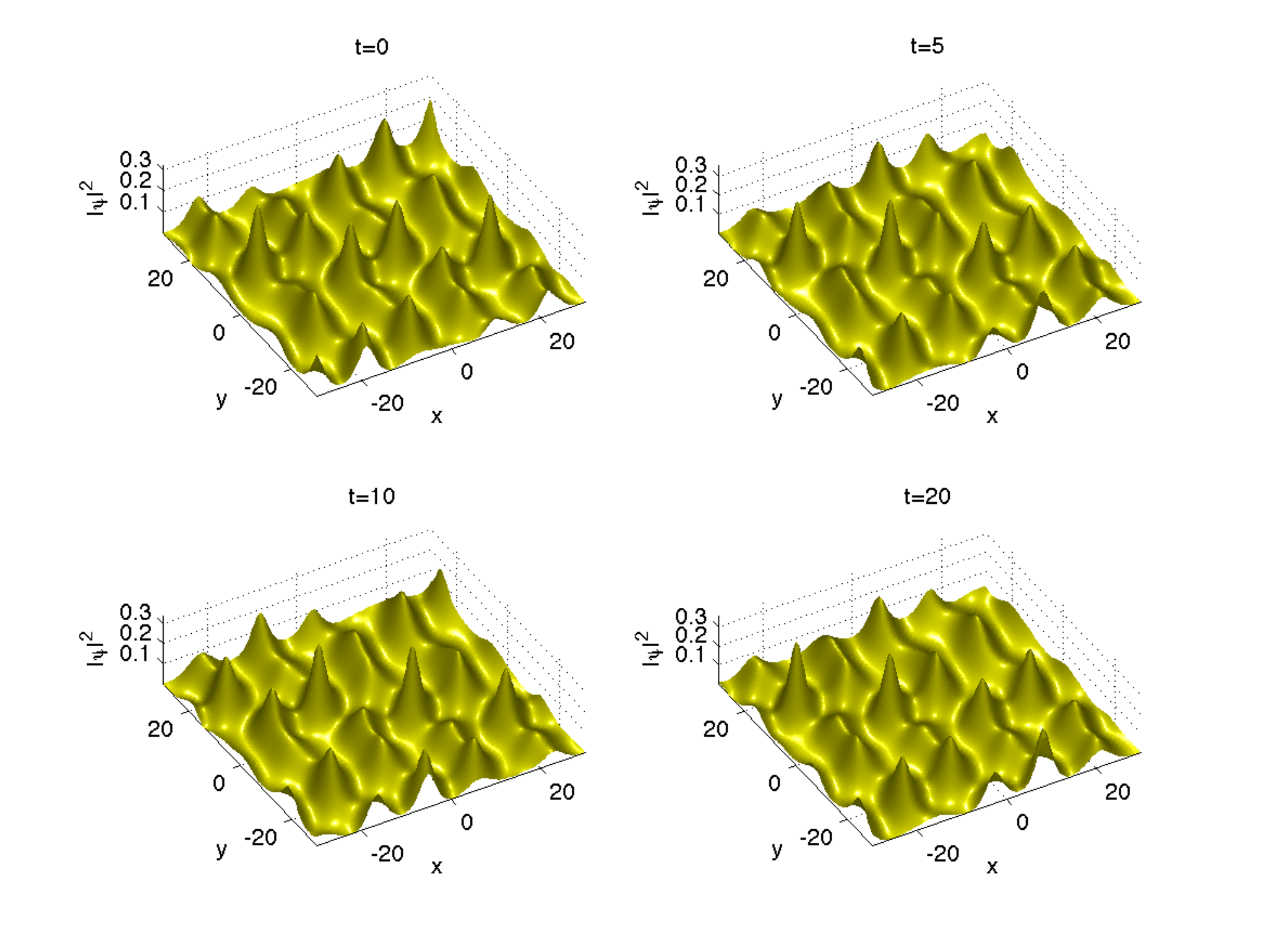}
\end{center}
 \caption{\textit{Solution to the DS$2^{-}$ equation for several values of $t$ on a hyperelliptic curve of 
 genus 4 with branch points $-2\pm \mathrm{i}$, $-1\pm \mathrm{i}$, $\pm \mathrm{i}$, $1\pm \mathrm{i}$, 
 $2\pm \mathrm{i}$ and 
 $a=(-1.5+2\mathrm{i})^{(1)}$, $b=(-1.5-2\mathrm{i})^{(2)}$.} }
   \label{figDS2mg4}
\end{figure}

\subsection{Solutions to the $n$-NLS$^{s}$ equation}

A straightforward way to obtain solutions (\ref{sol n-NLS}) to the 
$n$-NLS$^{s}$ equation is given on an $(n+1)$-sheeted branched covering of the complex plane, 
an approach that will be studied in more detail in the next section. 
As can be seen from the proof of Theorem 4.1 in \cite{Kalla}, the 
crucial point in the construction of these solutions is the fact that 
$\sum_{k=1}^{n+1}\mathbf{V}_{a_{k}}=0$. This implies that it is also possible 
to construct theta-functional $n$-NLS$^{s}$ solutions on hyperelliptic 
surfaces by introducing constants $\gamma_{k}$ via 
$\sum_{k=1}^{n+1}\gamma_{k}\mathbf{V}_{a_{k}}=0$ in the following corollary of Theorem 4.1 in \cite{Kalla}:

\begin{corollary} 
Let $\Rs_{g}$  be a real hyperelliptic curve of genus $g>0$ and denote by $\tau$ an 
anti-holomorphic involution. Choose the canonical homology basis 
which satisfies (\ref{hom basis}). Take $n\geq g$ and let 
$a_{1},\ldots,a_{n+1}\in\Rs_{g}(\R)$ not ramification points having distinct projection 
$\lambda(a_{j}), \, j=1,\ldots,n+1$, onto the $\lambda$-sphere.
Denote by $\ell_{j}$ an oriented contour between $a_{n+1}$ and $a_{j}$ which does not intersect cycles of the canonical homology basis. Let $\mathbf{d}_{R}\in\R^{g}$, $\mathbf{T}\in\Z^{g}$, and define
$ \mathbf{d}=\mathbf{d}_{R}+\frac{\mathrm{i}\pi}{2}(\text{diag}(\mathbb{H})-2\,\mathbf{T}). $
Take $\theta\in\R$ and let $\gamma_{g+1},\ldots,\gamma_{n}\in\R$ be 
arbitrary constants with $\gamma_{n+1}=1$. Put 
$\hat{s}=(\mbox{sign}(\gamma_{1})\,s_{1},\ldots,\mbox{sign}(\gamma_{n})\,s_{n})$ where $s_{j}$ is  given in (\ref{sj}), and the scalars  $\gamma_{j},\,j=1,\ldots,g$, are defined by
$\sum_{k=1}^{n+1}\gamma_{k}\mathbf{V}_{a_{k}}=0$.
Then  the following functions 
$\psi_{j}, \, j=1,\ldots,n$, give solutions of the $n$-NLS$^{\hat{s}}$  
equation (\ref{n-NLS})
\begin{equation}    
\psi_{j}(x,t)=|\gamma_{j}|^{1/2}\,|A_{j}|\,e^{\mathrm{i}\theta}\,\frac{\Theta(\mathbf{Z}-\mathbf{d}+\mathbf{r}_{j})}{\Theta(\mathbf{Z}-\mathbf{d})}\,\exp\{ -\mathrm{i}\,(E_{j}\,x-F_{j}\,t)\},\label{sol hyp $n$-NLS^{s}}
\end{equation}
where $|A_{j}|=|q_{2}(a_{n+1},a_{j})|^{1/2}\,\exp\left\{\tfrac{1}{2}\left\langle\mathbf{d},\mathbf{M}_{j}\right\rangle\right\}$. Here  $\mathbf{Z}=\mathrm{i}\,\mathbf{V}_{a_{n+1}}\,x+\mathrm{i}\,\mathbf{W}_{a_{n+1}}\,t$, where
the vectors $\mathbf{V}_{a_{n+1}}$ and $\mathbf{W}_{a_{n+1}}$ were introduced in (\ref{exp hol diff}), and $\mathbf{r}_{j}=\int_{\ell_{j}}\omega$. The scalars $E_{j},F_{j}$ are given by
\[E_{j}=K_{1}(a_{n+1},a_{j}), \qquad F_{j}=K_{2}(a_{n+1},a_{j})-2\sum_{k=1}^{n}\gamma_{k}\,q_{1}(a_{n+1},a_{k}),\]
where $q_{i},K_{i}$ for $i=1,2$ are defined in (\ref{q1})-(\ref{K2}).
If $\Rs_{g}$ is dividing and if $\mathbf{d}\in\R^{g}$, functions (\ref{sol hyp $n$-NLS^{s}})  give smooth solutions of $n$-NLS$^{\hat{s}}$. \label{theo hyp sol n-NLS}
\end{corollary}

As an example we consider, as for DS in genus $2$, the family of curves with the branch 
points $-2,-1,0,\epsilon,2,2+\epsilon$ for  $\epsilon=1$ 
and $\epsilon=10^{-10}$. In the former case the solutions will be 
periodic in the ($x,t$)-plane, in the latter almost solitonic. To obtain non-trivial solutions in the solitonic limit, 
we use $\mathbf{d}=\frac{1}{2}
\left[\begin{smallmatrix}
    1 & 1  \\
    0 & 0
\end{smallmatrix}\right]^{t}
$ in all examples.

In Fig.~\ref{fignlsh2mm} we show the case $a_{1}=(-1.9)^{(1)}$, 
$a_{2}=(-1.1)^{(1)}$ 
and $a_{3}=(-1.8)^{(1)}$, which leads to  a  solution of 2-NLS$^{\hat{s}}$ with $\hat{s}=(-1,-1)$. 
Interchanging $a_{2}$ and $a_{3}$ in the above example, we obtain a 
solution to 2-NLS$^{\hat{s}}$ with $\hat{s}=(1,-1)$ in Fig.~\ref{fignlsh2pm}.

\begin{figure}[htb!]
\begin{center}
\includegraphics[width=0.45\textwidth]{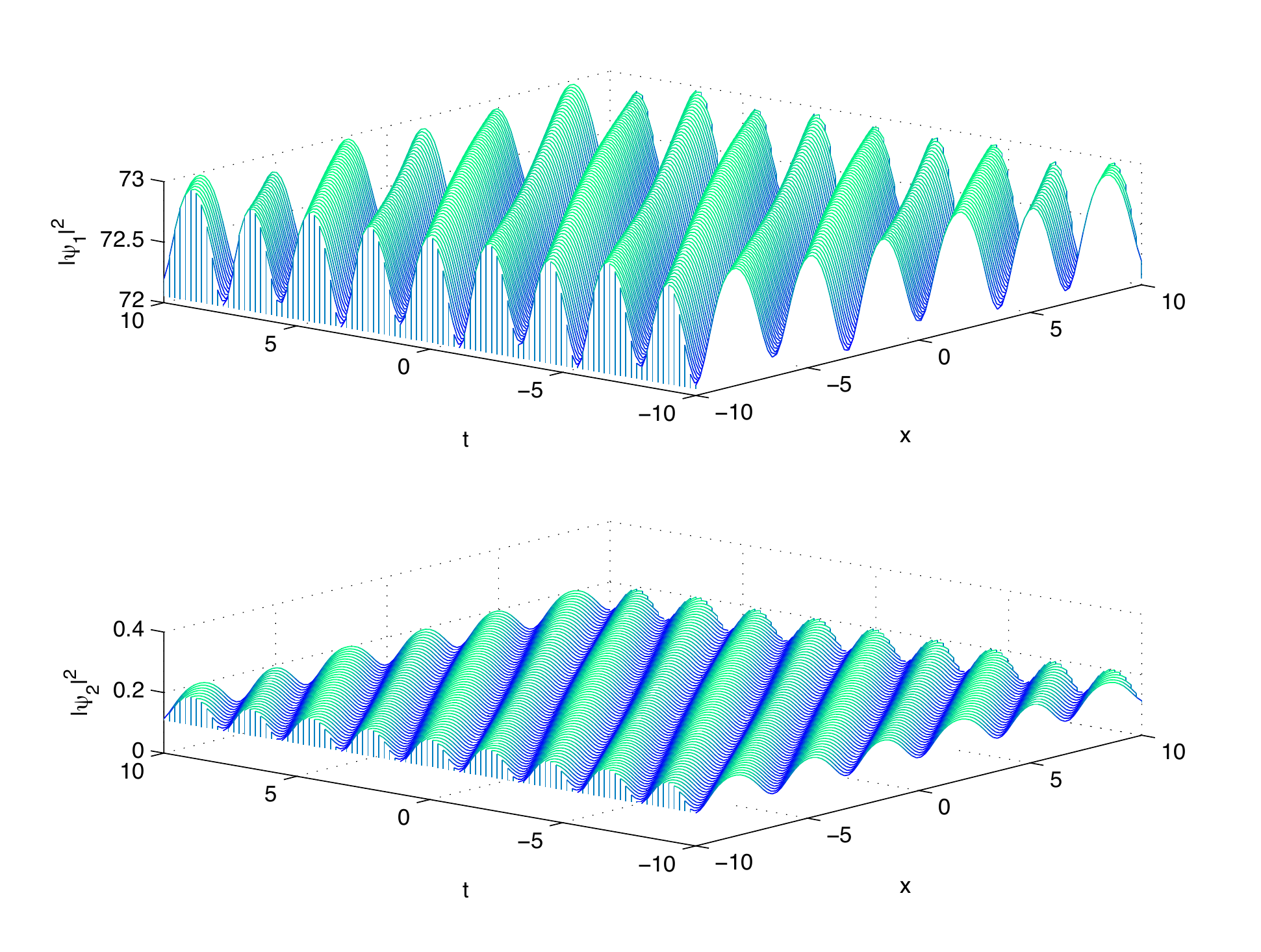}
\includegraphics[width=0.45\textwidth]{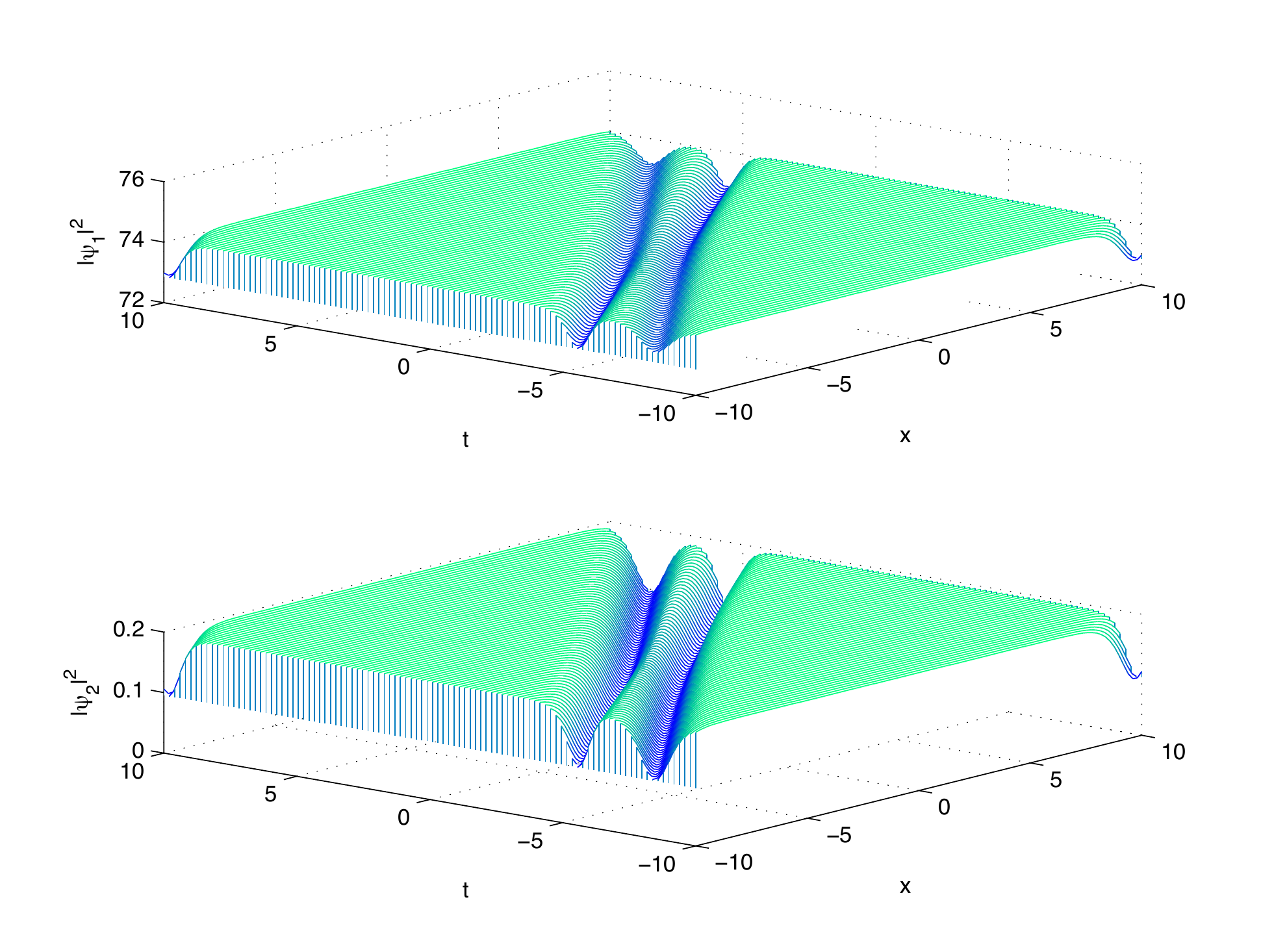}
\end{center}
 \caption{\textit{Solution (\ref{sol hyp $n$-NLS^{s}}) to the 2-NLS$^{\hat{s}}$ equation with $\hat{s}=(-1,-1)$ on a hyperelliptic curve of 
 genus 2 with branch points $-2,-1,0,\epsilon,2,2+\epsilon$ and 
 $a_{1}=(-1.9)^{(1)}$,  $a_{2}=(-1.1)^{(1)}$ 
 and $a_{3}=(-1.8)^{(1)}$ for $\epsilon=1$ on the left and 
 $\epsilon=10^{-10}$, the almost solitonic limit, on the right.}}
   \label{fignlsh2mm}
\end{figure}

Solutions  of 4-NLS$^{\hat{s}}$ can be studied in the same way on the  hyperelliptic curve of genus 4
with branch points 
$-4,-3,-2,-2+\epsilon,0,\epsilon,2,2+\epsilon,4,4+\epsilon$. We use 
$\mathbf{d}=\frac{1}{2}
\left[\begin{smallmatrix}
    1 & 1 & 1 & 1 \\
    0 & 0 & 0 & 0
\end{smallmatrix}\right]^{t}
$ and the points $a_{1}=(-3.9)^{(1)}$, $a_{2}=(-3.7)^{(1)}$, 
$a_{3}=(-3.5)^{(1)}$, 
$a_{4}=(-3.3)^{(1)}$ and $a_{5}=(-3.1)^{(1)}$. Since the vectors $\mathbf{V}_{a_{j}}$ and 
$\mathbf{W}_{a_{j}}$ are very 
similar in this case, the same is true for the  functions
$\psi_{j}$. Therefore, we will only show the square modulus of the first component  $\psi_{1}$ in 
Fig.~\ref{fignlsh4} for $\hat{s}=(1,-1,1,-1)$ on the left. Interchanging $a_{4}$ 
and $a_{5}$ in this case, one gets a solution to 4-NLS$^{\hat{s}}$ with 
$\hat{s}=(-1,1,-1,-1)$ which can be seen on the right of 
Fig.~\ref{fignlsh4}. 
\begin{figure}[htb!]
\begin{center}
\includegraphics[width=0.45\textwidth]{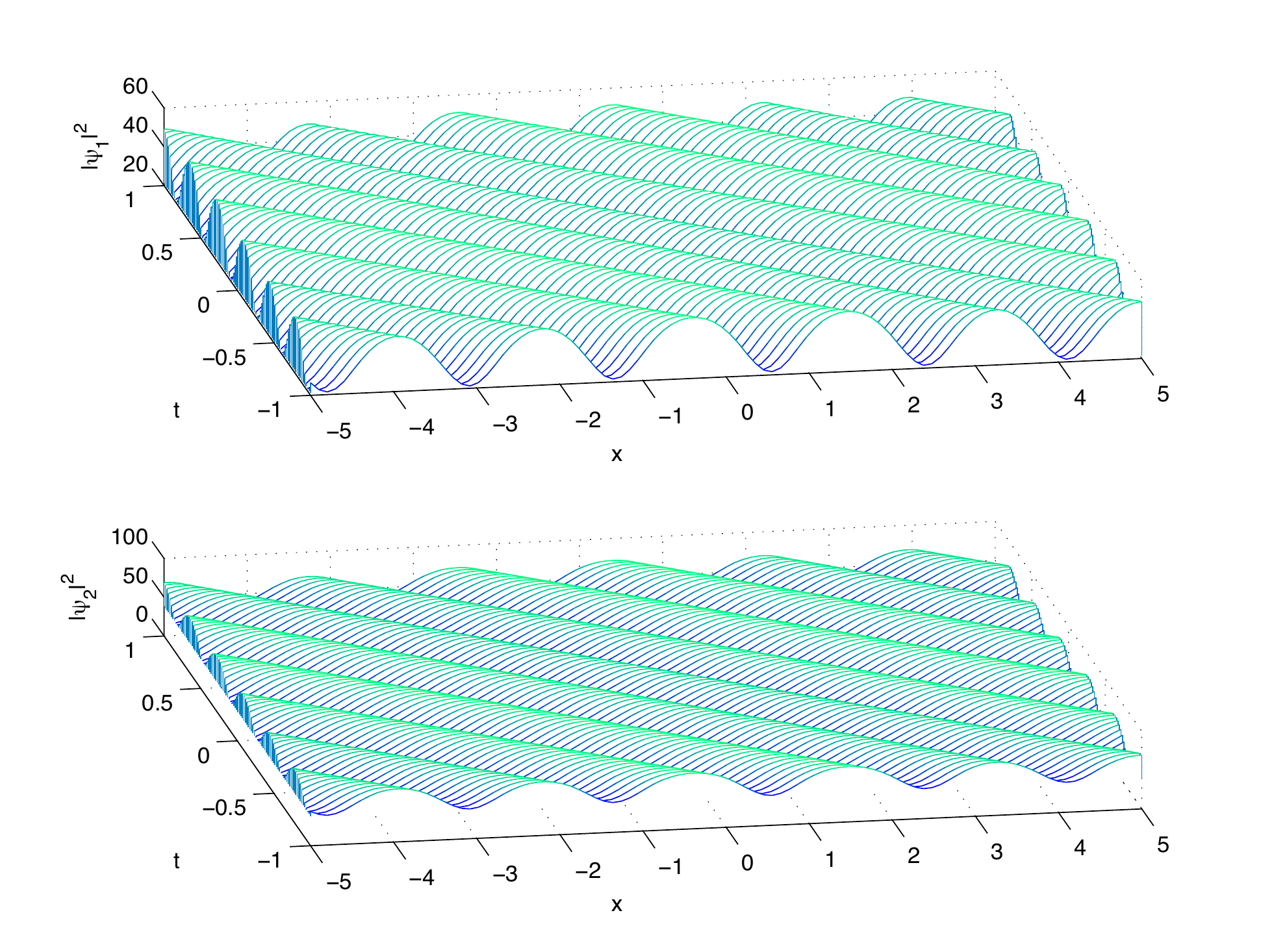}
\includegraphics[width=0.45\textwidth]{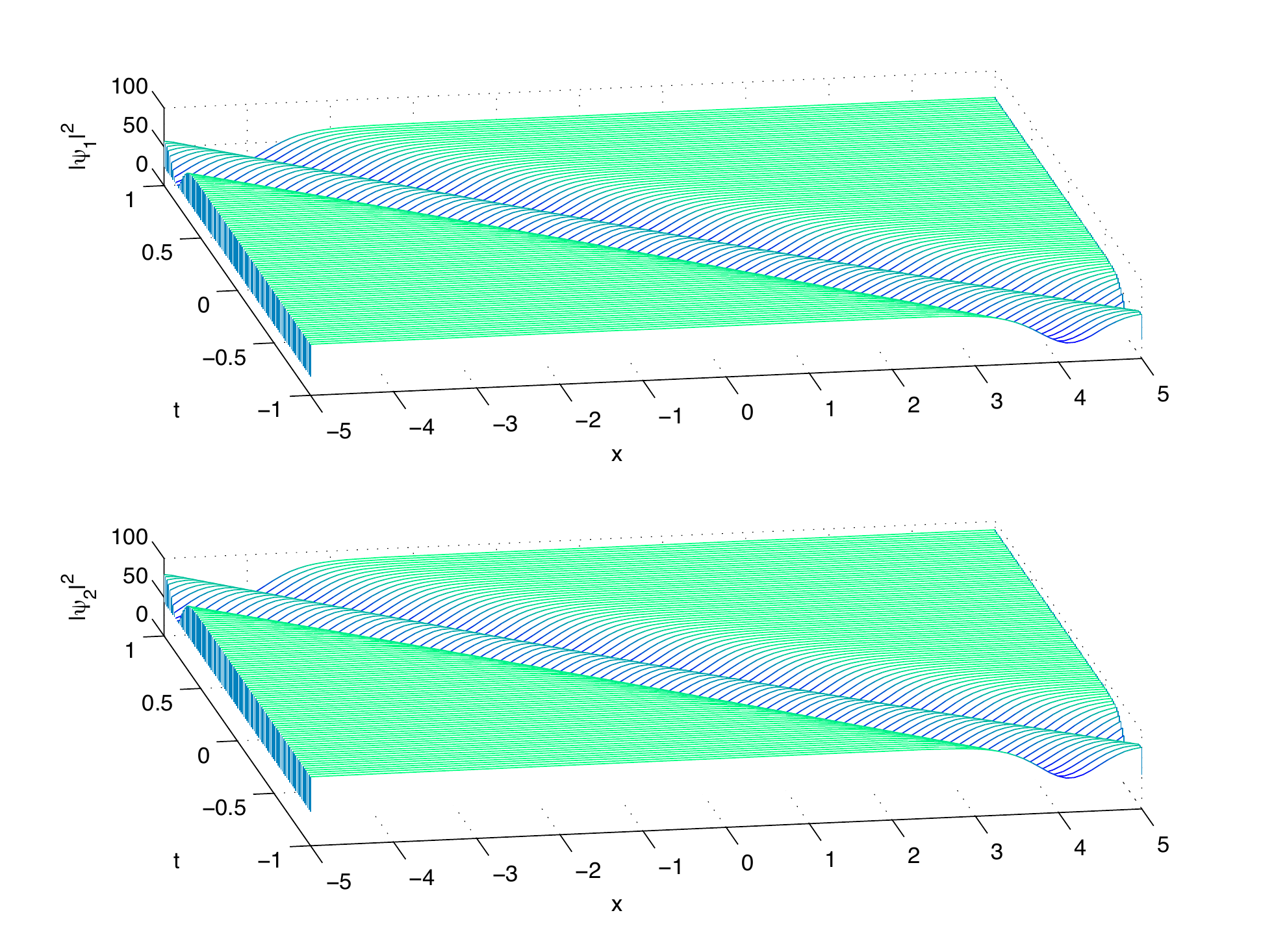}
\end{center}
 \caption{\textit{Solution (\ref{sol hyp $n$-NLS^{s}}) to the 2-NLS$^{\hat{s}}$ equation with $\hat{s}=(1,-1)$ on a hyperelliptic curve of 
 genus 2 with branch points $-2,-1,0,\epsilon,2,2+\epsilon$ and 
 $a_{1}=(-1.9)^{(1)}$, $a_{2}=(-1.8)^{(1)}$ 
 and $a_{3}=(-1.1)^{(1)}$ for $\epsilon=1$ on the left and 
 $\epsilon=10^{-10}$, the almost solitonic limit, on the right.}}
   \label{fignlsh2pm}
\end{figure}
\begin{figure}[htb!]
\begin{center}
\includegraphics[width=0.45\textwidth]{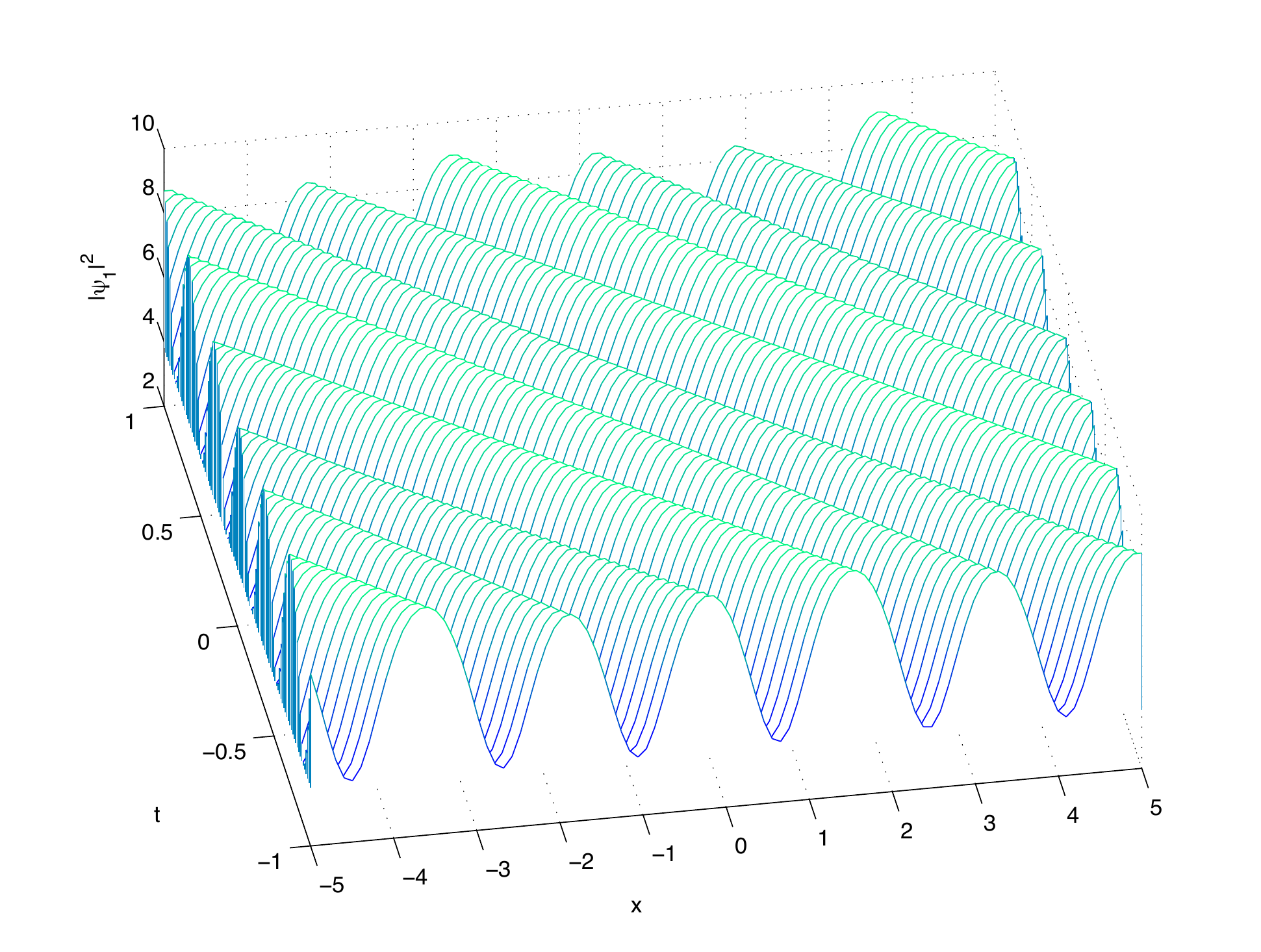}
\includegraphics[width=0.45\textwidth]{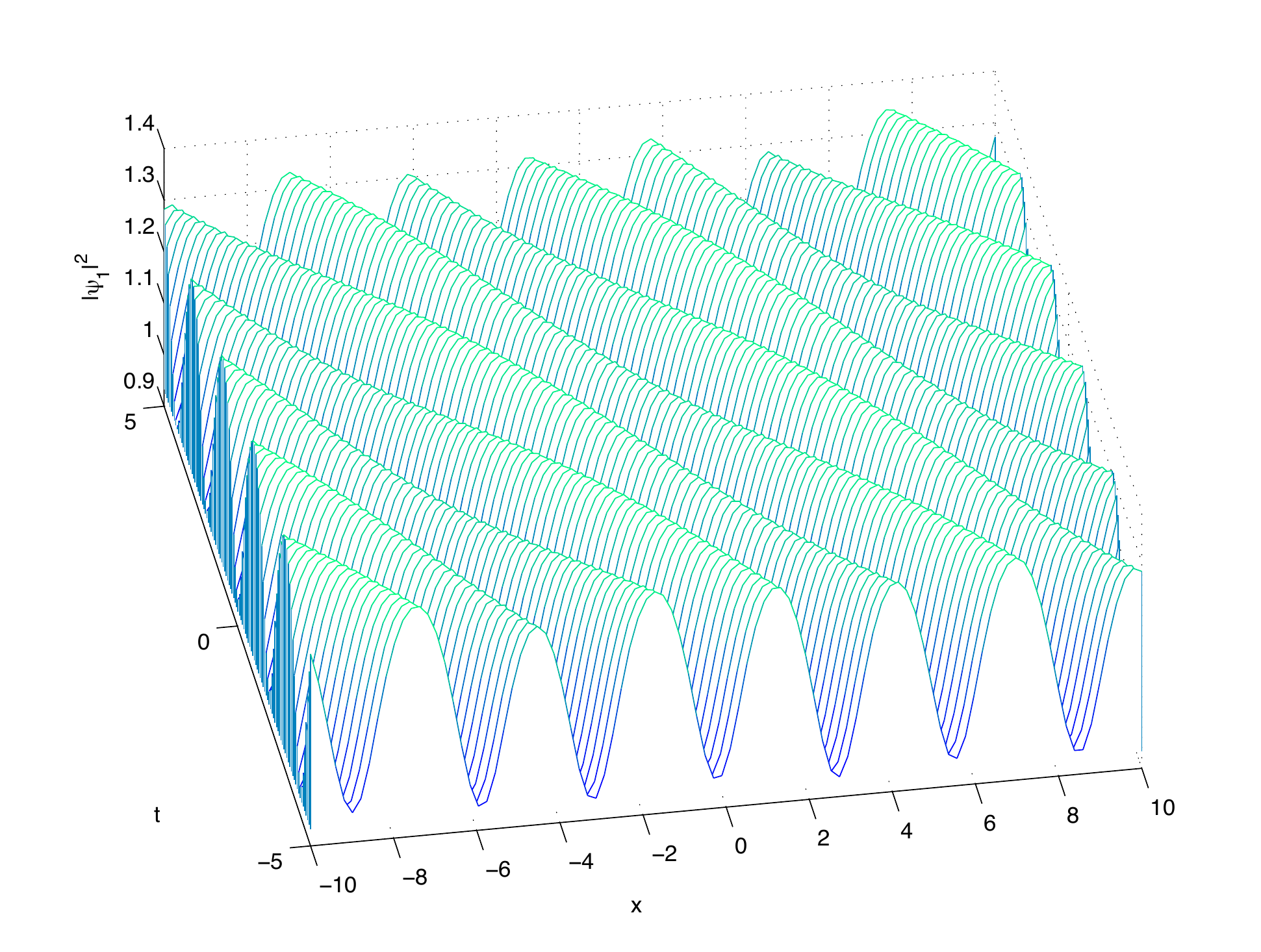}
\end{center}
 \caption{\textit{Solution to the 4-NLS$^{\hat{s}}$ equation  on a hyperelliptic curve of 
 genus 4 with branch points $-4,-3,-2,-2+\epsilon,0,\epsilon,2,2+\epsilon,4,4+\epsilon$ and 
 $\epsilon=1$ for  $a_{1}=(-3.9)^{(1)}$, $a_{2}=(-3.7)^{(1)}$, 
 $a_{3}=(-3.5)^{(1)}$, 
$a_{4}=(-3.3)^{(1)}$ and $a_{5}=(-3.1)^{(1)}$, which leads to  
$\hat{s}=(1,-1,1,-1)$, on 
the left, and for  $a_{1}=(-3.9)^{(1)}$, $a_{2}=(-3.7)^{(1)}$, 
$a_{3}=(-3.5)^{(1)}$, 
$a_{4}=(-3.1)^{(1)}$ and $a_{5}=(-3.3)^{(1)}$, which leads to  $\hat{s}=(-1,1,-1,-1)$ on the right.}}
 \label{fignlsh4}
\end{figure}
The almost solitonic limit $\epsilon=10^{-10}$ 
produces well known solitonic patterns as shown for instance for the 
DS equation in the previous subsection.

Hyperelliptic solutions to the $n$-NLS$^{\hat{s}}$ equation with all  signs $\hat{s}_{j}$ satisfying $\hat{s}_{j}=1$, can 
be constructed on a   curve without real branch points.  To obtain smooth 
solutions, we use $\mathbf{d}=0$. 
A solution of the 2-NLS$^{\hat{s}}$ equation is studied on 
the curve  of genus 2 with the branch points $-2\pm \mathrm{i}, -1\pm \mathrm{i}, 1\pm \mathrm{i}$ 
in Fig.~\ref{fignlsh2pp}.

\begin{figure}[htb!]
\begin{center}
  \includegraphics[width=0.6\textwidth]{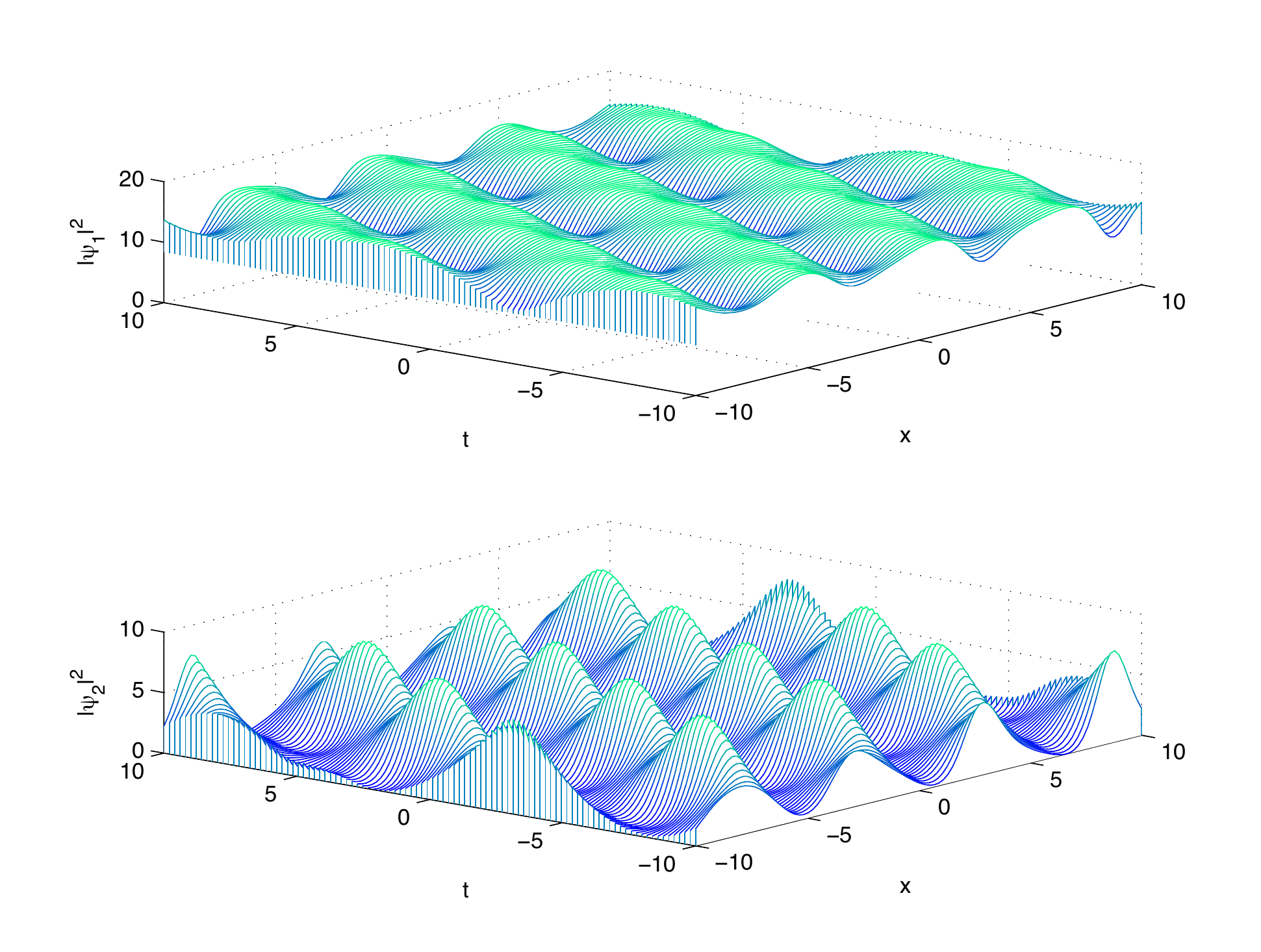}
\end{center}
 \caption{\textit{Solution to the 2-NLS$^{\hat{s}}$ equation with $\hat{s}=(1,1)$ on a hyperelliptic curve of 
 genus 2 with branch points $-2\pm \mathrm{i}$, $-1\pm \mathrm{i}$, $1\pm \mathrm{i}$ and 
 $a_{1}=(-1.9)^{(1)}$, $a_{2}=(-1.8)^{(2)}$ and $a_{3}=(-1.1)^{(1)}$.}}
   \label{fignlsh2pp}
\end{figure}

A typical example for a hyperelliptic 4-NLS$^{\hat{s}}$ solution with $\hat{s}=(1,1,1,1)$  can be obtained on a 
 curve  of
genus 4 with branch points $-2\pm \mathrm{i}, -1\pm \mathrm{i}, \pm \mathrm{i}, 1\pm \mathrm{i}, 2\pm \mathrm{i}$,
as shown in Fig.~\ref{fignlsh4p}. 
\begin{figure}[htb!]
\begin{center}
  \includegraphics[width=0.6\textwidth]{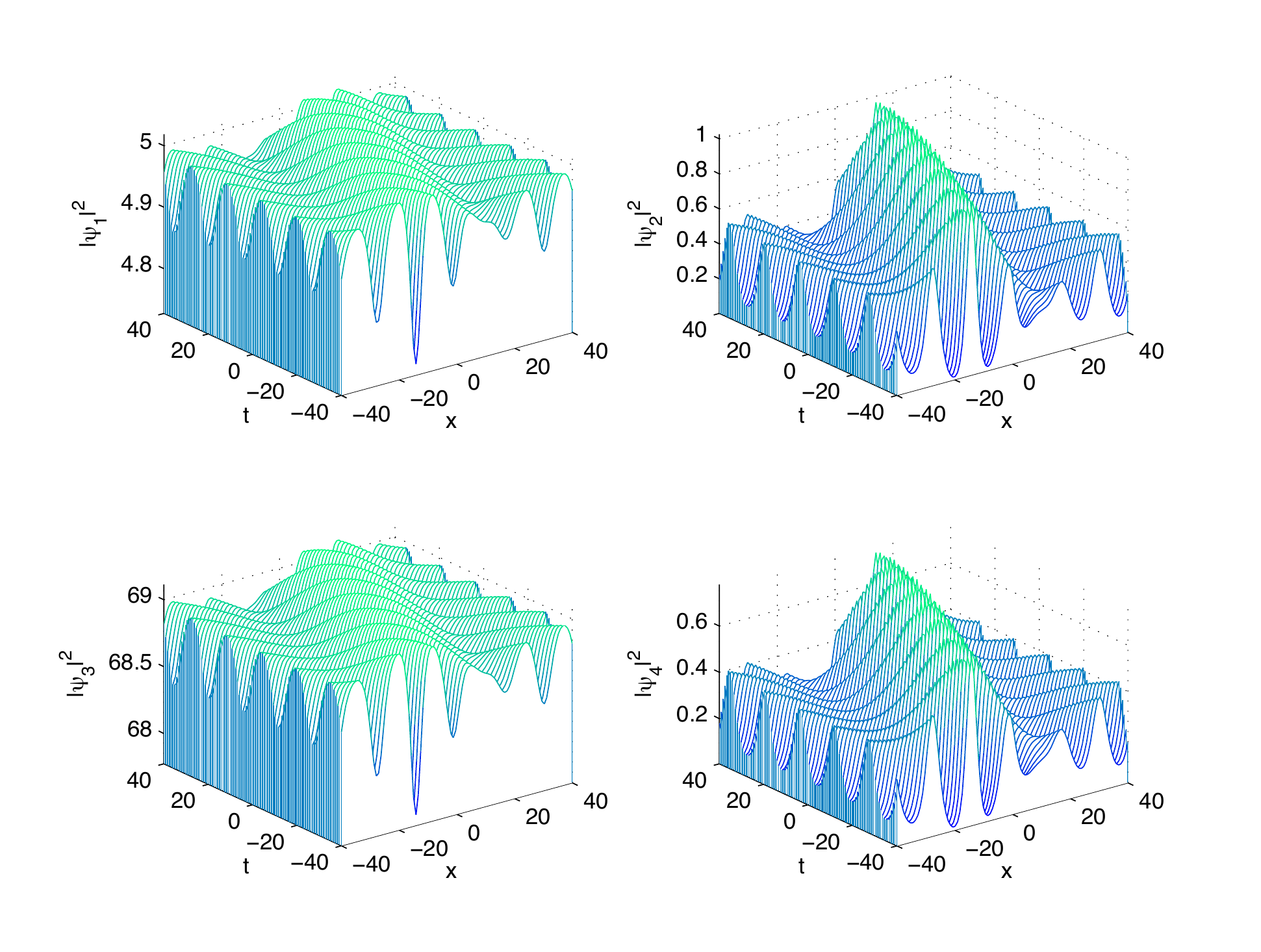}
\end{center}
 \caption{\textit{Solution to the 4-NLS$^{\hat{s}}$  equation  with $\hat{s}=(1,1,1,1)$ on a hyperelliptic curve of 
 genus 4 with branch points $-2\pm \mathrm{i}$, $-1\pm \mathrm{i}$, $\pm \mathrm{i}$, $1\pm \mathrm{i}$, 
 $2\pm \mathrm{i}$ and 
 $a_{1}=(-3.9)^{(1)}$, $a_{2}=(-3.7)^{(2)}$, $a_{3}=(-3.5)^{(1)}$, 
 $a_{4}=(-3.3)^{(2)}$ and $a_{5}=(-3.1)^{(1)}$.}}
   \label{fignlsh4p}
\end{figure}

\section{General real algebraic curves}

The quantities entering theta-functional solutions of the DS and $n$-NLS$^{s}$ equations are related to compact Riemann surfaces. Since 
all compact Riemann surfaces can be defined via compactified non-singular 
algebraic curves, convenient computational approaches as 
\cite{deco1, deco2} and \cite{FK} are based on algebraic curves: 
differentials, homology basis and periods of the Riemann surface can 
be obtained in an algorithmic way. We refer the reader to the cited 
literature for details. The identification  of the sheets of the 
covering   defined by the algebraic curve (\ref{algcur}) via the projection map $(x,y) \mapsto x$, is 
done, as in the hyperelliptic case, by analytic continuation of the 
roots $y_{i}$, $i=1,\ldots,N$  for some  non-critical point 
$x_{b}$ on the 
$x$-sphere,  along a set of 
contours specified in \cite{FK}. In the context of real  algebraic curves for 
which solutions of 
$n$-NLS$^{s}$ and DS are discussed here, an additional problem is to 
establish the action of the anti-holomorphic involution $\tau$ on 
points on different sheets. A typical problem is to find points $a\in\Rs_{g}$ 
and $b\in\Rs_{g}$ with the same projection onto the $x$-sphere such 
that $\tau a=b$; here $\tau$ is defined  via $\tau a = 
(\overline{x}(a),\overline{y}(a))$. To this end, the roots $y_{i}$, $i=1,\ldots,N,$ 
identified at $x=x_{b}$, are analytically continued to the points 
projecting to $x(a)$ on the $x$-sphere. It is then established which 
pairs of points in the different sheets satisfy  $\tau a=b$.

In contrast to the hyperelliptic curves of 
the previous section, it is not possible for general curves to 
introduce a priori a basis of the homology. Thus the cited codes use 
an algorithm by Tretkoff and Tretkoff \cite{tretalg} which produces a 
homology basis for a given branching structure of the covering which is in general not adapted to possible 
automorphisms of the curve. In the context of theta-functional 
solutions to integrable PDEs one is often interested in real curves. As 
discussed in \cite{Kalla},  the Vinnikov basis (i.e., the canonical homology basis  which satisfies (\ref{hom basis})) is 
convenient in this context. Since solutions and smoothness conditions for  
 $n$-NLS$^{s}$ and DS equations are formulated in this basis, a symplectic 
transformation relating the computed basis to the  Vinnikov basis needs 
to be worked out. This transformation is discussed in the present section and will 
be applied to examples of real algebraic curves.

\subsection{Symplectic transformation}

Let $\Rs_{g}$ be a real compact Riemann surface of genus $g$ and $\tau$ an anti-holomorphic involution defined on it.
Let $(\nu_{1},\ldots,\nu_{g})$ be a basis of holomorphic differentials such that
\begin{equation}
\overline{\tau^{*}\nu_{j}}=\nu_{j}, \qquad j=1,\ldots,g,\label{diff mu}
\end{equation}
where $\tau^{*}$ is the action of $\tau$ lifted to the space of 
holomorphic differentials: $\tau^* \omega(p) = \omega(\tau p)$ for 
any $p\in \mathcal{R}_{g}$. 
For an arbitrary canonical homology basis $(\mathbf{\mathcal{A}},\mathbf{\mathcal{B}})$, let us denote by $P_{\mathcal{A}}$ and $P_{\mathcal{B}}$ the matrices of $\mathcal{A}$ and $\mathcal{B}$-periods of the differentials $\nu_{j}$:
\begin{equation}
(P_{\mathcal{A}})_{ij}=\int_{\mathcal{A}_{i}}\nu_{j}, \qquad (P_{\mathcal{B}})_{ij}=\int_{\mathcal{B}_{i}}\nu_{j},  \qquad i,j=1,\ldots,g.  \label{matrix P}
\end{equation}
 In what follows $(\mathbf{\mathcal{A}},\mathbf{\mathcal{B}})$ denotes the Vinnikov basis.
From (\ref{diff mu}) and (\ref{hom basis}) we deduce the action of the complex conjugation on the matrices $P_{\mathcal{A}}$ and  $P_{\mathcal{B}}$:
\begin{equation}
(P_{\mathcal{A}})_{ij}\in\R,  \label{PerA Vin}
\end{equation}
\begin{equation}
\overline{P_{\mathcal{B}}}=-P_{\mathcal{B}}+\mathbb{H}P_{\mathcal{A}}. \label{PerB Vin}
\end{equation}

Denote by 
$(\mathbf{\tilde{\mathcal{A}}},\mathbf{\tilde{\mathcal{B}}})$ the 
homology basis on $\Rs_{g}$ produced by the Tretkoff-Tretkoff algorithm. 
From the symplectic transformation (\ref{transf Vinn2}) we obtain the 
following transformation law between the matrices 
$P_{\tilde{\mathcal{A}}},P_{\tilde{\mathcal{B}}}$ and $P_{\mathcal{A}},P_{\mathcal{B}}$ defined in (\ref{matrix P}):
\begin{equation}
\left(\begin{matrix}
A&B\\
C&D
\end{matrix}\right)
\left(\begin{matrix}
P_{\tilde{\mathcal{A}}}\\
P_{\tilde{\mathcal{B}}}
\end{matrix}\right) 
=
\left(\begin{matrix}
P_{\mathcal{A}}\\
P_{\mathcal{B}}
\end{matrix}\right).  \label{transf per}
\end{equation}
Therefore, by (\ref{PerA Vin}) one gets
\begin{align}
A \,\text{Re}\left(P_{\tilde{\mathcal{A}}}\right)+B\, \text{Re}\left(P_{\tilde{\mathcal{B}}}\right)&=P_{\mathcal{A}}  \label{A,B Re}\\
A \,\text{Im}\left(P_{\tilde{\mathcal{A}}}\right)+B \,\text{Im}\left(P_{\tilde{\mathcal{B}}}\right)&=0,\label{A,B Im}
\end{align}
and by (\ref{PerB Vin})
\begin{align}
C \,\text{Re}\left(P_{\tilde{\mathcal{A}}}\right)+D\, \text{Re}\left(P_{\tilde{\mathcal{B}}}\right)&=\frac{1}{2}\,\mathbb{H}P_{\mathcal{A}}  \label{C,D Re}\\
C \,\text{Im}\left(P_{\tilde{\mathcal{A}}}\right)+D \,\text{Im}\left(P_{\tilde{\mathcal{B}}}\right)&=\text{Im}\left(P_{\mathcal{B}}\right).\label{C,D Im}
\end{align}  

According to (\ref{A,B Re}), the matrix $A 
\,\text{Re}\left(P_{\tilde{\mathcal{A}}}\right)+B\, 
\text{Re}\left(P_{\tilde{\mathcal{B}}}\right)$ is invertible, since 
the matrix $P_{\mathcal{A}}$ of $\mathcal{A}$-periods of a basis of 
holomorphic differentials is always invertible (see, for instance, \cite{Bob}). 
Moreover, it is well known that the Riemann matrix 
$\mathbb{B}=2\mathrm{i}\pi\,P_{\mathcal{B}}\,(P_{\mathcal{A}})^{-1}$ 
has a (negative) definite real part, which is equal to 
$-2\pi\,\text{Im}(P_{\mathcal{B}})\,\text{Im}((P_{\mathcal{A}})^{-1})$ for the real matrix $P_{\mathcal{A}}$ here. Then, by (\ref{C,D Im}) the matrix $C \,\text{Im}\left(P_{\tilde{\mathcal{A}}}\right)+D \,\text{Im}\left(P_{\tilde{\mathcal{B}}}\right)$ is also invertible.

\begin{lemma}
The matrices $A,B,C,D\in\M_{g}(\Z)$ solving (\ref{A,B Re})-(\ref{C,D Im}) satisfy:
\begin{align}
A^{t}&=\text{Im}\left(P_{\tilde{\mathcal{B}}}\right)\left[C 
\,\text{Im}\left(P_{\tilde{\mathcal{A}}}\right)+D \,\text{Im}\left(P_{\tilde{\mathcal{B}}}\right)\right]^{-1} \label{mA}\\
B^{t}&=-\text{Im}\left(P_{\tilde{\mathcal{A}}}\right)\left[C 
\,\text{Im}\left(P_{\tilde{\mathcal{A}}}\right)+D \,\text{Im}\left(P_{\tilde{\mathcal{B}}}\right)\right]^{-1} \label{mB}\\
C^{t}&=\frac{1}{2}\,A^{t}\mathbb{H}-\text{Re}\left(P_{\tilde{\mathcal{B}}}\right)\left[A \,\text{Re}\left(P_{\tilde{\mathcal{A}}}\right)+B\, \text{Re}\left(P_{\tilde{\mathcal{B}}}\right)\right]^{-1}  \label{mC}\\
D^{t}&=\frac{1}{2}\,B^{t}\mathbb{H}+\text{Re}\left(P_{\tilde{\mathcal{A}}}\right)\left[A \,\text{Re}\left(P_{\tilde{\mathcal{A}}}\right)+B\, \text{Re}\left(P_{\tilde{\mathcal{B}}}\right)\right]^{-1}.  \label{mD}
\end{align}  \label{matrix ABCD}
\end{lemma}

\begin{proof}
Recall that symplectic matrices $M=\left(\begin{matrix}
A&B\\
C&D
\end{matrix}\right)\in Sp(2g,\Z)$ are characterized by
\begin{align}
A^{t}D-&C^{t}B=\mathbb{I}_{g},  \label{Symp1}\\
A^{t}C&=C^{t}A,  \label{Symp2}\\
D^{t}B&=B^{t}D.\label{Symp3}
\end{align}
Multiplying  equality (\ref{A,B Im}) from the left by the matrix $C^{t}$, we deduce from (\ref{Symp1}) and (\ref{Symp2}) that:
\begin{align}
&C^{t}A\,\text{Im}\left(P_{\tilde{\mathcal{A}}}\right)+C^{t}B\,\text{Im}\left(P_{\tilde{\mathcal{B}}}\right)=0 \nonumber\\
&C^{t}A\,\text{Im}\left(P_{\tilde{\mathcal{A}}}\right)+(A^{t}D-\mathbb{I}_{g})\,\text{Im}\left(P_{\tilde{\mathcal{B}}}\right)=0\nonumber\\
& 
A^{t}C\,\text{Im}\left(P_{\tilde{\mathcal{A}}}\right)+A^{t}D\,\text{Im}\left(P_{\tilde{\mathcal{B}}}\right)=\text{Im}\left(P_{\tilde{\mathcal{B}}}\right),\nonumber
\end{align}
which leads to (\ref{mA}). Equality (\ref{mB}) can be checked 
analogously with (\ref{Symp1}) and (\ref{Symp3}).
To prove (\ref{mC}), multiply equality (\ref{C,D Re}) from  the 
left by the matrix $A^{t}$. Using (\ref{Symp1}) and (\ref{Symp2}) one gets:
\begin{align}
&A^{t}C \,\text{Re}\left(P_{\tilde{\mathcal{A}}}\right)+A^{t}D\, \text{Re}\left(P_{\tilde{\mathcal{B}}}\right)=\frac{1}{2}A^{t}\mathbb{H}P_{\mathcal{A}}  \nonumber\\
& C^{t}A \,\text{Re}\left(P_{\tilde{\mathcal{A}}}\right)+(\mathbb{I}_{g}+C^{t}B)\, \text{Re}\left(P_{\tilde{\mathcal{B}}}\right)=\frac{1}{2}A^{t}\mathbb{H}P_{\mathcal{A}} \nonumber\\
& C^{t}\left(A \,\text{Re}\left(P_{\tilde{\mathcal{A}}}\right)+B\, \text{Re}\left(P_{\tilde{\mathcal{B}}}\right)\right)=\frac{1}{2}A^{t}\mathbb{H}P_{\mathcal{A}}-\text{Re}\left(P_{\tilde{\mathcal{B}}}\right),\nonumber
\end{align}
which by (\ref{A,B Re}) leads to (\ref{mC}). Identity (\ref{mD}) can be proved analogously.
\end{proof}

\begin{remark}
\rm{Lemma \ref{matrix ABCD} implies that it is sufficient to know the 
matrices $A$ and $B$ (or $C$ and $D$) to determine the symplectic 
matrix in (\ref{transf per}). In practice, this means that a 
convenient ansatz for one of the matrices has to be found. The others 
then follow from the relations in Lemma \ref{matrix ABCD}.}
\end{remark}

Thus to construct these matrices one first checks which of the 
matrices $\mathrm{Re}\left(P_{\tilde{\mathcal{A}}}\right)$, 
$\mathrm{Re}\left(P_{\tilde{\mathcal{B}}}\right)$, $\mathrm{Im}\left(P_{\tilde{\mathcal{A}}}\right)$, 
$\mathrm{Im}\left(P_{\tilde{\mathcal{B}}}\right)$ are invertible. This way a matrix can be 
identified (e.g.~$A$) in terms of which the others  can be 
expressed. The task is thus reduced to provide an ansatz for this 
matrix such that the others will have entire components. We 
illustrate this approach at the example of the Trott curve below.

\begin{proposition}
Let $(\mathbf{\tilde{\mathcal{A}}},\mathbf{\tilde{\mathcal{B}}})$ 
be the canonical homology basis obtained with the Tretkoff-Tretkoff 
algorithm; we denote with a tilde the quantities expressed in this basis.
Under the change of homology basis (\ref{transf Vinn2}),  solutions 
of $n$-NLS$^{s}$ and DS equations given in (\ref{sol n-NLS}) and 
(\ref{psi DS}), respectively, which are expressed in  the  basis 
satisfying (\ref{hom basis}), transform as follows: the vector $\mathbf{d}$ appearing in the solutions becomes $(2\mathrm{i}\pi)^{-1}\,\tilde{\mathbb{K}}^{t}\,\mathbf{d}$ where $\tilde{\mathbb{K}}=2\mathrm{i}\pi A+B\,\tilde{\mathbb{B}}$, and the theta function $\Theta=\Theta_{\B}$ with zero characteristic, transforms to the theta function $\Theta_{\tilde{\B}}[\tilde{\delta}]$ with characteristic $\tilde{\delta}=[\tilde{\delta}_{1},\tilde{\delta}_{2}]$ given by
\begin{align}
\tilde{\delta}_{1}&= \frac{1}{4}\,\text{diag}\left(B^{t}\mathbb{H}B-2\,\text{Re}\left(P_{\tilde{\mathcal{A}}}\right)\tilde{\mathbb{M}}^{-1}\text{Im}\left(P_{\tilde{\mathcal{A}}}^{t}\right)\right),  \label{caract1}\\ 
\tilde{\delta}_{2}&=\frac{1}{4}\,\text{diag}\left(A^{t}\mathbb{H}A-2\,\text{Re}\left(P_{\tilde{\mathcal{B}}}\right)\tilde{\mathbb{M}}^{-1}\text{Im}\left(P_{\tilde{\mathcal{B}}}^{t}\right)\right),  \label{caract2}
\end{align}
where
\begin{equation}
\tilde{\mathbb{M}}=\text{Im}\left(P_{\tilde{\mathcal{B}}}^{t}\right)\text{Re}\left(P_{\tilde{\mathcal{A}}}\right)-\text{Im}\left(P_{\tilde{\mathcal{A}}}^{t}\right)\text{Re}\left(P_{\tilde{\mathcal{B}}}\right).  \label{M}\\
\end{equation}    
 Moreover, the real constant $h$ appearing in (\ref{phi DS}) and (\ref{N3 DS}) becomes $h+\tilde{h}$ where $\text{Im}(\tilde{h})$ is given by
\begin{equation}
\text{Im}(\tilde{h})=\frac{1}{2}\,\ln\left\{\left|\frac{\Theta_{\tilde{\mathbb{B}}}[\tilde{\delta}](\mathbf{\tilde{Z}}+\mathbf{\tilde{r}})}{\Theta_{\tilde{\mathbb{B}}}[\tilde{\delta}](\mathbf{\tilde{Z}}-\mathbf{\tilde{r}})}\right|\right\}-\text{Im}(\tilde{G}_{3}),
\end{equation}
with $\mathbf{\tilde{Z}}=\mathrm{i}(\mathbf{\tilde{W}}_{a}-\mathbf{\tilde{W}}_{b})$,
and the vectors $\mathbf{N},\mathbf{M}$ defined in (\ref{hom basis 3}) become $A^{t}\mathbf{N}+C^{t}\mathbf{M}$ and $B^{t}\mathbf{N}+D^{t}\mathbf{M}$ respectively.
             
\end{proposition}

\begin{proof}
Under the change of the canonical homology basis (\ref{transf Vinn2}), the vector $\omega=(\omega_{1},\ldots,\omega_{g})^{t}$ of normalized holomorphic  differentials transforms as
\begin{equation}
\omega=2\mathrm{i}\pi\, (\tilde{\mathbb{K}}^{t})^{-1}\,\tilde{\omega},  \label{transf hol diff}
\end{equation}
where $\tilde{\mathbb{K}}=2\mathrm{i}\pi A+B\,\tilde{\mathbb{B}}$.  According to the transformation law (\ref{transf theta}) of theta functions, it can be checked after straightforward calculations, that under this  change of homology basis, quantities (\ref{q2})-(\ref{K2}) transform as:
\begin{align}
q_{2}(a,b)&=\tilde{q}_{2}(a,b)\,\exp\left\{-\mathbf{\tilde{r}}^{t}\,(\tilde{\mathbb{K}}^{t})^{-1}B\,\tilde{\mathbf{r}}\right\},  \label{transf q2}
\\
q_{1}(a,b)&=\tilde{q}_{1}(a,b)+\frac{1}{2}\,\mathbf{\tilde{V}}_{a}^{t}\,(\tilde{\mathbb{K}}^{t})^{-1}B\,\mathbf{\tilde{V}}_{b}, \label{transf q1}
\end{align}
\begin{equation}
K_{1}(a,b)=\tilde{K}_{1}(a,b)+\frac{1}{2}\,\left(\mathbf{\tilde{V}}_{a}^{t}\,(\tilde{\mathbb{K}}^{t})^{-1}B\,\mathbf{\tilde{r}}+\mathbf{\tilde{r}}^{t}\,(\tilde{\mathbb{K}}^{t})^{-1}B\,\mathbf{\tilde{V}}_{a}\right), \label{transf K1}
\end{equation}
\begin{equation}
K_{2}(a,b)=\tilde{K}_{2}(a,b)-\frac{1}{2}\,\left(\mathbf{\tilde{W}}_{a}^{t}\,(\tilde{\mathbb{K}}^{t})^{-1}B\,\mathbf{\tilde{r}}+\mathbf{\tilde{r}}^{t}\,(\tilde{\mathbb{K}}^{t})^{-1}B\,\mathbf{\tilde{W}}_{a}\right)-\mathbf{\tilde{V}}_{a}^{t}\,(\tilde{\mathbb{K}}^{t})^{-1}B\,\mathbf{\tilde{V}}_{a}. \label{transf K2}
\end{equation}
We deduce that solutions of the  $n$-NLS$^{s}$ and DS equations given in (\ref{sol n-NLS}) and 
(\ref{psi DS}), respectively, transform as follows: the vector $\mathbf{d}$ becomes $(2\mathrm{i}\pi)^{-1}\,\tilde{\mathbb{K}}^{t}\,\mathbf{d}$, and the theta function $\Theta=\Theta_{\B}$ with zero characteristic, transforms to the theta function $\Theta_{\tilde{\B}}[\tilde{\delta}]$ with characteristic $\tilde{\delta}$. To compute the vectors of the characteristic $\tilde{\delta}$ we consider the
inversion of the symplectic matrix in (\ref{transf Vinn2}) which leads to
\begin{equation}
\left(\begin{matrix}
\mathbf{\tilde{\mathcal{A}}}\\
\mathbf{\tilde{\mathcal{B}}}
\end{matrix}\right) 
=
\left(\begin{array}{rrr}
D^{t}&-B^{t}\\
-C^{t}&A^{t}
\end{array}\right)\left(\begin{matrix}
\mathbf{\mathcal{A}}\\
\mathbf{\mathcal{B}}
\end{matrix}\right). \nonumber
\end{equation}
Since the characteristic used in \cite{Kalla} to construct solutions (\ref{sol n-NLS}), (\ref{psi DS}) of $n$-NLS$^{s}$ and DS is zero, we get with  (\ref{transf caract})
\begin{align}
\left(\begin{matrix}
\tilde{\delta}_{1}\\
\tilde{\delta}_{2}
\end{matrix}\right)
&=
\frac{1}{2}\,\text{Diag}\left(\begin{matrix}
D^{t}B\\
C^{t}A
\end{matrix}\right) \nonumber
\end{align}
(note that $D^{t}B$ and $C^{t}A$ are symmetric matrices).
Substituting (\ref{mA}) and (\ref{mB}) in (\ref{mC}) (resp. (\ref{mD})), it can be checked that
\begin{align}
C^{t}A&=\frac{1}{2}\,\left(A^{t}\mathbb{H}A-2\,\text{Re}\left(P_{\tilde{\mathcal{B}}}\right)\tilde{\mathbb{M}}^{-1}\,\text{Im}\left(P_{\tilde{\mathcal{B}}}^{t}\right)\right),\nonumber\\
D^{t}B&=\frac{1}{2}\,\left(B^{t}\mathbb{H}B-2\,\text{Re}\left(P_{\tilde{\mathcal{A}}}\right)\tilde{\mathbb{M}}^{-1}\,\text{Im}\left(P_{\tilde{\mathcal{A}}}^{t}\right)\right),\nonumber
\end{align}
with 
\[\tilde{\mathbb{M}}=\text{Im}\left(P_{\tilde{\mathcal{B}}}^{t}\right)\text{Re}\left(P_{\tilde{\mathcal{A}}}\right)-\text{Im}\left(P_{\tilde{\mathcal{A}}}^{t}\right)\text{Re}\left(P_{\tilde{\mathcal{B}}}\right).\]

Moreover, the real constant $h$ appearing in the solutions (\ref{psi DS})-(\ref{phi DS}) of the Davey-Stewartson equations becomes $h+\tilde{h}$, where $\tilde{h}$ is given by
\begin{equation}
\tilde{h}=-\tilde{\mathbf{V}}_{a}^{t}\,(\tilde{\mathbb{K}}^{t})^{-1}B\,\tilde{\mathbf{V}}_{a}-\tilde{\mathbf{V}}_{b}^{t}\,(\tilde{\mathbb{K}}^{t})^{-1}B\,\tilde{\mathbf{V}}_{b}. \label{tilde h}
\end{equation}
Notice that the construction of the solutions (\ref{psi DS}) given in \cite{Kalla} allows to express the imaginary part of the constant $\tilde{h}$ (\ref{tilde h}) in terms of the characteristic $\tilde{\delta}$. Namely, since the reality condition
\begin{equation}
\psi^{*}=\rho\, \overline{\psi} \label{real cond}
\end{equation}
is satisfied for the Vinnikov basis, where the function $\psi^{*}(\xi,\eta,t)$ reads
\begin{equation}
\psi^{*}(\xi,\eta,t)=-\kappa_{1}\kappa_{2}\,\frac{q_{2}(a,b)}{A}\,\frac{\Theta(\mathbf{Z}-\mathbf{d}-\mathbf{r})}{\Theta(\mathbf{Z}-\mathbf{d})}\,\exp\left\{ \mathrm{i}\left(G_{1}\,\xi+G_{2}\,\eta-G_{3}\,\tfrac{t}{2}\right)\right\}, \label{psi* DS}
\end{equation}
it also holds for the computed basis. Therefore,
putting $\xi=\eta=0$, $t=2$, $\mathbf{d}=0$, and taking the modulus of each term in (\ref{real cond}) expressed in the computed basis, one gets:
\[\left|\frac{\Theta_{\tilde{\mathbb{B}}}[\tilde{\delta}](\mathbf{\tilde{Z}}-\mathbf{\tilde{r}})}{\Theta_{\tilde{\mathbb{B}}}[\tilde{\delta}](\mathbf{\tilde{Z}})}\exp\{-\mathrm{i}\,(\tilde{G}_{3}+\tilde{h})\}\right|=\left| \frac{\Theta_{\tilde{\mathbb{B}}}[\tilde{\delta}](\mathbf{\tilde{Z}}+\mathbf{\tilde{r}})}{\Theta_{\tilde{\mathbb{B}}}[\tilde{\delta}](\mathbf{\tilde{Z}})}\exp\{\mathrm{i}\,(\tilde{G}_{3}+\tilde{h})\}\right|\]
where $\mathbf{\tilde{Z}}=\mathrm{i}(\mathbf{\tilde{W}}_{a}-\mathbf{\tilde{W}}_{b})$. 
We deduce that
\begin{equation}
\text{Im}(\tilde{h})=\frac{1}{2}\,\ln\left\{\left|\frac{\Theta_{\tilde{\mathbb{B}}}[\tilde{\delta}](\mathbf{\tilde{Z}}+\mathbf{\tilde{r}})}{\Theta_{\tilde{\mathbb{B}}}[\tilde{\delta}](\mathbf{\tilde{Z}}-\mathbf{\tilde{r}})}\right|\right\}-\text{Im}(\tilde{G}_{3}).
\end{equation}
\end{proof}

\begin{remark}
In the case where the spectral curve is an M-curve, i.e. 
$\mathbb{H}=0$, the vectors of characteristic (\ref{caract1}) and 
(\ref{caract2}) do not depend explicitly on the symplectic matrix appearing in the change of homology basis and are uniquely defined by:
\begin{align}
\tilde{\delta}_{1}&= \frac{1}{2}\,\text{diag}\left(\text{Re}\left(P_{\tilde{\mathcal{A}}}\right)\left[\text{Im}\left(P_{\tilde{\mathcal{B}}}^{t}\right)\text{Re}\left(P_{\tilde{\mathcal{A}}}\right)-\text{Im}\left(P_{\tilde{\mathcal{A}}}^{t}\right)\text{Re}\left(P_{\tilde{\mathcal{B}}}\right)\right]^{-1}\text{Im}\left(P_{\tilde{\mathcal{A}}}^{t}\right)\right),  \label{caract1 bis}\\ 
\tilde{\delta}_{2}&=\frac{1}{2}\,\text{diag}\left(\text{Re}\left(P_{\tilde{\mathcal{B}}}\right)\left[\text{Im}\left(P_{\tilde{\mathcal{B}}}^{t}\right)\text{Re}\left(P_{\tilde{\mathcal{A}}}\right)-\text{Im}\left(P_{\tilde{\mathcal{A}}}^{t}\right)\text{Re}\left(P_{\tilde{\mathcal{B}}}\right)\right]^{-1}\text{Im}\left(P_{\tilde{\mathcal{B}}}^{t}\right)\right).  \label{caract2 bis}
\end{align}
\end{remark}

It would be possible to compute the theta-functional solutions in 
the  Vinnikov basis once the symplectic 
transformation between this basis and the 
basis determined by the code is known. 
However, since this symplectic transformation is not unique,  the found 
 Vinnikov basis leads  in general to a Riemann matrix for which the theta series 
converges only slowly, i.e., the value $N_{\theta}$ in 
(\ref{thetanum}) has to be chosen very large. To avoid this problem, we 
compute the theta function always in the typically more convenient 
Tretkoff-Tretkoff basis with the characteristic of the theta 
functions given by (\ref{caract1})-(\ref{M}).

\subsection{Trott curve}

The Trott curve \cite{Trott} given by the algebraic equation 
\begin{equation}
    144\,(x^{4}+y^{4})-225\,(x^{2}+y^{2})+350\,x^{2}y^{2}+81=0
    \label{trott}
\end{equation}
is an M-curve with respect to the anti-holomorphic involution $\tau$ defined by $\tau(x,y)=(\overline{x},\overline{y})$, and is of genus 3. Moreover, this curve has real branch points only (and 28 real bitangents, namely, tangents to the curve in two 
places). Our computed matrices of $\tilde{\mathcal{A}}$ and 
$\tilde{\mathcal{B}}$-periods read\footnote{For the ease of 
representation we only give 4 digits here, though at least 12 digits 
are known for these quantities.}
{\small\[P_{\tilde{\mathcal{A}}}=\left(\begin{array}{ccc}
\,0.0235 \mathrm{i}&0.0138 \mathrm{i}&\,0.0138 \mathrm{i}\\
0&0.0277 \mathrm{i}&0\\
-0.0315\,\,\,\,\,\,&0&0.0250\,
\end{array}\right),\]\\
\[P_{\tilde{\mathcal{B}}}=\left(\begin{array}{ccc}
-0.0315+0.0235 \mathrm{i}&0.0138 \mathrm{i}&-0.0250+0.0138 \mathrm{i}\\
0&-0.025+0.0277 \mathrm{i}&0.0250\\
-0.0235 \mathrm{i}\,\,&0.0138 \mathrm{i}&\,\,0.0138 \mathrm{i}
\end{array}\right).
\]}\\
The Trott curve being an M-curve, the vectors of the characteristic $\tilde{\delta}$ satisfy (\ref{caract1 bis}) and (\ref{caract2 bis}), which leads to $\tilde{\delta}=
\frac{1}{2}[\begin{smallmatrix}
    0 & 0 & 0  \\
    1 & 1 & 0
\end{smallmatrix}]^{t}
$.

A possible choice of a symplectic transformation bringing the 
computed basis to  the Vinnikov basis is:\\
{\small\[A=\left(\begin{array}{rrr}
1&0&0\\
0&1&0\\
0&0&1
\end{array}\right), \quad 
B=\left(\begin{array}{rrr}
-1&0&0\\
0&-1&0\\
0&0&0
\end{array}\right), \quad
C=\left(\begin{array}{rrr}
1&0&0\\
0&1&0\\
0&0&0
\end{array}\right), \quad 
D=\left(\begin{array}{rrr}
0&0&0\\
0&0&0\\
0&0&1
\end{array}\right).\]}\\
Note that the matrices $A,B,C,D$ are not unique since the action (\ref{hom basis}) of the anti-holomorphic involution on the basic cycles allows for permutations of $\mathcal{A}_{j}$-cycles for instance. 
These matrices can be computed as follows. Since the Trott curve is an M-curve, one has $\mathbb{H}=0$. Moreover, the matrix Im$(P_{\tilde{\mathcal{B}}})$ being invertible here, by (\ref{A,B Im}) one gets:
\begin{equation}
B=-A \,\text{Im}\left(P_{\tilde{\mathcal{A}}}\right) \left(\text{Im}\left(P_{\tilde{\mathcal{B}}}\right)\right)^{-1}.  \label{Trott B}
\end{equation}
With (\ref{Symp1}) and (\ref{Symp2}) it follows that
\begin{equation}
A^{t}\left(D+C\,\text{Im}\left(P_{\tilde{\mathcal{A}}}\right) \left(\text{Im}\left(P_{\tilde{\mathcal{B}}}\right)\right)^{-1}\right)=\mathbb{I}_{3}.  \label{Trott symp}
\end{equation}
The computed matrix $\text{Im}\left(P_{\tilde{\mathcal{A}}}\right) 
\left(\text{Im}\left(P_{\tilde{\mathcal{B}}}\right)\right)^{-1}$ 
being (within numerical precision) equal to
\[\text{Im}\left(P_{\tilde{\mathcal{A}}}\right) \left(\text{Im}\left(P_{\tilde{\mathcal{B}}}\right)\right)^{-1}=\left(\begin{array}{rrr}
-1&0&0\\
0&-1&0\\
0&0&0
\end{array}\right),\]
and with $C,D\in\M_{3}(\Z)$, 
we get from (\ref{Trott symp}) that $\det A=1$. Since 
$A\in\M_{3}(\Z)$, the condition $\det A=1$ implies $A\in Gl_{3}(\Z)$.
For any $A\in Gl_{3}(\Z)$, one can see from (\ref{Trott B}), 
(\ref{mC}) and (\ref{mD}) that $B,C,D \in \M_{3}(\Z)$, and therefore 
that the matrices $A,B,C,D$ give a solution of (\ref{A,B 
Re})-(\ref{C,D Im}). The choice $A=\mathbb{I}_{3}$ leads to the above matrices.

The Trott curve has real fibers and can thus be used to construct 
solutions to the 3-NLS equation via the projection map 
$f:(x,y)\mapsto x$, which is a real meromorphic function of degree $4$ 
on the curve. We consider the points on the  curve  stable with 
respect to $\tau$ and 
projecting to the point with $x=0.1$ in the $x$-sphere, and choose $\mathbf{d}=0$. The corresponding solution to the 
3-NLS equation can be seen in Fig.~\ref{figtrottnls}.
\begin{figure}[htb!]
\begin{center}
  \includegraphics[width=0.7\textwidth]{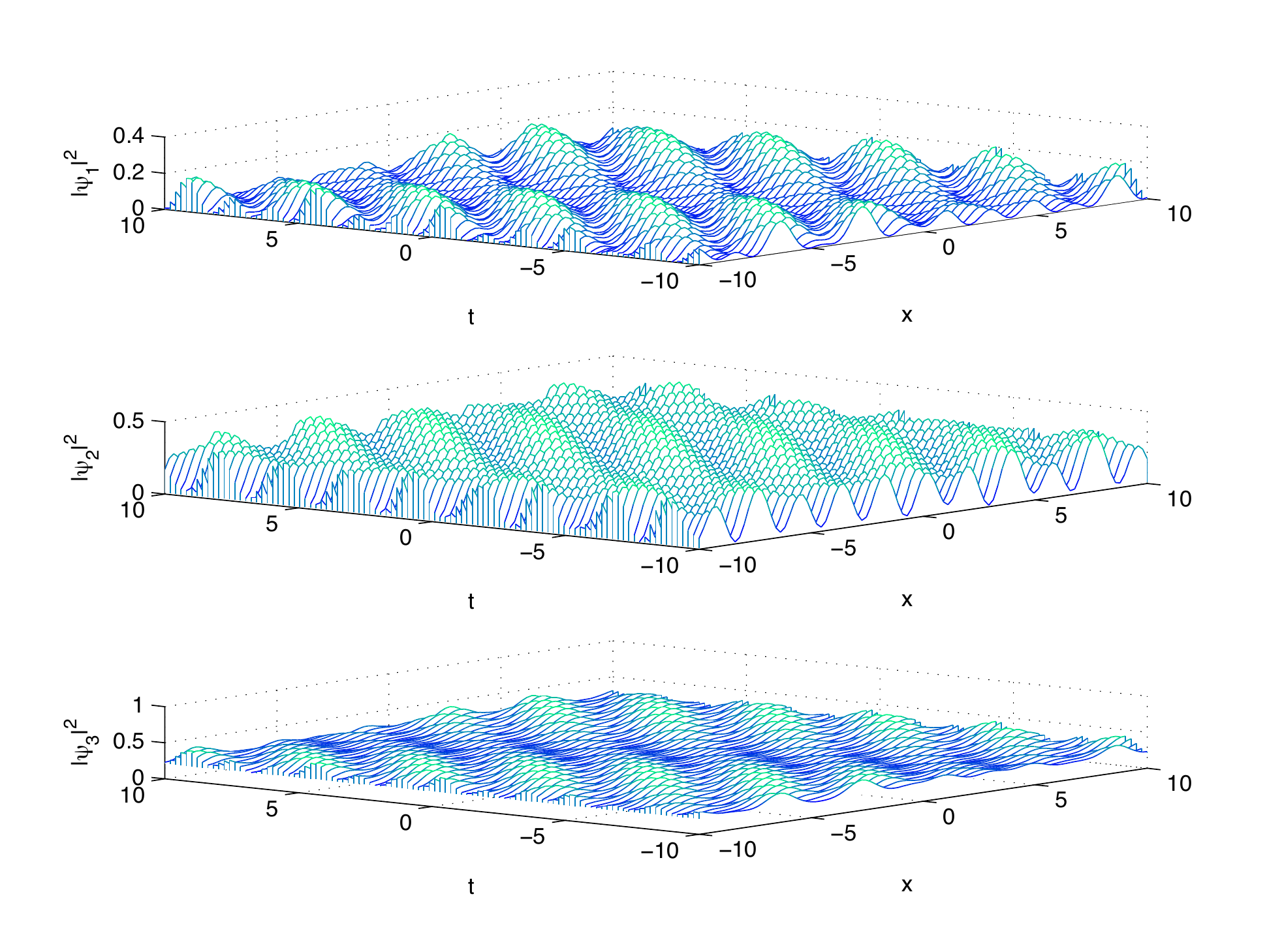}
\end{center}
 \caption{\textit{Solution (\ref{sol n-NLS}) to the 3-NLS$^{s}$  equation  on the Trott curve 
 for the points with $x=0.1$ on the $x$-sphere. The sheets are identified 
 at the points projecting to  $x=-1.0129, (0.9582\mathrm{i},- 0.9582\mathrm{i},0.1146\mathrm{i},- 
 0.1146\mathrm{i})$. The vector of signs equals $s=(1,-1,-1)$ 
 from top to bottom.}}
   \label{figtrottnls}
\end{figure}

A solution to the DS1$^{+}$ equation on this  curve can be 
constructed for points $a$ and $b$ stable with respect to the 
involution $\tau$. The solution for $a=(-0.2)^{(1)}$, $b=(0.2)^{(2)}$ and the choice $\mathbf{d}=0$ can be 
seen in Fig.~\ref{figtrottds1p}. Note that in accordance with Remark 
2.1,  one would obtain a solution of DS1$^{-}$ for the choice  $a=(-0.2)^{(1)}$ and 
$b=(0.2)^{(1)}$.

Similarly, a solution to the DS2$^{+}$ equation can be obtained for 
points $a$ and $b$ subject to $\tau a=b$. For $a=(0.1+\mathrm{i})^{(1)}$ 
and $b=(0.1-\mathrm{i})^{(1)}$ we 
get Fig.~\ref{figtrottds2p}.

\begin{figure}[htb!]
\begin{center}
\includegraphics[width=0.7\textwidth]{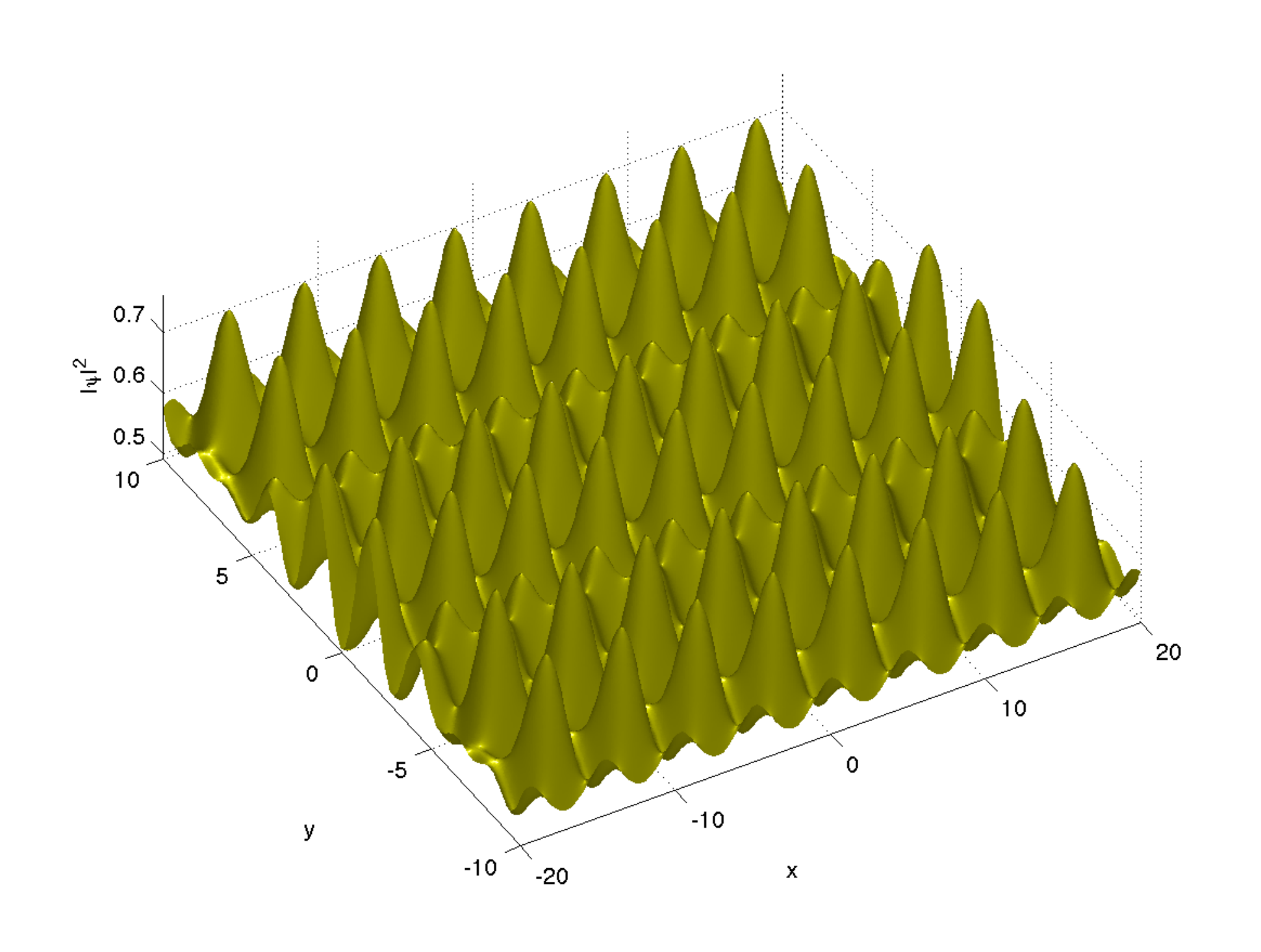}
\end{center}
 \caption{\textit{Solution to the DS1$^{+}$ equation  on the Trott curve 
 for the points $a=(-0.2)^{(1)}$ and $b=(0.2)^{(2)}$ at $t=0$.}}
 \label{figtrottds1p}
\end{figure}

\begin{figure}[htb!]
\begin{center}
\includegraphics[width=0.7\textwidth]{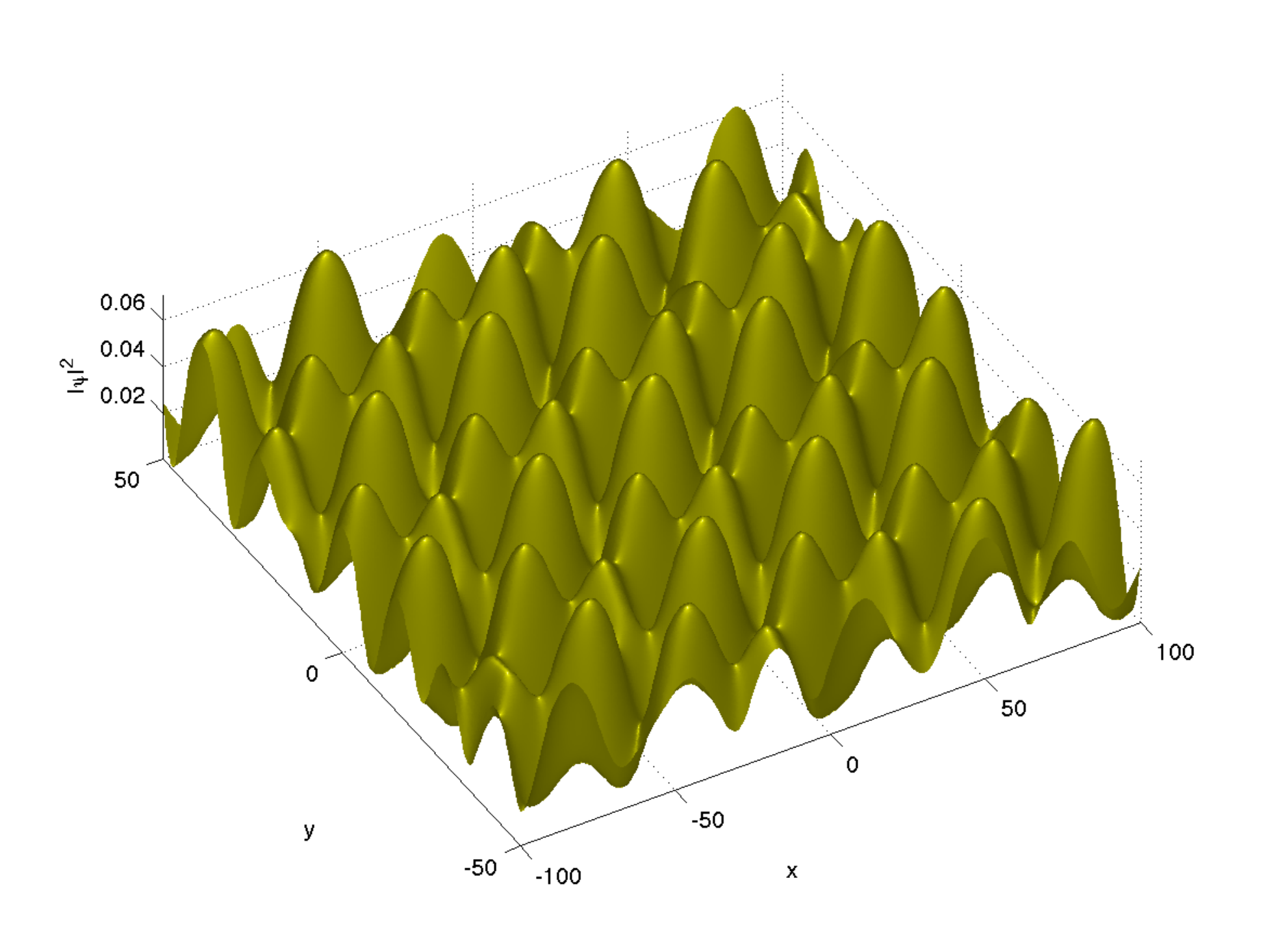}
\end{center}
 \caption{\textit{Solution to the DS2$^{+}$ equation  on the Trott curve 
 for the points $a=(0.1+\mathrm{i})^{(1)}$ and $b=(0.1-\mathrm{i})^{(1)}$ at $t=0$.}}
 \label{figtrottds2p}
\end{figure}


\subsection{Dividing curves without real branch point}

We consider

the curve given by the equation 
\begin{equation}
    30x^4-61x^{3}y+41y^2x^2-43x^2-11y^3x+42xy+y^4-11y^2+9=0
    \label{divg3}
\end{equation}
which was  studied in \cite{Dubm} and \cite{Vin}. It is a genus 3 curve, dividing with respect to the anti-holomorphic 
involution $\tau$, without real branch point. This curve 
admits two real ovals. In this case the matrix $\mathbb{H}$ has the 
form  
{\small\[\mathbb{H}=
\left(\begin{matrix}
    0 & 1 & 0  \\
    1 & 0 & 0   \\
    0 & 0 & 0
\end{matrix}\right).
\]}
The  period matrices computed by the code read
{\small\[P_{\tilde{\mathcal{A}}}=\left(\begin{array}{rrr}
-0.2721-0.0977\mathrm{i}&-0.3193+0.1914\mathrm{i}&-1.0668+0.4293\mathrm{i}\\
0.2721+0.0977\mathrm{i}&-0.3193-0.3341\mathrm{i}&-1.0668-0.4316\mathrm{i}\\
0.2721-0.0977\mathrm{i}&0.4676-0.3341\mathrm{i}&0.7992-0.4316\mathrm{i}
\end{array}\right),\]\\
\[P_{\tilde{\mathcal{B}}}=\left(\begin{array}{rrr}
-0.2721-0.2932\mathrm{i}&-0.3193+0.3341\mathrm{i}&-1.0668+0.4316\mathrm{i}\\
0.2721+0.2932\mathrm{i}&-0.3193-0.7169\mathrm{i}&-1.0668-1.2903\mathrm{i}\\
0.2721-0.0977\mathrm{i}&0.4676+0.1914\mathrm{i}&0.7992+0.4293\mathrm{i}
\end{array}\right).
\]}\\
After some calculations, one finds that the following matrices 
$A,B,C,D$ provide a solution of (\ref{A,B Re})-(\ref{C,D Im}):\\
{\small\[A=\left(\begin{array}{rrr}
-1&2&-1\\
2&-1&0\\
0&2&-1
\end{array}\right), \,
B=\left(\begin{array}{rrr}
1&0&1\\
0&1&0\\
1&0&0
\end{array}\right), \,
C=\left(\begin{array}{rrr}
1&-1&-1\\
-1&1&-1\\
0&0&1
\end{array}\right), \,
D=\left(\begin{array}{rrr}
0&1&1\\
1&0&1\\
0&0&-1
\end{array}\right).\]}\\
From (\ref{caract1}) and (\ref{caract2}) one gets for the characteristic:  $\tilde{\delta}=
\frac{1}{2}[\begin{smallmatrix}
    0 & 0 & 1  \\
    1 & 1 & 0
\end{smallmatrix}]^{t}
$.

The curve (\ref{divg3}) has real fibers and can thus be used to construct 
solutions to the focusing 3-NLS equation. We consider the  points on the  curve 
with $x=2.5$ and stable with respect to $\tau$, and we choose $\mathbf{d}=0$. The corresponding solution to the 
focusing 3-NLS equation can be seen in Fig.~\ref{figdiv3nls}.

A solution to the DS1$^{-}$ equation can be constructed by choosing 
the points $a=(-4)^{(1)}$ and $b=(-3)^{(2)}$  see 
Fig.~\ref{figdiv3ds1m}.
\begin{figure}[htb!]
\begin{center}
  \includegraphics[width=0.7\textwidth]{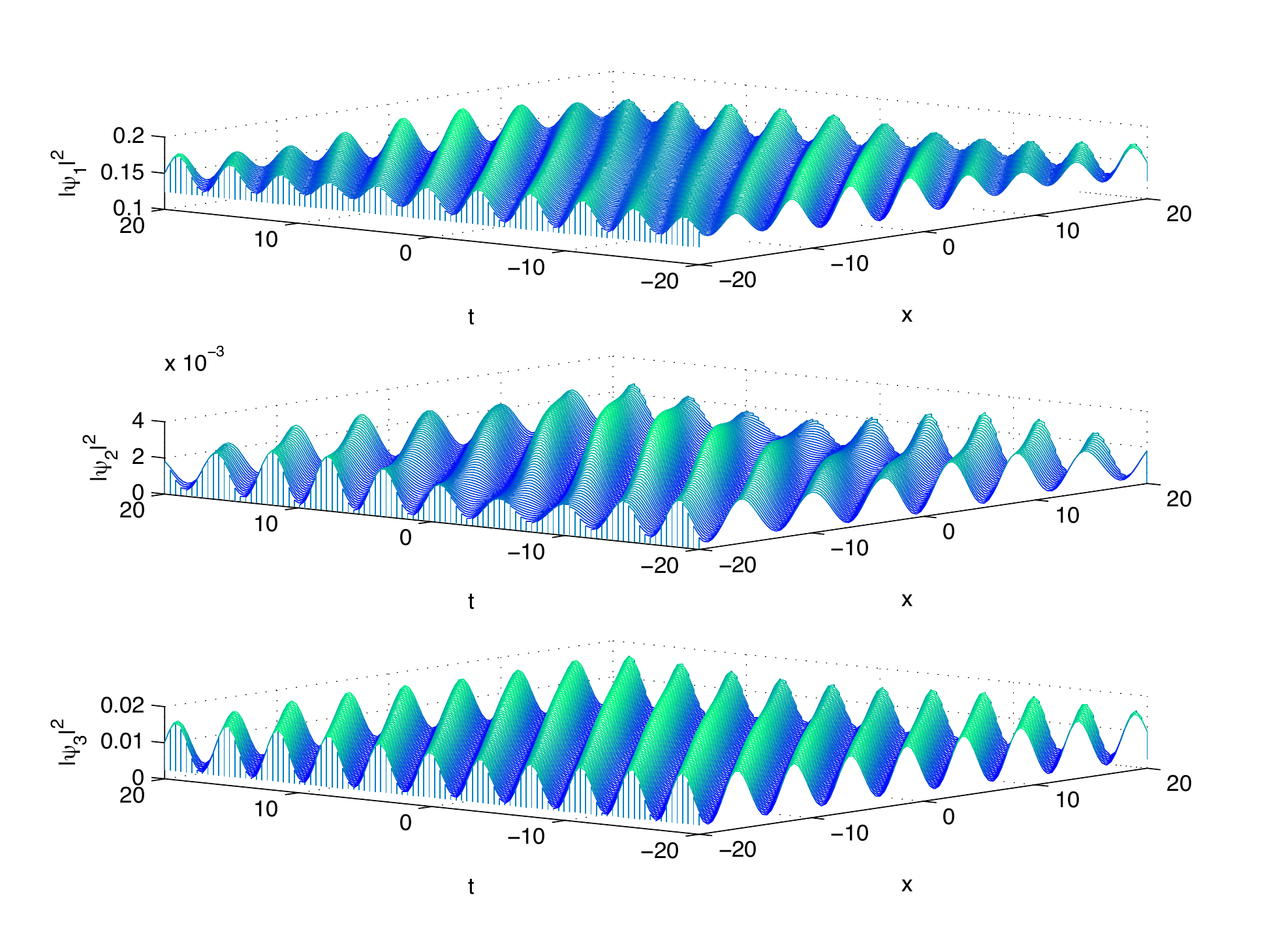}
\end{center}
 \caption{\textit{Solution to the 3-NLS$^{s}$  equation  on the dividing curve 
 (\ref{divg3}) of genus 3 
 for the points with $x=2.5$ on the $x$-sphere. The sheets are identified 
 at the fiber over $-2.1404 + 0.4404\mathrm{i}, 
 (-12.2492 + 2.0113\mathrm{i},  -5.1634 + 1.3519\mathrm{i},  -4.5915 + 0.9380\mathrm{i},  
 -1.5405 + 0.5429\mathrm{i})$. The vector of signs is $s=(1,1,1)$.}}
   \label{figdiv3nls}
\end{figure}

\begin{figure}[htb!]
\begin{center}
\includegraphics[width=0.7\textwidth]{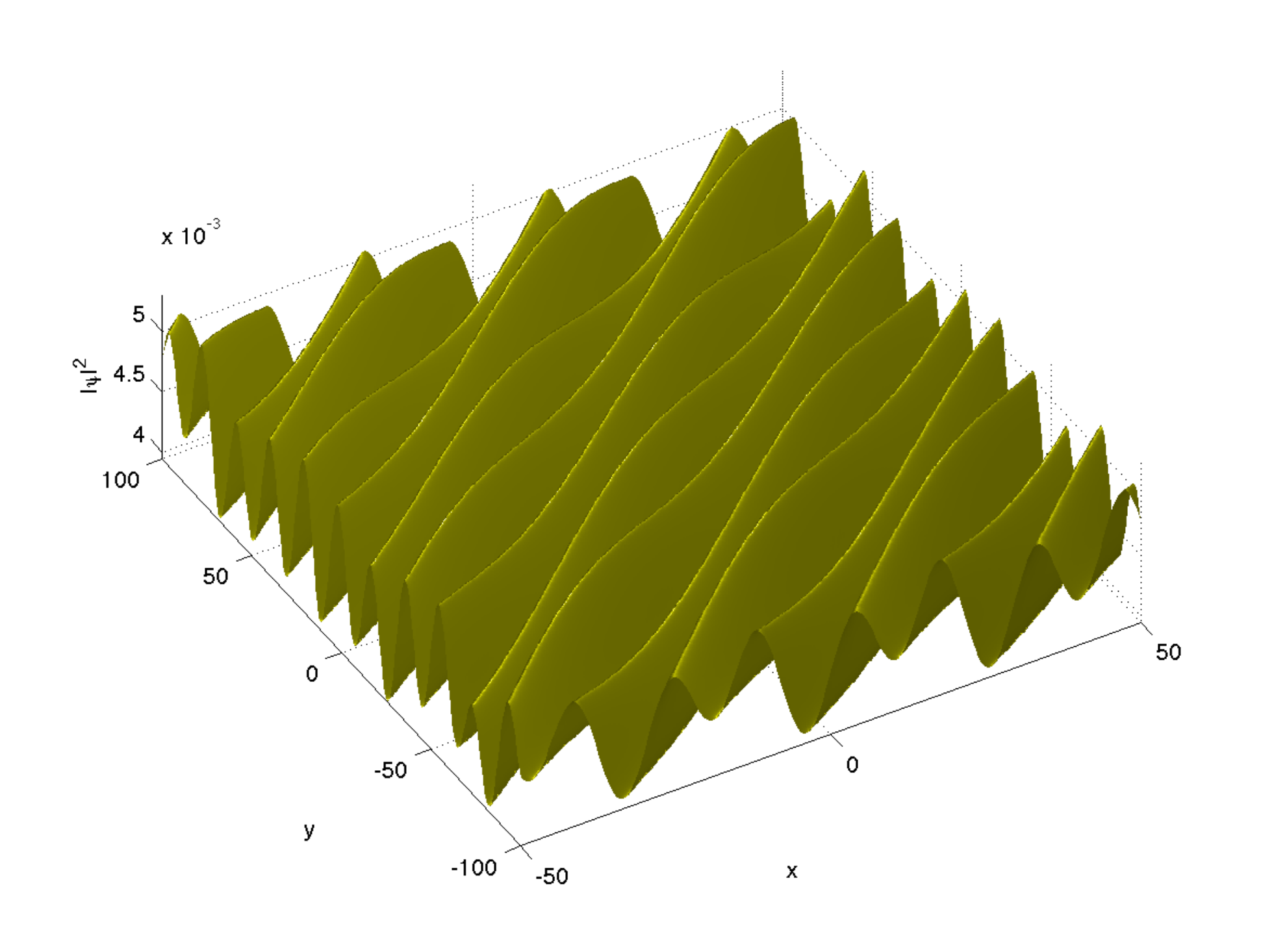}
\end{center}
 \caption{\textit{Solution to the DS1$^{-}$ equation  on the dividing curve  (\ref{divg3}) of genus 3 
 for the points $a=(-4)^{(1)}$ and $b=(-3)^{(2)}$ at $t=0$.}}
 \label{figdiv3ds1m}
\end{figure}

\subsection{Fermat curve}
The Fermat curves 
\begin{equation}
    y^{n}+x^{n}+1=0, \qquad n>2, \quad \text{$n$ even},
    \label{fermat}
\end{equation}
are real curves without real oval with respect to $\tau$. We consider here the curve with 
$n=4$ that has genus 3. The matrix $\mathbb{H}$ has the 
form  
{\small\[\mathbb{H}=
\left(\begin{matrix}
    0 & 1 & 0  \\
    1 & 0 & 0   \\
    0 & 0 & 0
\end{matrix}\right),
\]}
and we find 
$$   P_{\tilde{\mathcal{A}}}=\left(\begin{array}{ccc}
    0.9270\, &  \,- 0.9270\mathrm{i} & -0.9270\mathrm{i}\\
	     0&  0&            - 1.8541\mathrm{i}  \\
	     0.9270\mathrm{i} &        -0.9270\,& - 0.9270\mathrm{i}
	\end{array}\right),$$

	$$
	P_{\tilde{\mathcal{B}}}=\left(\begin{array}{ccc}
	0.9270 + 0.9270\mathrm{i} &  0.9270 - 0.9270\mathrm{i} & 0\,\,\\
	0&  -0.9270 + 	0.9270\mathrm{i}&      0.9270 - 0.9270\mathrm{i}       \\
		    -0.9270 &     -0.9270\mathrm{i}   & - 0.9270\mathrm{i}
	       \end{array}\right).
$$
The following matrices 
$A,B,C,D$ provide a solution of (\ref{A,B Re})-(\ref{C,D Im}):\\
{\small\[A=\left(\begin{array}{rrr}
0&1&1\\
1&0&0\\
0&0&1
\end{array}\right), \,
B=\left(\begin{array}{rrr}
-1&-2&-1\\
0&0&-1\\
-1&-1&0
\end{array}\right), \,
C=\left(\begin{array}{rrr}
0&1&0\\
0&0&1\\
1&-1&0
\end{array}\right), \,
D=\left(\begin{array}{rrr}
0&0&-1\\
0&-1&0\\
0&0&1
\end{array}\right),\]}\\
which leads to the characteristic:  $\tilde{\delta}=
\frac{1}{2}[\begin{smallmatrix}
    0 & 0 & 1  \\
    0 & 1 & 0
\end{smallmatrix}]^{t}
$.

To construct a solution of the DS2$^{-}$ 
equation on the Fermat  curve, we choose the points 
$a=(-1.5+\mathrm{i})^{(1)}$ 
and $b=(-1.5-\mathrm{i})^{(3)}$. 
The resulting solution for the choice $\mathbf{d}=0$ can be seen in 
Fig.~\ref{figDS2mfermat}.

\begin{figure}[htb!]
\begin{center}
\includegraphics[width=0.7\textwidth]{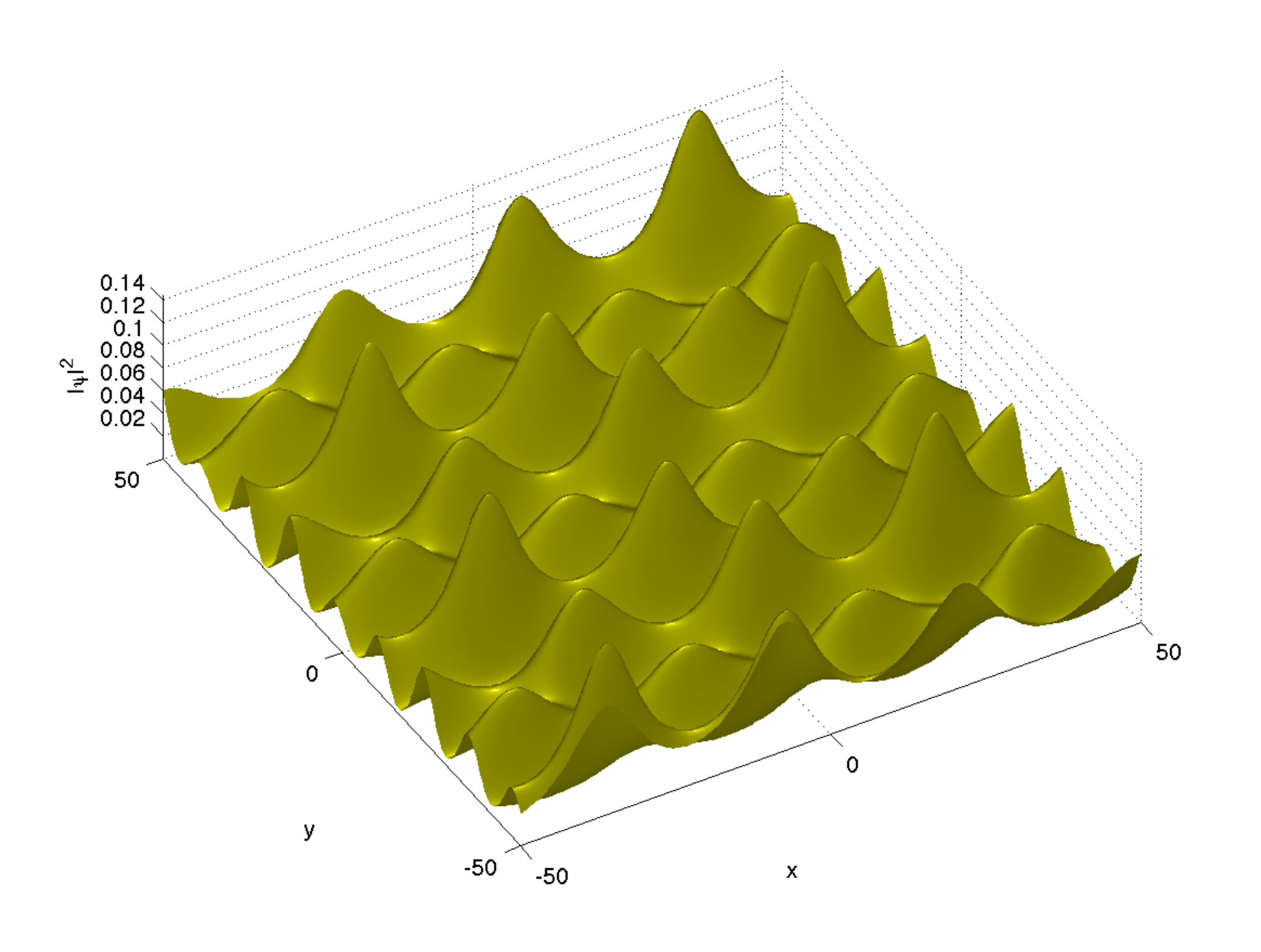}
\end{center}
 \caption{\textit{Solution to the DS2$^{-}$ equation  on the Fermat curve  
 (\ref{fermat}) of genus 3 
 for the points $a=(-1.5+\mathrm{i})^{(1)}$ 
and $b=(-1.5-\mathrm{i})^{(3)}$ at $t=0$.}}
 \label{figDS2mfermat}
\end{figure}

\section{Conclusion}
In this paper we have presented the state of the art of the numerical 
evaluation of solutions to integrable equations in terms of 
multi-dimensional theta functions associated to  real Riemann surfaces by 
using an approach via  real algebraic curves. It 
was shown that real hyperelliptic  curves parametrized by 
the branch points can be treated with machine precision for a wide 
range of the parameters. Even almost degenerate situations where the 
branch points coincide pairwise can be handled as long as at least 
one cut stays finite. This approach to real hyperelliptic curves 
\cite{FK1,FK2} is being generalized to arbitrary hyperelliptic 
curves. 

As discussed in \cite{FK}, the main difficulty  for general algebraic curves is 
the correct numerical identification of the branch points. The case of degenerations for given branch 
points has not yet been studied numerically, but is planned for the 
future. In what concerns the solutions (\ref{sol n-NLS}) to $n$-NLS$^{s}$ and 
similar solutions to the DS and the Kadomtsev-Petviashvili equations, the main problem  in the context of real Riemann surfaces is to find the symplectic 
transformation  leading to the homology basis introduced in \cite{Vin}, for which the solutions  of the studied equations, with regularity conditions,  can be conveniently formulated. This problem 
has been reduced to find a single $g\times g$-matrix for given periods and 
real ovals, the latter encoded by the matrix $\mathbb{H}$. For 
M-curves, where the matrix  $\mathbb{H}$ vanishes, a general formula 
for the characteristic (\ref{caract1})-(\ref{M}) could be 
given. In the general case, an algorithm along the lines indicated in 
the previous section to find the transformation 
will be based on a sufficiently general ansatz for one of the 
 matrices entering the symplectic transformation which is the subject of future work.

\end{document}